\pdfoutput=1 
\documentclass[11pt]{article}

\usepackage[IL2]{fontenc}
\usepackage[letterpaper,margin=1.00in]{geometry}
\usepackage{amsmath, amssymb, amsthm, amsfonts}
\usepackage{bbm}
\usepackage{amscd}
\usepackage{mathrsfs}

\usepackage{comment} 
\usepackage{ifthen}
\usepackage{tikz}
\usetikzlibrary{positioning,decorations.pathreplacing}
\usepackage{ bbold }
\usepackage{appendix}
\usepackage{graphicx}
\usepackage{color}
\usepackage{algorithm}
\usepackage[noend]{algpseudocode}
\usepackage{epstopdf}
\usepackage{wrapfig}
\usepackage{paralist}
\usepackage{wasysym}
\usepackage[disable]{todonotes}

\usepackage{pdflscape} 

\newcommand{\change}[2]{{\color{red} #1}{\color{blue} #2}}

\usepackage{framed}
\usepackage[framemethod=tikz]{mdframed}
\usepackage[bottom]{footmisc}
\usepackage{enumitem}
\setitemize{noitemsep,topsep=3pt,parsep=3pt,partopsep=3pt}
\usepackage[font=small]{caption}
\usepackage{xspace}

\usepackage{thmtools} 
\usepackage{thm-restate} 

\usepackage{hyperref}
\hypersetup{
    unicode=false,          
    colorlinks=true,        
    linkcolor=red,          
    citecolor=darkgreen,        
    filecolor=magenta,      
    urlcolor=cyan           
}

\newtheorem{theorem}{Theorem}[section]
\newtheorem{lemma}[theorem]{Lemma}
\newtheorem{meta-theorem}[theorem]{Meta-Theorem}
\newtheorem{claim}[theorem]{Claim}
\newtheorem{remark}[theorem]{Remark}

\newtheorem{proposition}[theorem]{Proposition}

\newtheorem{definition}[theorem]{Definition}

\newtheorem{question}[theorem]{Question}

\newtheorem{fact}[theorem]{Fact}

\usepackage[capitalize, nameinlink]{cleveref}
\crefname{theorem}{Theorem}{Theorems}
\crefname{proposition}{Proposition}{Propositions}
\crefname{observation}{Observation}{Observations}
\crefname{lemma}{Lemma}{Lemmas}
\crefname{claim}{Claim}{Claims}
\crefname{problem}{Problem}{Problems}
\crefname{conjecture}{Conjecture}{Conjectures}
\crefname{question}{Question}{Questions}
\crefname{example}{Example}{Examples}
\crefname{fact}{Fact}{Facts}

\definecolor{darkgreen}{rgb}{0,0.5,0}

\algnewcommand\algorithmicswitch{\textbf{switch}}
\algnewcommand\algorithmiccase{\textbf{case}}

\algdef{SE}[SWITCH]{Switch}{EndSwitch}[1]{\algorithmicswitch\ #1\ \algorithmicdo}{\algorithmicend\ \algorithmicswitch}%
\algdef{SE}[CASE]{Case}{EndCase}[1]{\algorithmiccase\ #1}{\algorithmicend\ \algorithmiccase}%
\algtext*{EndSwitch}%
\algtext*{EndCase}%

\newcommand{\fC}{\mathcal{C}}

\newcommand{\fI}{\mathcal{I}}

\newcommand{\fK}{\mathcal{K}}

\newcommand{\fR}{\mathcal{R}}
\newcommand{\fS}{\mathcal{S}}

\newcommand{\fU}{\mathcal{U}}

\newcommand{\eps}{\varepsilon}

\renewcommand{\P}{\textrm{P}}

\newcommand{\E}{\textrm{E}}
\newcommand{\e}{\textrm{e}}

\newcommand{\poly}{\operatorname{poly}}

\newcommand{\R}{\mathbb{R}}

\renewcommand{\paragraph}[1]{\vspace{0.15cm}\noindent {\bf #1}:}

\newcommand{\FullOrShort}{full}
\ifthenelse{\equal{\FullOrShort}{full}}{
  
  \newcommand{\fullOnly}[1]{#1}
  \newcommand{\shortOnly}[1]{}

  }{

    \newcommand{\fullOnly}[1]{}
    \newcommand{\IncludePictures}[1]{}
   
  }

\newcommand{\sigmabar}{\overline{\sigma_{i+1}}}
\newcommand{\avg}{\text{AVG}}
\newcommand{\hit}{\textrm{HIT}}

\renewcommand{\phi}{\varphi}
\newcommand{\argmin}{\textrm{argmin}}
\usepackage{listings}
\renewcommand{\tilde}{\widetilde}

\date{}
\title{A Nearly Tight Analysis of Greedy k-means++}
\author{
Christoph Grunau   \\
\texttt{cgrunau@inf.ethz.ch}\\
ETH Zurich \\
    \and
    Ahmet Alper Özüdoğru\\
    \texttt{oahmet@student.ethz.ch}\\
ETH Zurich \\
    \and
Václav Rozhoň \\
\texttt{rozhonv@inf.ethz.ch}\\
ETH Zurich \\
    \and
	Jakub Tětek\\
	\texttt{j.tetek@gmail.com}\\
	BARC, Univ. of Copenhagen
}

\begin{document}

\maketitle

\begin{abstract}
    The famous $k$-means++ algorithm of Arthur and Vassilvitskii [SODA 2007] is the most popular way of solving the $k$-means problem in practice. 
    The algorithm is very simple: it samples the first center uniformly at random and each of the following $k-1$ centers is then always sampled proportional to its squared distance to the closest center so far. 
    Afterward, Lloyd's iterative algorithm is run. 
    The $k$-means++ algorithm is known to return a $\Theta(\log k)$ approximate solution in expectation. 

    In their seminal work, Arthur and Vassilvitskii [SODA 2007] asked about the guarantees for its following \emph{greedy} variant: in every step, we sample $\ell$ candidate centers instead of one and then pick the one that minimizes the new cost. 
    This is also how $k$-means++ is implemented in e.g. the popular Scikit-learn library [Pedregosa et al.; JMLR 2011]. 
    
    
    We present nearly matching lower and upper bounds for the greedy $k$-means++: We prove that it is an $O(\ell^3 \log^3 k)$-approximation algorithm. 
    On the other hand, we prove a lower bound of $\Omega(\ell^3 \log^3 k / \log^2(\ell\log k))$. 
    Previously, only an $\Omega(\ell \log k)$ lower bound was known [Bhattacharya, Eube, Röglin, Schmidt; ESA 2020] and there was no known upper bound.

\end{abstract}

\thispagestyle{empty}
\newpage
\tableofcontents
 \thispagestyle{empty}
\newpage
\clearpage
\setcounter{page}{1}

\section{Introduction}
\label{sec:intro}

This paper is devoted to analyzing a natural and frequently-used greedy variant of the famous $k$-means++ clustering algorithm \cite{arthur2007k}. 
The difference between $k$-means++ and its greedy variant is very small: $k$-means++ samples one center in each step while greedy $k$-means++ samples $\ell$ candidate centers and then selects the one that decreases the current cost the most. 
While it is well known that the $k$-means++ algorithm is $\Theta(\log k)$-approximate, analyzing its greedy variant remained wide open. In this paper we show that greedy $k$-means++ is $O(\ell^3\log^3 k)$-approximate. 
Suprisingly, this is nearly tight. Specifically, we prove a lower bound of $\Omega(\ell^3 \log^3 k / \log^2(\ell\log k))$. 

\paragraph{Clustering}
Clustering is one of the most important tools in unsupervised machine learning. 
The task is to divide given input data into clusters of neighboring data points. 
There are many ways of formalizing that task, but one of the most popular ones is the $k$-means problem. 

In the $k$-means problem, we are given a set of points $X \subseteq \mathbb{R}^d$, as well as a parameter $k$. 
We are asked to find a set of $k$ \emph{centers} $K \subseteq \mathbb{R}^d$ that minimizes the sum of squared distances of points of $X$ to their respective closest centers. 
Namely, if we define the \emph{cost} of a point $x$ with respect to a set of centers $C$ as $\phi(x, C) := \min_{c \in C} ||x-c||^2\change{}{_2}$, we wish to find a set $C$ with $|C| = k$ that minimizes the expression $\phi(X, C) := \sum_{x \in X} \phi(x, C)$. 

The $k$-means problem is NP-hard \cite{aloise2009np,mahajan2009planar} and also NP-hard to approximate within some constant factor $c > 1$ \cite{awasthi2015hardness,lee2017improved}. 
On the other hand, there is a long line of work on approximation algorithms, with the current record holder being the algorithm of  \cite{cohenaddad2022best_apx} with approximation ratio of 5.912. Moreover, $1+\eps$-approximation can be reached for constant $d$ \cite{feldman2007ptas} and constant $k$ \cite{kumar2004simple}. 

However, these algorithms are not used in practice. Instead, practitioners rely on the so-called Lloyd's heuristic \cite{Lloyd} which can start with an arbitrary solution and iteratively makes its cost smaller, until convergence. 

Lloyd's heuristic is not ideal: it is prone to get stuck in bad local optima~\cite{arthur2007k}. 
In particular, it is not a constant approximation algorithm. 
A remedy to this problem is seeding it with a solution that is already good as the Lloyd's heuristic can then only make its cost smaller.  
In practice, such a seed can be simply a random subset of $X$ of size $k$. 
This natural option is for instance one of the options in the implementation used by Scikit learn \cite{scikit-learn} or R \cite{Rmanual}. 
Such an approach does not lead to any approximation ratio guarantees; its approximation ratio can be arbitrarily bad in some simple instances, e.g. whenever we have $k$ well-separated clusters lying in one line. 

\paragraph{$k$-means++}
A major result of Arthur and Vassilvitskii  \cite{arthur2007k,ostrovsky2013effectiveness} is a simple seeding algorithm known as $k$-means++ that both works well in practice and has desirable theoretical worst-case guarantees. 

The $k$-means++ algorithm works as follows. 
We sample the set $C \subseteq X$ sequentially, one center at a time. 
The first center we sample as a uniform point from $X$. Each next center is sampled proportional to its current cost. That is, if $C_i$ is the already constructed set of centers, we sample $x \in X$ as the next one with probability $\phi(x, C_i) / \phi(X, C_i)$.  
The pseudocode is in \cref{alg:kmpp}. 

\begin{algorithm}[ht]
	\caption{$k$-means++ seeding}
	\label{alg:kmpp}
	Input: $X$, $k$
	\begin{algorithmic}[1]
		\State Uniformly sample $x \in X$ and set $C_1 = \{ x \}$.
		\For{$i \leftarrow 1, 2, 3, \dots, k-1$}
		\State Sample $x \in X$ w.p. $\frac{\phi(x, C_{i})}{\phi(X, C_{i})}$ and set $C_{i+1} = C_{i} \cup \{x\}$.
		\EndFor
	\State \Return $C := C_k$
	\end{algorithmic}
\end{algorithm}

Arthur and Vassilvitskii proved that \cref{alg:kmpp} is $\Theta(\log k)$ approximate, in expectation\footnote{
In practice, the algorithm is rerun several times, which boosts the guarantee in expectation to a high probability guarantee. }. 
Although this approximation guarantee is not even constant, a benchmark achievable by many other known polynomial-time algorithms \cite{kanungo2004local, lattanzi2019better,ahmadian2019better},  the main point of the analysis is that we cannot construct an adversarial  instance where the Lloyd's heuristic seeded by $k$-means++ seeding can be arbitrarily bad, as is the case for e.g. the uniform random seeding. 
On practical data sets, the $k$-means++ seeding gives consistently better results than the random seeding \cite{arthur2007k} and is implemented in popular machine learning libraries like Scikit-learn \cite{scikit-learn}. 

However, it turns out that the algorithm implemented in the popular Scikit-learn library is \emph{not} the basic $k$-means++ (\cref{alg:kmpp}), but its \emph{greedy} variant described in \cref{alg:kmpp_greedy}. This algorithm in fact comes from the original paper of Arthur and Vassilvitskii \cite{arthur2007k} who mention that it gives better empirical results. They say that their analysis ``do[es] not carry over to this scenario'' and that ``it would be interesting to see a comparable (or better) asymptotic result''.

The greedy $k$-means++ algorithm works as follows. In every step, we sample $\ell$ \emph{candidate centers}  $c_{i+1}^1, \dots, c_{i+1}^\ell$ from the constructed distribution, not just one. 
Next, for each candidate center $c_{i+1}^j$ we compute the new cost of the solution $\phi(X, C \cup \{c_{i+1}^j\})$ if we add this candidate to our set of centers. Then we pick the candidate center that minimizes this expression. 

\begin{algorithm}[ht]
	\caption{Greedy $k$-means++ seeding}
	\label{alg:kmpp_greedy}
	Input: $X$, $k$, $\ell$
	\begin{algorithmic}[1]
		\State Uniformly independently sample $c^1_1, \dots, c^\ell_1 \in X$;
		\State Let $c_1 = \arg \min_{c \in \{c^1_1, \dots, c^\ell_1\} } \phi(X, c)$ and set $C_1 = \{ c_1 \}$.
		\For{$i \leftarrow 1, 2, 3, \dots, k-1$}
		\State Sample $c^1_{i+1}, \dots, c^\ell_{i+1} \in X$ independently, sampling $x$ with probability $\frac{\phi(x, C_i)}{\phi(X, C_i)}$;
		\State Let $c_{i+1} = \arg \min_{c \in \{c^1_i, \dots, c^\ell_i\} } \phi(X, C_{i} \cup \{c\})$ and set $C_{i+1} = C_{i} \cup \{c_{i+1}\}$.
		\EndFor
	\State \Return {$C := C_k$}
	\end{algorithmic}
\end{algorithm}




In \cite{bhattacharya2019noisy}, the authors show that \cref{alg:kmpp_greedy} is $\Omega(\ell \log k)$ approximate, in expectation. That is, the enhanced \cref{alg:kmpp_greedy} has worse theoretical guarantees than the original \cref{alg:kmpp}!
This should not, however, be so surprising. In fact, if we think about $\ell$ going to infinity, the algorithm becomes essentially deterministic; in every step, it will simply pick the point that decreases the cost the most. Such an algorithm heuristically makes sense, but lacks worst-case guarantees\footnote{Note that the negative cost $-\phi(X, C)$ is submodular in $C$, but because of the negative sign we cannot use the well-known fact that the greedy algorithm yields a $1-1/e$ approximation of the optimum. }. 
To see this, imagine two clusters with many points and a single lonely point in between: taking the lonely point results in a substantial drop of the cost, but it cannot be part of any solution with bounded approximation factor.

So, we cannot hope that \cref{alg:kmpp_greedy} gets better guarantees than \cref{alg:kmpp}. 
However, in the words of Arthur and Vassilvitskii, its guarantees could still be ``comparable''. 
This is the main result of this paper:

\begin{restatable}{theorem}{greedyUb}
\label{thm:greedy_ub}
Greedy $k$-means++ (\cref{alg:kmpp_greedy}) is an $O(\ell^3 \cdot \log^3 k)$-approximation algorithm, in expectation. 
\end{restatable}

On the other hand, we provide the following near-matching lower bound. 

\begin{restatable}{theorem}{greedyLB}
\label{thm:greedy_lb}
For every $k$ and $\ell \le k^{0.1}$, there exists a point set $X \subseteq \mathbb{R}^d$ for some $d \in \mathbb{N}$ where \cref{alg:kmpp_greedy} outputs $\Omega(\ell^3\log^3k / \log^2(\ell \log k))$ approximate solution with constant probability. 
\end{restatable}

We believe that the $\tilde{\Theta}(\ell^3 \log^3 k)$-approximation\footnote{Throughout this paper, we use $\tilde{O}(f(x))$ to denote $O(f (x)\poly \log (f(x)))$ and $\tilde{\Omega}, \tilde{\Theta}$ are defined analogously.}  bound 
 is a truly unexpected twist! However bizarre it may sound now, we hope to  give an adequate intuition behind it in \cref{sec:intuitive}.


\paragraph{Greedy rule is crucial}
The second result of this paper is that the greedy heuristic in \cref{alg:kmpp_greedy} is in fact crucial to getting polylogarithmic approximation. 
To this end, we generalize \cref{alg:kmpp_greedy} by allowing each center to be chosen from $\ell$ candidates by an arbitrary rule. 
The approximation ratio of this general algorithm becomes polynomial in $k$, unless we know the specifics of the rule. Concretely, we prove:

\begin{theorem}[Informal version of \Cref{thm:adversary_lb}]
\label{thm:adversary_lb_informal}
There exists a point set $X  \subseteq \mathbb{R}^d$ and a rule $\fR$ such that a variant of \cref{alg:kmpp_greedy} that uses $\fR$ instead of the greedy rule is $\Omega(k^{1 - 1/\ell})$-approximate with constant probability. 
\end{theorem}

This theorem also suggests that the greedy heuristic, among all others, makes a lot of sense!
On the other hand, we get the following upper bound.  
\begin{theorem}[Informal version of \Cref{thm:adversary_ub}]
\label{thm:adversary_ub_informal}
For any rule $\fR$, a variant of \cref{alg:kmpp_greedy} that uses $\fR$ instead of the greedy rule is $O(  k^{2-1/\ell} \cdot \ell\log k)$-approximate. 
\end{theorem}

We note that from the perspective of Algorithms with Predictions \cite{mitzenmacher_vassilvitskii2020algorithms_with_predictions}, \cref{thm:adversary_ub_informal} shows that whatever rule we use to generalize the greedy $k$-means++ algorithm, we still know that the algorithm remains somewhat comparable to the optimum solution. 

The gap between the lower bound and upper bound of \cref{thm:adversary_ub_informal,thm:adversary_lb_informal} can be tightened if one analyzes a certain natural stochastic process that can be understood without knowing anything about $k$-means(++); we defer the discussion of this interesting open problem to \cref{subsec:lprocess}. 


\paragraph{Related work}
In this section, we list some related work with algorithms derived from $k$-means++. To see more work about $k$-means in general, see for example the introduction of \cite{cohenaddad2022best_apx}. 

To the best of our knowledge, the only paper exploring \cref{alg:kmpp_greedy} from the theoretical perspective after the seminal paper \cite{arthur2007k} of Arthur and Vassilvitskii, is the paper of Bhattacharya, Eube, R\"{o}glin, and Schmidt \cite{bhattacharya2019noisy}. 
There, the authors show an $\Omega(\ell \log k)$ lower bound. They also consider a problem closely related to analysis of \cref{alg:kmpp_adversary}, we discuss it in greater detail in \cref{sec:l_adversary_upper_bound}. 

The empirical work related to \cref{alg:kmpp_greedy} starts with the PhD thesis of Vassilvitskii \cite{vassilvitskii_thesis} that reports experiments with $\ell = 2$, the comparative study of \cite{celebi2013comparative} on the other hand advertises the choice of $\Theta(\log k)$. This is also the choice taken in the Scikit-learn implementation \cite{scikit-learn} that chooses $\ell = \lfloor 2 + \ln k\rfloor$. 

Another related work to $k$-means++ include the following. Lattanzi and Sohler \cite{lattanzi2019better,choo2020kmeans} present a variant of $k$-means++ inspired by the local search algorithm of \cite{kanungo2004local}. 
A popular distributed variant of the $k$-means++ algorithm is the $k$-means$\|$ algorithm of Bahmani, Moseley, Vattani, Kumar, and Vassilvitskii \cite{bahmani2012scalable,bachem2017distributed,rozhon2020simpleAndTight,makarychev_reddy_shan2020improved} that achieves the same guarantees as $k$-means++ in $O(\log n)$ Map-Reduce rounds. 
Other lines of work study bicriteria guarantees of $k$-means++ \cite{aggarwal2009adaptive,wei2016constant,makarychev_reddy_shan2020improved}, analyze bad instances \cite{brunsch2013bad}, speed-up $k$-means++ by subsampling \cite{bachem2016approximate,bachem2016fast} or adapt it to the setting with outliers \cite{bhaskara2019kmeans++_with_outliers_and_penalties,grunau_rozhon2020adapting_kmeans_to_outliers}.

\paragraph{Roadmap}
In \cref{sec:intuitive}, we explain intuitively the proofs of \cref{thm:greedy_ub,thm:greedy_lb,thm:adversary_ub,thm:adversary_lb}. 
\cref{sec:preliminaries} collects some basic preliminary results that we need to use. 
In \cref{sec:greedy_lemma} we prove the main result necessary to prove \cref{thm:greedy_ub} which is then proved in \cref{sec:greedy_ub}. 
In \cref{sec:log3_lb} we then construct the almost matching lower bound. 
The analysis of \cref{alg:kmpp_adversary} is deferred to \cref{sec:l_adversary_upper_bound,sec:l_adversary_lower_bound}. 

\section{Intuitive explanations}
\label{sec:intuitive}

This section is devoted to an intuitive explanation of \cref{thm:greedy_ub,thm:greedy_lb,thm:adversary_ub,thm:adversary_lb}. 
We start by reviewing the analysis of the $k$-means++ algorithm by \cite{arthur2007k} in \cref{subsec:kmeans++}. 
Next, in \cref{subsec:strike} we identify the issues with generalizing the analysis of $k$-means++ to its greedy variant from \cref{alg:kmpp_greedy}. 
In \cref{subsec:hope}, we discuss a crucial lemma that essentially says that the greedy rule implies that not so many candidate centers are sampled in total from each optimal cluster. 
In \cref{subsec:return}, we show how this lemma implies \cref{thm:greedy_ub}. 
Finally, we present the lower bound in \cref{subsec:lb}.

\subsection{Analyzing $k$-means++}
\label{subsec:kmeans++}

We need to start by reviewing the analysis of the $k$-means++ algorithm by Arthur and Vassilvitskii. 
Reviewing it will allow us to explain which parts of the argument cannot be simply generalized in the analysis of \cref{alg:kmpp_greedy}. 

For $K \subseteq X$ we define $\phi^*(K) = \phi(K, \mu(K))$ where $\mu(K) = (\sum_{c \in K} c)/|K|$ is the center of mass of $K$. We note that $\phi^*(K)$ is the optimal $k$-means cost of $K$ achievable with one center. 
At the core of the $k$-means++ analysis lies the following lemma. 

\begin{lemma}[Informal version of \cref{lem:5apx}]
\label{lem:5apx_intuitive}
Let $X \subseteq \R^d$ be the input point set to the $k$-means problem and $C \subseteq \R^d$ a set of already selected centers. 
Let $K$ be an arbitrary subset of $X$. 
If we sample a point of $K$ such that a point $c \in K$ is sampled with probability $\phi(c, C) / \phi(K, C)$, we have $\E[ \phi(K, C \cup \{c\}) ] \le 5\phi^*(K)$. 
\end{lemma}

Intuitively, the reason why the lemma holds is that either all points $C$ are far away from $\mu(K)$ and then we essentially sample $c \in K$ from a uniform distribution; such a distribution would even lead to $\E[\phi(K, \{c\})] = 2\phi^*(K)$ by a simple averaging argument (cf. \cref{lem:2apx}). The other option is that some point of $C$ is already close to $\mu(K)$, but then $\phi(K, C)$ is already small. 

Whenever we use this lemma, we have $K$ to be a \emph{cluster} of some fixed optimal solution. Here, given a set of centers $C$, we define a cluster $K \subseteq X$ of a center $c \in C$ as the set of the points for which $c$ is the closest center, that is, $K = \{ x\in X : c = \argmin_{c' \in C} \phi(x, c')\}$. 
The usefulness of \cref{lem:5apx_intuitive} comes from its corollary that whenever \cref{alg:kmpp} samples the new center $c_i$ in the $i$-th step and it happens that $c_i \in K$, then the cost of $K$ becomes $5$-approximated, in expectation. This holds even though the algorithm itself has no idea about what the optimal clusters are. 

So, if it somehow happened that each of the $k$ sampled centers was from a different optimal cluster, we would get that \cref{alg:kmpp} is a $5$-approximation. 
Of course, there can be some clusters with no sampled center in the end because we happened to hit some other optimal cluster twice or more. 
This is the reason behind the final $O(\log k)$ approximation. 

Let us be more precise now: Given a set of already taken centers $C_{i}$, we say that a cluster $K$ from the optimal solution is \emph{covered} in the $i+1$th step if $K \cap C_{i} \not= \emptyset$. We accordingly split $X$ into points in covered and uncovered clusters, i.e., $X = X_{i}^\fC \sqcup X_{i}^{\fU}$.\footnote{The notation $A = B \sqcup C$ means $A = B \cup C$ and $B \cap C = \emptyset$.  } 
Finally, we say that a step $i+1$ is good whenever $c_{i+1}$ is sampled from an uncovered cluster and bad otherwise. 
It is the bad steps, where the algorithm really loses ground with respect to the optimal solution. The question is, by how much? 

Intuitively, if there are $u_i$ uncovered clusters remaining after $i$ steps and we had a bad $i+1$th step, in the end we will need to pay in our solution for one of the $u_i$ currently uncovered clusters. 
The cost of the largest currently uncovered cluster may be as large as $\phi(X_i^\fU, C_i)$. However, we should hope that  we only need to pay the cost of the average one, i.e., $\phi(X_{i}^\fU, C_{i})/u_i$. 
The reason for that is that we expect the size of the average uncovered cluster to only go down in the future, i.e., we have $\E\left[\phi(X_j^\fU, C_j)/u_j\right] \le \phi(X_i^\fU, C_i)/u_i$ for $j > i$. 
This is because if each new covered cluster were chosen from uncovered ones uniformly at random, the average cost of an uncovered cluster would stay the same. However, we are actually covering the more costly clusters with higher probability which makes the expected average cost of an uncovered cluster decrease in the future steps. 

So, having a bad step incurs a cost of $\phi(X_{i}^\fU, C_{i})/u_i$. This cost can be huge but, very conveniently, in that case, the probability of the $i+1$-th step being bad was very small to begin with. 
More concretely, the probability of having a bad step is equal to $\phi(X_{i}^\fC, C_{i}) / \phi(X, C_{i})$. Since the numerator of that expression is in expectation at most $5OPT$ by \cref{lem:5apx_intuitive}, we conclude that the expected cost incurred by the fact that we may sample from an already covered cluster in step $i+1$ is  at most
\[
\frac{5OPT}{\phi(X, C_{i})}
\cdot \frac{\phi(X_{i}^\fU, C_{i})}{u_i}
\le \frac{5OPT}{u_i}. 
\]

The variable $u_i$ starts at $u_0 = k$ and then decreases in each step by at most one until $u_{k-1} \ge 1$ in the $k$th step. 
Thus the contribution to the cost over all $k$ steps can be upper bounded by $O(OPT) \cdot (\frac{1}{k} + \frac{1}{k-1} + \dots + \frac{1}{1}) = O(OPT \cdot \log k)$. 

This concludes the intuitive analysis of \cref{alg:kmpp}. 
The formal proof can be found in \cite{arthur2007k} and a more detailed exposition can be found in lecture notes of Dasgupta \cite{Dasgupta}.   

\subsection{$k$-means++ strikes back}
\label{subsec:strike}

Here is the main problem with generalizing the $k$-means++ analysis to \cref{alg:kmpp_greedy}: We can sample many centers from the same cluster $K$ of the optimal solution, before the greedy heuristic decides to pick one! 
In other words: in $k$-means++ we can simply use \cref{lem:5apx_intuitive} to conclude that whenever we sample a point from an optimal cluster $K$, we expect its cost to become a $5$-approximation of the optimal cost. 
But what if we sample a candidate center from $K$ in every  step of the algorithm and the greedy heuristic chooses to actually pick the center only when it results in a bad approximation for $K$? 
In general, we can guarantee only expected $5\cdot (k\ell)$ approximation for the cluster $K$, when a center from it is picked, instead of $5$ approximation. 

In fact, let us now turn this idea into a lower bound. We will now show that with $\ell = \Omega(\log k)$ there is a rule $\fR$ such that \cref{alg:kmpp_adversary} has approximation factor of $\Omega(k)$, with constant probability. A simple generalization of the following argument will then yield \cref{thm:adversary_lb} in \cref{sec:l_adversary_lower_bound}. 

\begin{figure}[th]
    \centering
    \includegraphics{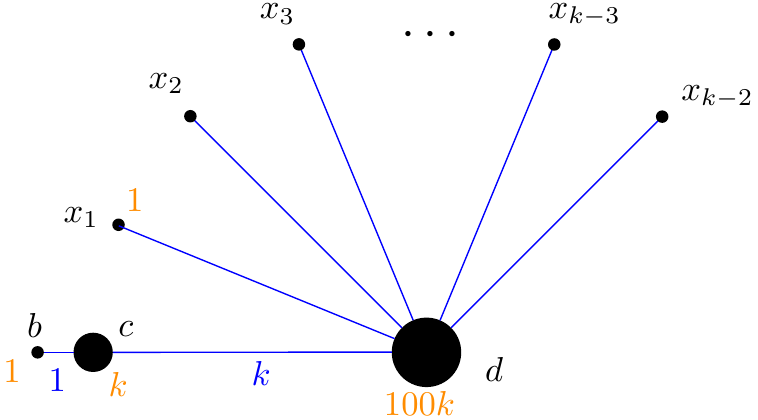}
    \caption{Illustration for the lower bound in the adversarial case. 
    Edge weights defining the metric are in blue while vertex weights are in orange and heavier points have a larger disk. 
    All nodes $x_i$ have weight $1$.\\
    Consider a rule that does not want to take the point $c$ as a center unless it has to, but it wants to take $b$ as a center whenever it can. 
    Such a rule takes $b$ a center with constant probability which results in $\Omega(k)$ approximation. 
    Notice that in this case $K = \{b,c\}$ is an optimal cluster such that  we sample a candidate center from it $\Omega(k)$ times but select each candidate as a center only when it results in a bad approximation of the cost of $K$. }
    \label{fig:simple_lb}
\end{figure}\todo{V: I would add orange ones for $b$ and $x_1$ }

The weighted point set $X$ that we use for the lower bound is in \cref{fig:simple_lb}. We can use integer weights, even though they are not part of the original problem formulation, by replacing an element with weight $w$ by $w$ copies with unit weight.
We will explain the lower bound in the setup where the points of $X$ are endowed with an arbitrary metric, not necessarily the Euclidean one. Generalizations to the Euclidean space are routine and left to the formal proofs. 

We start with a point $d$ in the center of the picture that has weight much larger than all the other points combined; thus with constant probability at least one candidate center is $d$ in the first step and our rule will then select $d$ as the first center. 

In a distance $k$ around $d$, there are $k-2$ dummy points $x_1, \dots, x_{k-2}$, each with weight one. Moreover, there is an additional pair of points $c$ and $b$, with $c$ having weight $k$ and $b$ having weight $1$. Their distance from each other is $1$ and the distance of $c$ from $d$ is $k$. The set $X = \{d, x_1, \dots, x_{k-2}, b, c\}$ is endowed with a tree metric generated by the  defined distance (see \cref{fig:simple_lb}). 

Let us compute the optimal cost of this instance:  
The optimum solution would take as centers all points except of $b$. Its cost is hence $w(b) \cdot d(b, c)^2 = 1$.

Now, let us choose the rule $\fR$ in \cref{alg:kmpp_adversary} as follows: whenever $d$ or $b$ is sampled as a candidate center, we select it. On the other hand, we do not select $c$ as a center unless we have to (which happens only when all candidate centers sampled are $c$). 

We can see that with this rule, with high probability, we select $d$ in the very first step and then in each of the following $k/2$ steps we have only negligible probability of adding $c$ to the set of centers: For that to happen, we would need all candidate centers to be $c$ which happens with probability at most $\left( \frac{w(c)}{w(c) + k/2 \cdot w(x_1)} \right)^\ell = (2/3)^\ell < 1/\poly(k)$ by our assumption $\ell = \Omega(\log k)$. 
On the other hand, the probability of sampling $b$ and thus adding it to the set of centers is $\Omega(1/k)$ in every step (unless $b$ or $c$ is selected as a center), hence it is constant after $k/2$ steps. 

In this case, the solution of the algorithm has to leave some other point from $\{x_1, \dots, x_{k-2}, c\}$ from the set of centers. If it leaves some $x_i$, it needs to pay $w(x_i) \cdot d(x_i, d)^2 = k^2$. If it leaves $c$, it needs to pay $w(c) \cdot d(b,c)^2 = k$. That is, the solution of \cref{alg:kmpp_adversarial} is  $k$-approximate or worse with constant probability. 

A more careful analysis in \cref{thm:adversary_lb} shows that for smaller $\ell$, we should choose the weight of $c$ to be $k^{1 - 1/\ell}$, to get a $\Omega(k^{1- 1/\ell})$ lower bound. 

\subsection{A new hope}
\label{subsec:hope}

Let us observe that greedy $k$-means++ would not be fooled by the example from the previous \cref{subsec:strike}. 
Once the point $c$ is sampled as a candidate center, it is also selected as a center by the greedy rule. 
This is because it results in a bigger drop in the cost than if we picked any of the dummy points $x_j$ (we assume $d$ is already picked). 

So, we may hope that in greedy $k$-means++ each optimal cluster cannot be hit by many candidate centers before it becomes covered. 
Our main technical contribution towards \cref{thm:greedy_ub} is indeed a proof that each optimal cluster $K$ is not hit by many candidate centers by greedy $k$-means++, before it becomes covered or before its cost is comparable with the optimal one. 

Recall that a cluster $K$ is covered in the $i$-th step if $K \cap C_i \not= \emptyset$. We additionally say that $K$ is solved if $\phi(K, C_i) \le 10^5 \phi^*(K)$. 
Finally, we define $\hit(K)$ to be the count of candidate centers $c_i^j \in K$ for all $1 \le i \le k$ and $1 \le j \le \ell$ where we count a hit only when $K$ is not covered or solved in the respective step. Our result is then the following. 

\begin{restatable}{lemma}{greedyLemma}
\label{lem:greedy_lemma}
For any optimal cluster $K$ we have $\E[\hit(K)] =  O(\ell^2 \log^2(k))$. 
\end{restatable}

Going back to our issue with generalizing the analysis of $k$-means++, we note that the above lemma implies that although we can no longer say that the cost of a cluster that just became covered is $5$ approximated in expectation, we can at least expect it to be $5 \cdot O(\ell^2 \log^2 k)$ approximated. This implies that the final approximation guarantee picks up an additional $O(\ell^2 \log^2 k)$ factor. This almost explains the final $O(\ell^3 \log^3 k)$ guarantee that we achieve in \cref{thm:greedy_ub} -- we pick up the remaining $\ell$ factor inside the main analysis, essentially because the probability of having a bad step increases by a factor of $\ell$. 
In the rest of this section, we sketch the proof of \cref{lem:greedy_lemma}. 

An important measure of progress in our analysis is the size of the \emph{neighborhood} of $K$. 
Given a point set $C_i$, we define $R_i = d(\mu(K), C_i)$ and define the neighborhood $N_i$ of a cluster $K$ in the $i+1$-th step as the set of points closer to $\mu(K)$ than $R_i$ (see \cref{fig:upper_bound}). 
Since we assume $K$ is not solved, i.e., that the current cost $\phi(K, C_i)$ of $K$ is at least $10^5\phi^*(K)$, we know that the distance $R_i$ is much larger than the distance of an average point of $K$ from $\mu(K)$. 
This means that most of the points of $K$ have to lie in $N_i$. For simplicity, we will next assume that $K \subseteq N_i$.  

We will also, for the sake of simplicity, assume that every point in $N_i$ has distance at most $R_i/10$ from $\mu(K)$. 
One reason this assumption makes our life easier is that every point in $N_i$ has cost between $(9R_i/10)^2$ and $(11R_i/10)^2$. That is, up to small factors we can think of all points of $N_i$ having the cost $R_i^2$. 

In every step of the algorithm, we say that we are either in the \emph{easy} case or in the \emph{hard} case.
We say that we are in the easy case if it holds that for $c$ sampled proportionally to its cost we have $\phi(X, C_i) - \phi(X, C_i \cup \{c\}) \le \phi(N_i, C_i)/2$ with probability at least  $1-1/\ell$, otherwise we are in the hard case. 
In the following discussions, we explain what needs to be done if all steps are easy and then if all the steps are hard. The fact that some steps are easy and some are hard does not make the final analysis more complicated. 


\paragraph{Easy case}
The reason why the easy case is in fact easy is that in that case, we can verify that whenever we sample a candidate center from $N_i$, the greedy rule adds it to the set of centers with constant probability. This intuitively means that for every hit of $K$ we are getting a lot of progress in terms of the size of $N_i$ going down rapidly. 

More precisely, we first observe that for a fixed $1 \le j \le \ell$, whenever $c_{i+1}^j \in N_i$, we have $\phi(X, C_i) - \phi(X, C_i\cup \{c_{i+1}^j\}) \ge \phi(N_i, C_i)/2$. This inequality holds because of our simplifying assumption that all points in $N_i$ have distance at most $R_i/10$ from $\mu(K)$: this assumption implies that the cost of each point $x$ in $N_i$ drops from $d(x,C_i)^2 \ge (9R_i/10)^2$ to at most $(d(c_{i+1}^j, \mu(K)) + d(\mu(K), x))^2 \le (2R_i/10)^2$.   

This means that whenever we sample a candidate center $c_{i+1}^j \in N_i$ and we are in the easy case, we also have constant probability of $(1-1/\ell)^{\ell-1} \ge 1/\e$ that all other sampled points result in a smaller cost drop than $c_{i+1}^j$ and thus the greedy rule decides that $c_{i+1} = c_{i+1}^j$. 
We claim this means that in the easy case, a sampled point from $N_i$ means a constant probability of $|N_{i+1}| \le |N_i|/2$. 
To see this, let us order the points of $N_i$ in their increasing distance from $\mu(K)$ as $n_1, n_2, \dots, n_{|N_i|}$ (see \cref{fig:upper_bound}). 
Importantly, if $c_{i+1}^j = n_t$, then $N_{i+1}$ does not contain any of the points $n_{t+1}, n_{t+2}, \dots, n_{|N_i|}$.  
Recall that all points of $N_i$ have the same cost, up to constant factors. 
Hence, if we condition on $c_{i+1}^j \in N_i$, we know it is essentially a point of $N_i$ selected uniformly at random. This means that with constant probability $c_{i+1}^j \in \{n_1, \dots, n_{|N_i|/2}\}$ and hence $|N_{i+1}| \le |N_i|/2$, as we wanted. 

To conclude, note that whenever we sample a point from $N_i$, we also hit $K$ with probability $\phi(K, C_i) / \phi(N_i, C_i) \approx |K|/|N_i|$. But above discussion shows that a hit of $N_i$ implies that $|N_{i+1}| \le |N_i|/2$ (with constant probability). 
Hence, the expected total number of hits of $K$ of the easy case can be upper bounded by $\frac{|K|}{|X|} + \frac{|K|}{|X|/2} + \dots + \frac{|K|}{2|K|} + \frac{|K|}{|K|} = O(1)$. 

\paragraph{ Hard case}
Next, let us turn to the hard case. There, we at least know that with probability $1 - (1 - 1/\ell)^\ell \ge 1-1/\e$, at least one candidate center $c_{i+1}^j$ for $1\le j \le \ell$ makes the cost drop by at least $\phi(N_i, C_i)/2$. 
We hence get that $\phi(X, C_i) - \E[\phi(X, C_{i+1})] \ge (1-1/\e)\phi(N_i, C_i)/2 \ge \phi(N_i, C_i)/4$. 

Note that whenever it is the case that $\phi(N_i, C_i) \le \phi(X, C_i)/k$, this implies that the probability we sample from $N_i$ (and hence from $K$) is $\frac{\phi(N_i, C_i)}{\phi(X, C_i)} \le 1/k$. Even if we sum up over all $k$ steps, the total contribution in terms of number of samples from $K$ is negligible.  So, let us assume that $\phi(N_i, C_i) \ge \phi(X, C_i)/k$.

We will now consider the sequence of sampling steps of the algorithm until a step $j$ when it selects a center from $N_i$. 
Until then, we have for each step $i+1 \le i' < j$ that $N_i = N_{i'}$, that is, the neighborhood of $K$ is the same. 
Notice that the cost of all points in $N_i$ also stays roughly the same during this time. 
This is because of our simplifying assumption that all points of $N_i$ have distance at most $R_i/10$ from $\mu(K)$ and hence the cost of all points of $N_i$ is between $(9R_i/10)^2$ and $(11R_i/10)^2$. This remains so even after sampling new centers with distance larger than $R_i$ from $K$. 

After $\phi(X, C_i)/\phi(N_i, C_i)$ steps, we expect to sample $O(\ell)$ many candidate centers from $N_i$. Meanwhile, the cost of $X$ is expected to drop from $\phi(X, C_i)$ to $3/4 \, \cdot \phi(X, C_i)$ or smaller. 

Recall that we assume that $\phi(N_i, C_i) \ge \phi(X, C_i)/k$. This means that after expected $O(\ell \log k)$ candidate centers sampled from $N_i$, we expect that $\phi(X, C_{i'}) < \phi(N_i, C_{i'})$ which is a contradiction. 
In other words, when we finally reach the step $j$ when a center is picked from $N_i$, this center is some point out of $O(\ell \log k)$ candidate centers sampled from $N_i$ so far between the steps $i$ and $j$. Each of these candidates is essentially a uniformly random point of $N_i$ (see \cref{fig:upper_bound}). 

We can now (as in the easy case) order the points of $N_i$ as $n_1, \dots, n_{|N_i|}$ in the order of increasing distance from $\mu(K)$. We sampled $O(\ell \log k)$ candidate centers from $\{n_1, \dots, n_{|N_i|}\}$ essentially uniformly at random; hence with constant probability none of them hits the set of top $|N_i|/O(\ell \log k)$ points $\{n_{|N_i| - |N_i|/O(\ell \log k)}, \dots, n_{|N_i|}\}$, hence with constant probability we have $|N_{i+1}| \le |N_i|(1-1/O(\ell\log k))$.

\begin{figure}[t]
    \centering
    \includegraphics[width = \textwidth]{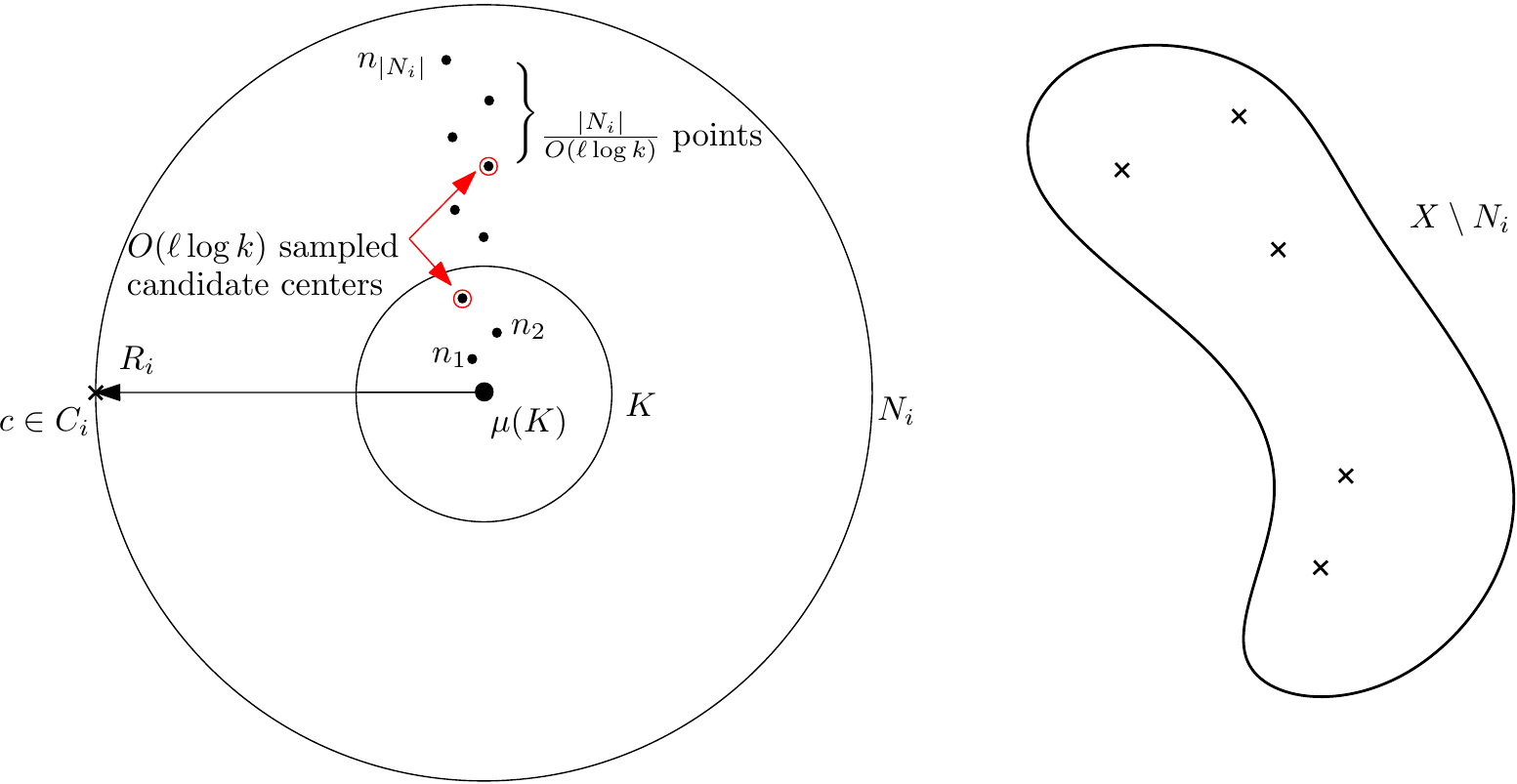}
    \caption{The figure shows an optimal cluster $K$ and its neighborhood $N_i$.
    The points of $N_i$ are sorted as $n_1, \dots, n_{|N_i|}$ based on their distance from $\mu(K)$. Note that the picture is not to scale as our simplifying assumption says that even $n_{|N_i|}$ has distance at most $R_i/10$ from $\mu(K)$. \\
    During the steps $i' \ge i$ until the step $j$ with $N_j \not= N_i$, the relative cost $\phi(N_{i}, C_{i'})/\phi(X, C_{i'})$ of the neighborhood $N_i$ increases from at least $1/k$ to at most $1$, we expect to sample $\ell \log k$ candidate centers from $N_i$ (highlighted by red color). 
    One of these samples needs to be chosen as the new center in the step $j$ that defines the new $N_j \not= N_i$. 
    Since the red points are chosen essentially uniformly at random, we expect the top $|N_i|/(\ell \log k)$ points of $N_i$ to be removed from $N_{i+1}$. }
    \label{fig:upper_bound}
\end{figure}

That is, the size of $|N_i|$ is expected to drop by $1-1/O(\ell \log k)$ factor while $O(\ell \log k)$ candidate centers hit $N_i$, which corresponds to expected $\frac{|K|}{|N_i|}\cdot O(\ell \log k)$ hits of $K$. 
Put differently, after $\frac{|K|}{|N_i|}O(\ell^2\log^2 k)$ hits of $K$ we expect the size of the neighborhood $N_i$ to halve. 

Similarly to the easy case, this implies that the total number of hits of $K$ can be upper bounded by $O(\ell^2\log^2 k) \cdot \left( \frac{|K|}{|X|} + \frac{|K|}{|X|/2} + \dots + \frac{|K|}{2|K|} + \frac{|K|}{|K|}\right) = O(\ell^2\log^2 k)$. This finishes the hard case analysis and hence the whole proof. 

The full proof of \cref{lem:greedy_lemma} is in \cref{sec:greedy_lemma}. The only substantial difference from the above sketch is that as we do not have the simplifying assumption that all points of $N_i$ are $R_i/10$-close to $\mu(K)$, we need to work with two sets $N_i^{small}, N_i^{big}$ instead. 
We remark that one can replace the term $O(\ell^2\log^2 k)$ by $O\left(\ell \log \frac{OPT(1)}{OPT(k)}\right)$ with a substantially easier proof, where $OPT(\tilde{k})$ is the size of the optimal solution with $\tilde{k}$ centers. We give a proof sketch in \cref{sec:alternative}.

\subsection{Return of the guarantees}
\label{subsec:return}

After proving \cref{lem:greedy_lemma}, we are ready to prove the main result, \cref{thm:greedy_ub}. 
This is proven by adapting the $k$-means++ analysis of Arthur and Vassilvitskii \cite{arthur2007k}. As we have seen, their analysis gives $O(\log k)$ approximation guarantee. We pick up additional $O(\ell^2\log^2 k)$ factor after using \cref{lem:greedy_lemma} to conclude that when a cluster $K$ becomes covered, we expect its cost to be $5\cdot \E[\hit(K)]\cdot\phi^*(K) = O(\ell^2\log^2 k)\cdot \phi^*(K)$. Finally, we need one more $\ell$ factor since the probability of a bad step can now be bounded only by $\frac{\ell \phi(X_i^\fU, C_i)}{\phi(X, C_i)}$ instead of $\frac{\phi(X_i^\fU, C_i)}{\phi(X, C_i)}$. This results in $O(\ell^3\log^3 k)$ expected approximation guarantee. 

The above discussion may suggest that everything simply falls in place but there is one important subtlety that our analysis needs to deal with. 
In the $k$-means++ analysis, we paid the cost of $\phi(X_i^\fU, C_i)/u_i$ in every bad step. 
Recall that this was substantiated by the fact that the size of an average uncovered cluster is only expected to go down in the future steps, i.e., for $k$-means++ we can prove that $\E\left[\phi(X_{i+1}^\fU, C_{i+1}) / u_{i+1}\right] \le \phi(X_i^\fU, C_{i})/u_i$. 
However, this bound is not necessarily true for $\ell > 1$. In fact, in the case of an adversarial rule, one can imagine the average size of a cluster increases substantially during the algorithm. 
Whether it is indeed so is an exciting open problem described in  \cref{subsec:lprocess}. This is also the reason behind the mismatch in our upper and lower bounds of \cref{thm:adversary_ub,thm:adversary_lb}. 

Fortunately, the fact that our rule is greedy and not an arbitrary one saves us again. 
To see this, first assume that instead of the greedy rule we use a different, idealized, rule that picks the candidate center $c_{i+1}^j$ that minimizes the expression $\phi(X_{i+1}^\fU, C_{i} \cup \{c_{i+1}^j\})$ instead of $\phi(X, C_{i}\cup \{c_{i+1}^j\})$ like the greedy rule. 
For such an idealized algorithm we have
\[
\E\left[ \phi(X_{i+1}^\fU, C_{i+1})/u_{i+1}\right]
\le 
\E\left[ \phi(X_{i+1}^\fU, C_{i}\cup \{c_{i+1}^1\})/u_{i+1}\right]
\le \phi(X_i^\fU, C_i)/ u_i. 
\]
The first inequality above is saying that our rule picks the best candidate $c_{i+1}^j$ which is at least as good as the first candidate $c_{i+1}^1$. The second inequality is just the reasoning of the original $k$-means++ analysis.  
Therefore, for this idealized algorithm, the original analysis of $k$-means++ immediately generalizes. 

Fortunately, our greedy rule is almost this idealized algorithm!
The difference between the idealized minimization of $\phi(X_{i+1}^\fU, C_{i+1})$ and the actual greedy minimization of $\phi(X, C_{i+1})$ just creates two small mismatches: When the greedy rule considers a candidate center $c_{i+1}^j \in K$, it, in addition to the idealized rule, takes into account 1) the decrease in the cost of already covered clusters $\phi(X_i^\fC, C_i) - \phi(X_i^\fC, C_i\cup \{c_{i+1}^j\})$ and 2) the cost of the newly covered cluster $\phi(K,C_i \cup \{ c_{i+1}^j\})$. 
To give an example of the point (1), we can have a cluster $K$ with $\phi(K, C_i) \approx 0$ that is very close to covered clusters. When we sample $c_{i+1}^j \in K$ and we have $\phi(X_i^\fC, C_i \cup \{c_{i+1}^j\}) \ll \phi(X_i^\fC, C_i)$; for the greedy rule then taking $c_{i+1}^j$ is an attractive option although taking it increases the average cost of the remaining uncovered clusters.

\todo{kuba precte}
Fortunately, the two mismatches can be handled. 
After some calculations, it turns out that the first mismatch implies that we need to pay an additional cost of $\phi(X_i^\fC, C_i)/u_i$ per step in the analysis, but fortunately we already pay this term in the original analysis because of the possibility that the $i+1$th step can be bad. 

One can also show that the second mismatch implies that we need to pay additional factor of $\sum_{K \in \fK_i^\fU} \frac{\phi(K, C_i)}{\phi(X, C_i)}\phi^*(K)$ where $\fK_i^\fU$ is the set of uncovered clusters. Fortunately, if we sum this expression over all steps, this is simply counting the number of hits to each optimal cluster $K$. 
So, in total, we need to pay additional term $ \sum_{K \in \fK_0^\fU} \E[\hit(K)] \phi^*(K)$ in the approximation guarantee, but we are again already paying this term anyway because this is our upper bound on the cost of a cluster $K$ once it becomes covered. 

To conclude, the mismatch between the idealized algorithm and the actual greedy algorithm can be accounted for and the increase in the approximation factor is asymptotically dominated by terms we already have to pay in the analysis anyway for different reasons.

\subsection{Matching lower bound}
\label{subsec:lb}

At this point, we have already a good understanding of where different terms in the approximation guarantee $O(\ell^3 \log^3 k)$ are coming from. 
This allows us to construct a point set where the above analysis is close to tight. 

\todo{change e0 add nt+1}

\begin{figure}[t]
    \centering
    \includegraphics[width = .9\textwidth]{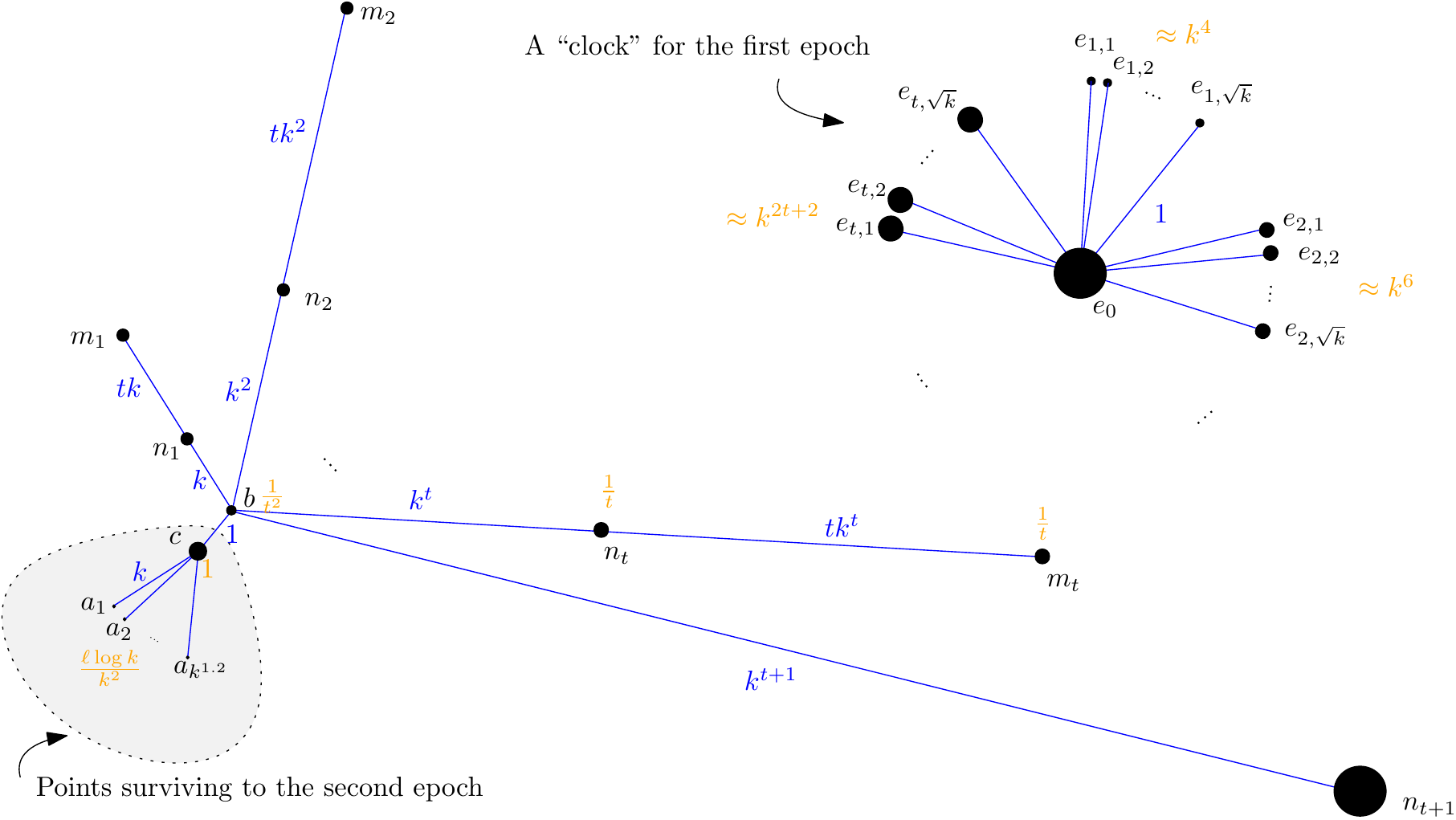}
    \caption{This figure shows the point set $X$ used for the lower bound (up to small changes by $\poly\log (\ell \log k)$ factors). 
    The point weights are orange and heavier points are represented by larger disks. \\
    In this intuitive explanation, the metric is not Euclidean but it is a tree (or a bit more precisely forest) metric induced by the blue edges with the distances given by blue numbers. 
    The image is not to scale, for example, the distance from $b$ to $n_t$ is in fact much larger than the distance from $b$ to $m_2$. \\
    In the first two steps, the points $n_{t+1}$ and $e_0$ are taken. 
    During the $i$th phase of the first epoch we mostly take just points of $E_i = \{e_{i,1},\dots, e_{i,\sqrt{k}}\}$ that serve as a ``clock''. This clock is ticking for enough time so that the algorithm samples $n_i$ as a candidate. Because of $m_i$, $n_i$ is then selected by the greedy rule. 
    This drastically reduces the costs of the not-yet-taken points of $X \setminus E$ so that nothing interesting happens until all points of $E_i$ are taken (the clock for this phase stops ticking) and then we go to the next phase. 
    The aim of these phases is that in each one of them we have a small probability of $1/t$ of sampling  $b$ and taking it as the center. As there are $t$ phases, we get constant probability of taking it. \\
    In the second epoch, we simply show that the greedy $k$-means++ samples and then takes $c$ with constant probability, which increases the approximation factor by additional $\ell\log k$.  }
    \label{fig:adversary_lb_intuitive}
\end{figure}\todo{add $n_{t+1}$.}

\paragraph{Point set definition}
We describe the point set $X$ next (see also \cref{fig:adversary_lb_intuitive}) and then explain intuitively what happens when the greedy $k$-means++ is run on it. We will describe the lengths and weights here up to factors of size $O(\log(\ell \log k))$ that need to be added to make the construction work. 
Given a parameter $k$, there are $1+\tilde{k} = O(k^{1.2})$ points in $X$ and we ask for a solution with $\tilde{k}$ centers, hence exactly one point of $X$ is going to be excluded. We set $t = \ell \log k$.

\begin{enumerate}
    \item There is a point $b$ for which we have $w(b) = \frac{1}{t^2}$.
    \item A point $c$ is at distance 1 from $b$. We have $w(c) = 1$. The points $\{b,c\}$ form one cluster of the optimal solution. While the optimal solution would take $c$ as a center and would not take $b$, we will argue that greedy $k$-means++ takes $b$ as a center with constant probability. 
    \item We have a set of points $N = \{n_1, n_2, \dots, n_{t + 1}\}$ and  $M = \{m_1, \dots, m_t\}$ defined as follows. 
    We have $d(b, n_i) = k^i, d(n_i, m_i) = tk^i$ and $w(n_i) = w(m_i) = \frac{1}{t}$. 
    An exception is the point $n_{t + 1}$ that gets very large weight so that it is selected as a center in the first two steps of the algorithm.  
    \item We have $E = \{e_0, \bigcup_{e \in E_1} e, \dots, \bigcup_{e \in E_t} e\}$ where $E_i = \{e_{i,1}, \dots,  e_{i, \sqrt{k}}\}$. 
    Each point in $E \setminus \{e_0\}$ has distance $1$ from the point $e_0$ which is very far from $\{b,c\}\cup N\cup M\cup A$. The point $e_0$ has very large weight so that it is selected as a center in the first two steps of the algorithm. 
    Each $e_{i,j} \in E_i$ has the same weight $w_i \approx k^{2i+2}$. We postpone the exact definition of $w_i$ as it needs to be quite precise. 
    \item We have $A = \{a_1, \dots, a_{k^{1.2}}\}$ at distance $k$ from $c$. 
    The weight of each $a_i$ is $\frac{\ell\log k }{k^2}$ so their total weight is $\frac{\ell \log k }{k^{0.8}}$. 
\end{enumerate}

The optimal solution on $X$ takes all the points except $b$ as centers and hence it pays $1/(\ell^2 \log^2k)$. 

Here is an intuitive explanation how  greedy $k$-means++ runs on $X$. There are two epochs. In the first epoch, the set $A$ is too small to be discovered by the algorithm, so only $\{b,c\}\cup N\cup M \cup E$ is relevant. Our aim is to show that with constant probability we reach a situation where only the points in $A \cup \{c\}$ are not taken; when this happens we say that the second epoch starts. 

Note that $K = \{b,c\}$ is a cluster in the optimal solution. Hence, we are proving that in the first epoch the cost of $K$ under the greedy $k$-means++ is approximated only by factor $\Omega(\ell^2\log^2 k)$, matching the bound in the proof of \cref{lem:greedy_lemma}. Moreover, for each $i$ the sets $\{n_1, m_1, \dots, n_i, m_i\}$ are playing the role of the neighborhood $N_i$ of $K$ in the analysis from \cref{subsec:hope}, while the set $E$ plays the role of $X \setminus K$. 

\paragraph{ First epoch }
Here is what is going to happen in the first epoch. We can split it into $t$ phases where at the beginning of the $i$th phase all the points in $N_{\ge i+1} \cup M_{\ge i+2} \cup E_{\ge i+1}$ \footnote{The notation $N_{\ge i+1}$ means $\{n_{i+1}, \dots, n_t\}$. } are already selected as centers. 
During the $i$-th phase, we do not sample points from $A \cup E_{<i}$ as they have too small a probability of being sampled. Mostly, we sample just points of $E_i$ since each point there has cost about $k^{2i+2}$. These points serve as a kind of clock. During the time we are sampling mostly points of $E_i$, we also have a small chance of sampling points of $\{b,c\}\cup N_{\le i}\cup M_{\le i+1}$. Since we start with $|E_i| \ge \sqrt{k}$, before we add all of $E_i$ to the set of centers, we expect to sample the point $c$ about $\ell \log k$ times. Similarly, each point in $N_{\le i } \cup M_{\le i+1}$ is expected to be sampled constantly many times. The point $b$ has only probability of around $(\ell\log k)/t^2 = 1/t$ of being sampled. 

Now we can fine-tune the weight of points in $E_i$ so that the drop in the cost if we take $e \in E_i$ is larger than the drop resulting from taking $c$ but smaller than the drop of $b$ (taking $b$ results in a larger drop than $c$, since $c$ is further from $N_{\le i} \cup M_{\le i}$ than $b$). 
This means that the drop of $e$ is smaller than the drop of the points $\{b, n_i, m_{i+1}\}$, but it is larger than the drop of all other points. 
In fact, the reason why we always have a pair $\{n_i, m_i\}$ is to make the cost drop of $n_i$ large: as $n_i$ lies roughly in the center of mass of $c$ and $m_i$, it is a very attractive point to take from the perspective of the greedy rule. 
So, during the $i$th phase we are essentially just waiting until we sample $n_i$. When we encounter it, we add it to the set of centers which decreases the cost of $\{b,c\}\cup N_{< i} \cup M_{\le i}$ dramatically so that in the rest of the phase only the points of $E_i$ are sampled. Also, the point $m_{i+1}$ that is leftover from the previous phase is at some point sampled and selected as a center. 

This process is running for $t$ phases and in each phase, we have probability of $1/t$ of sampling $b$. 
After $b$ is sampled, all other points except for $A \cup \{c\}$ are sampled in the following steps; the weight of all of them is much larger than the weight of points of $A \cup \{c\}$. When this is done, the first epoch is finished. 

\paragraph{Second epoch}
While the first epoch corresponds to the bound that we lose in \cref{lem:greedy_lemma} for not approximating well clusters that get covered, the second epoch corresponds to the rest of the analysis in which we lose additional $\ell \log k$ factor because of the fact that some steps are bad. 
In our case, a bad step means that we select a point $c$ from the optimal cluster $K = \{b,c\}$ that we already covered in the first epoch by selecting $b$. 

The second epoch begins when all points except of $A \cup \{c\}$ are already taken as centers. 
Note that in the following $k^{1.2} - k^{1.1}$ steps we have a constant probability that we sample $c$ as a candidate center. 
In that case, the greedy heuristic decides to pick $c$ over the other candidates from $A$. 
This is because the cost drop induced by any $a \in A$ is simply $w(a) \cdot d(a, b)^2 \approx \ell\log k$, whereas selecting $c$ makes the cost of each $a \in A$ drop by roughly $w(a) \cdot \left((k+1)^2 - k^2\right) \approx \frac{\ell \log k}{k^2} \cdot k = \frac{\ell \log k}{k}$. However, there are at least $k^{1.1}$ points in $A$, hence the drop of $c$ is larger than the drop of any $a \in A$. 

\paragraph{Putting it together}
Putting the two epochs together, we get a constant probability that the greedy $k$-means++ algorithm first fails by covering $K$ using $b$, instead of $c$, and then it fails again by taking $c$ although $K$ is already covered. This means that one point of $A$ is not covered in the end and we need to pay $\ell \log k$ for it, whereas the optimum opts not to take $b$ and hence pays only $1/(\ell^2 \log^2 k)$. 

One problem with the above analysis is that at each epoch we have a constant probability of not sampling $n_i$ before all points of $E_i$ are added to the set of centers. In that case, our analysis fails since it is probable that the algorithm will soon afterward sample $c$ and take it as a center. 
However, adjusting weights by $\poly(\log t) = \poly \log(\ell\log k)$ factors makes the failure probability of one phase smaller than $1/(2t)$ so that we can union bound over them. 

\section{Preliminaries}
\label{sec:preliminaries}

We use $X$ for the input point set. In case $X$ is weighted as discussed below in the lower bound section, every point $x \in X$ comes with nonnegative weight $w(x)$, in the unweighted case we have $w(x) = 1$. 
Whenever we talk about an optimal cluster $K$, we tacitly assume a fixed optimal solution $C^*$ and $K \subseteq X$ is a set of points defined by some $c \in C^*$ as $K = \{ x \in X : c = \argmin_{c' \in C^*} \phi(x, c')\}$. 

We use $d(x,y)$ to denote the distance between two points $x,y$ rather then $||x-y||$, since all our upper bounds generalize to general metric space (the only difference being that \cref{lem:2apx,lem:5apx} require larger constant factors in that case).

For $x \in X$ we use $\phi(x, C) = w(x) \cdot \min_{c \in C}d(x,c)^2$ and for $K \subseteq X$ we define $\phi(K, C) = \sum_{x \in K} \phi(x, C)$. 
For a cluster $K$ we write $\phi^*(K) = \phi(K \mu(K))$ where $\mu(K)$ is the center of mass of $K$, i.e., $\mu(K) = (\sum_{x \in K} x)/w(K)$ for $w(K) = \sum_{x \in K} w(x)$. That is, $\phi^*(K)$ is the smallest cost of $K$ achievable if we have just one center. 

When we write $\E_{\ge i+1}[ X]$ in relation to \cref{alg:kmpp,alg:kmpp_greedy,alg:kmpp_adversary}, we tacitly assume that the randomness of the first $i$ steps is fixed and the expectation is over the randomness in the rest of the algorithm. Similarly, $\E_{i+1}[X]$ is an expectation over the randomness in the step $i+1$. 

\paragraph{Standard lemmas from \cite{arthur2007k}}
We will need some lemmas from \cite{arthur2007k}. 

First, note that for the squared distance we have the approximate triangle inequality
\begin{equation}
\label{eq:triangle_inequality}
    d(x, z)^2 \le 2(d(x,y)^2 + d(y,z)^2)
\end{equation}

The following lemmas can be found in \cite{arthur2007k}. 

\begin{lemma}[Lemma 2.1 in \cite{arthur2007k}]
\label{lem:steiner}
For any $K, c$ we have $\phi(K, c) = \phi^*(K) + |K|c^2$. 
\end{lemma}

\begin{lemma}[Lemma 3.1 in \cite{arthur2007k}]
\label{lem:2apx}
If $c$ is a uniformly randomly selected point of $K$, we have $\E[\phi(K, c)] = 2\phi^*(K)$. 
\end{lemma}

\begin{lemma}[Lemma 3.2 in \cite{arthur2007k}, Lemma 4.2 in \cite{makarychev_reddy_shan2020improved}]
\label{lem:5apx}
\todo{This should be lemma 3.2 in AV07}Fix two point sets $K,C$ and sample a random points $c \in K$ with probability proportional to $\phi(c, C)$. Then, 
\[
\E[\phi(K, C\cup \{c\})] \le 5\phi^*(K).
\]
\end{lemma}
We note that the original version of the lemma from \cite{arthur2007k} had the constant being 8, this was improved to 5 in the work of \cite{makarychev_reddy_shan2020improved}. 

\paragraph{Lemmas for lower bounds}
For lower bounds, it will be easier to work with the more general weighted version of the $k$-means problem. Any lower bound for the weighted version can be lifted to an unweighted one by multiplying all weights by a large number and rounding them to the closest integer. 

Even more generally, it will be more convenient to prove our lower bounds for the generalized $k$-means problem where the input contains not only a weighted point set $X$ and $k$, but also a set of prescribed centers $C_0$. We next use $(X, k, C_0)$ to denote input to such generalized problem. We will use the following fact. 

\begin{lemma}
\label{lem:prescribed}
Suppose that any of the algorithms \cref{alg:kmpp,alg:kmpp_greedy,alg:kmpp_adversary} returns at least $\alpha$-approximate solution on some $k$-means instance $(X,k,C_0)$ with constant probability. Then there exists an instance $(X', k+|C_0|, \emptyset)$ where the algorithm is still at least $\alpha$-approximate with constant probability. 
\end{lemma}
\begin{proof}[Proof Sketch]
We simply define $X'$ to be $X\cup C_0$, where we make the weight of each point in $C_0$ substantially larger than the total weight of $X$ so that the first $|C_0|$ steps only points of $C_0$ can be sampled, with high probability.   
\end{proof}
Above construction can, in general, substantially increase the weights. However, we only use it in case $|C_0| \le 2$ where this is not the case. 

Another useful fact is that if we restrict our solution to $k$-means to points of $X$, we lose only a $2$-factor in approximation. This lemma follows directly from \cref{lem:2apx}. 
\begin{lemma}
\label{lem:wlog_opt_subset}
Whenever there is some solution $C_k$ to an instance of $k$-means, there is also a solution $C'_k \subseteq X$ such that $\phi(X, C'_k) \le 2\phi(X, C_k)$. 
\end{lemma}

Finally, our lower bounds would be easier if they were proven in general metric spaces. 
It is simple, though a bit technical, to make them work in Euclidean spaces. To this end, we need the following result about embedding a star metric to the Euclidean space. 

\begin{fact}
\label{fact:simplex}
There is a way to arrange $d$ vectors $v_1, \dots, v_d$ in Euclidean space in such that for any two $v_i, v_j$, we have $\langle v_i, v_j \rangle = -1/(d-1)$. 
In particular, we get this property by arranging $v_1, \dots, v_d$ as vertices of a $d-1$ dimensional simplex. 
\end{fact}

\paragraph{Few more results}
All logarithms in this paper are natural. We will use the following facts. 
\begin{fact}
\label{fact:log_cool}
For any $x > 0$ we have 
$\log(x) \ge 1 - 1/x   
$.
\end{fact}

\begin{fact}
\label{fact:one_minus_x_lowerbound}
For any $x \in [-1, 0]$ we have 
$\left(1 +\frac{x}{2}\right) \geq e^x$.
\end{fact}

We also need the following lemma that says that if we learn about a random variable that it is larger than some other independently sampled variables it can only increase its expectation. 

\begin{lemma}
\label{lem:conditioning}
Let $X, Y_1, \dots, Y_t$ be independent random variables. 
We have $$\E[X | X = \min(X, Y_1, \dots, Y_t)] \le \E[X].$$ 
\end{lemma}
\begin{proof}
We assume all variables are discrete. Let us consider any $(y_1, \dots, y_t) \in \R^t$, and fix $Y_i = y_i$ for all $1 \le i \le t$. 
We have that the conjunction of the event $\forall i : Y_i = y_i$ together with the event $X = \min(X, Y_1, \dots, Y_t)$ is equivalent to the conjunction of the event $\forall i : Y_i = y_i$ with $X \le \min(y_1, \dots, y_t)$. 
Hence, we can write
\begin{align*}
    \E[X |  X = \min(X, Y_1, \dots, Y_t)]
    &= \sum_{(y_1, \dots, y_t) \in \R^t} \P(\forall i : Y_i = y_i) \cdot \E\left[ X | X = \min(X, Y_1, \dots, Y_t) \wedge \forall i : Y_i = y_i\right]\\
    &= \sum_{(y_1, \dots, y_t) \in \R^t} \P(\forall i : Y_i = y_i) \cdot \E\left[ X | X \le \min(y_1, \dots, y_t) \wedge \forall i : Y_i = y_i\right]\\
    &= \sum_{(y_1, \dots, y_t) \in \R^t} \P(\forall i : Y_i = y_i) \cdot \E\left[ X | X \le \min(y_1, \dots, y_t)\right]
\end{align*}
However, for any $(y_1, \dots, y_t) \in \R^t$ we have $\E\left[ X | X \le \min(y_1, \dots, y_t)\right] \le \E[X]$ and the lemma follows. 
\end{proof}

\section{Bounds on cluster hits by greedy $k$-means++}
\label{sec:greedy_lemma}

In this section we give a formal proof of  \cref{lem:greedy_lemma}. We start by giving preparatory definitions and the statement of \cref{lem:greedy_lemma_induction} -- inductive variant of \cref{lem:greedy_lemma} -- in \cref{subsec:definitions}.
The lemma is then proved in \cref{subsec:proof_lemma_greedy}.

\subsection{Preparatory definitions and results}
\label{subsec:definitions}

We start by formally defining covered and solved clusters. 

\begin{definition}[Covered Cluster]
\label{def:covered}
Consider some optimal cluster $K$. We refer to $K$ as covered with respect to a set of centers $C$ if $K \cap C \not= \emptyset$. 
\end{definition}

\begin{definition}[Solved Cluster]
\label{def:solved}
Consider some optimal cluster $K$. We refer to $K$ as solved with respect to point set $C$ if 
\[
\phi(K,C) \leq 10^5 \phi^*(K).
\]
\end{definition}
We also say that $C$ is covered/solved in the $i+1$th step of the algorithm if it is covered/solved with respect to $C_i$. 

Next, we formally define the object of interest, that is, the number of points we are going to sample from $K$ during \cref{alg:kmpp_greedy}. 

\begin{definition}[Number of points we will sample]
\label{def:hit}
Let $\hit_{i+1}^j(K)$ be an indicator variable for the conjunction of the following three events:
\begin{enumerate}
    \item $c_{i+1}^j \in K $
    \item $K$ is not covered with respect to $C_{i}$
    \item $K$ is not solved with respect to $C_{i}$
\end{enumerate}
We further define $\hit_i(K) :=  \sum_{j = 1}^\ell \hit_{i}^j(K)$, $\hit_{\ge i}(K) = \sum_{\iota = i}^{k} \hit_\iota(K)$ and $\hit(K) = \hit_{\ge 1}(K)$. 
\end{definition}

Recall that \cref{lem:greedy_lemma} asks us to prove that $\hit(K) = O(\ell^2 \log^2 k)$. 

We will need a few more definitions to state our main technical result \cref{lem:greedy_lemma_induction} which is an inductive version of \cref{lem:greedy_lemma}. 
We start with the definition of the parameter $R_i$. This parameter is, up to constant factors, equal to the distance $d(\mu(K), C_i)$ of the center of mass of $K$ with the closest center of $C_i$. 
However, it may be that $R_i = R_{i+1}$ although $d(\mu(K), C_i) = d(\mu(K), C_{i+1})$. This is because whenever $R_i$ changes, we want it to change by a factor of $10$ for a reason explained later. 

We also define the index $i_0$ as the smallest index of a step where the size of $K$ becomes ``non-negligible''. Before step $i_0$, we have only a negligible probability of sampling a point from $K$. 
We note that from now on, we consider the cluster $K$ fixed, so that we can talk about $R_i$ instead of $R_i(K)$, etc. 

\begin{definition}[Parameters $i_0$ and $R_i$]
\label{def:R}
Fix an optimal cluster  $K$. 
Let $i_0 = \min \{i \in \{1,2,\ldots,k\} \colon \phi(K, C_i) \ge \phi(X, C_i)/k\}$. 
For every $i \in \{i_0,i_0 + 1, \ldots,k\}$, we define a parameter $R_i$ with
\[
R_{i_0} = d(\mu(K), C_{i_0})
\]

and for $i > i_0$, we define 
\[
R_i = \left\{
\begin{array}{ll}
R_{i-1} & d(\mu(K),C_{i}) > R_{i-1}/10 \\
d(\mu(K),C_{i}) & \, \textrm{otherwise.} \\
\end{array}
\right.
\]
\end{definition}
Note that $R_i$ satisfies
\begin{align}
\label{eq:Ri_reasonable} 
d(\mu(K),C_{i}) \le R_i \le 10 d(\mu(K),C_{i}).   
\end{align}
That is, $R_i$ is roughly equal to the distance of the cluster center to the set $C_i$. 
Next, we always implicitly assume that $i \ge i_0$, i.e., $R_i$ is well defined. 

Based on $R_i$ we also define 
\begin{equation}
    \label{eq:def_small}
    N_i^{small} = B(\mu(K), R_i/100)
\end{equation}
and
\begin{equation}
    \label{eq:def_big}
    N_i^{big} = B(\mu(K), R_i/10).
\end{equation}

These definitions correspond to the neighborhood $N_i$ from the intuition section in \cref{subsec:hope}. 
We need two different neighborhoods due to technical difficulties like the fact that only for $x \in N_i^{small}$ we have $d(x, C) = \Omega(R_i)$ (in the proof sketch of \cref{subsec:hope} we worked under a simplifying assumption that all $x \in N_i^{big}$ are also in $N_i^{small}$). The reason for the slightly weird definition of $R_i$ is that we want that whenever $R_{i+1} < R_i$, then $N_{i+1}^{big} \subseteq N_i^{small}$. 

In the following claim, we summarize several basic properties of the cluster $K$ and its neighborhoods $N_i^{small}$ and $N_i^{big}$. 
These claims substantiate some simple intuition about these sets like that all points in $N_i^{small}$ have essentially the same cost.  
A careful reader will note that we do not try to optimize the constants (and we warn that it will get worse). 
\begin{claim}
\label{cl:basic_properties}
Assume $K$ is not solved  with respect to $C_{i}$ for $i \ge 1$. Then, we have:
\begin{enumerate}
\item For each point $x \in N_i^{small}$, we have 
\[
R_i^2/200 \le \phi(x, C_i) \le 3R_i^2.
\]
and for each point $x \in N_i^{big}$ we have
\[
0 \le \phi(x, C_i) \le 3R_i^2. 
\]
\item We have \[
|N_i^{small}|\cdot R_i^2/200 
\le \phi(N_i^{small}, C_i) \le 3 |N_i^{small}| \cdot R_i^2
\]
and
\[
\phi(N_i^{big}, C_i) \le 3|N_i^{big}| \cdot R_i^2. 
\]

    \item We have \[
|K|R_{i}^2/400 \le \phi(K, C_{i}) \le 2|K|R_{i}^2.
\]
    \item At least $|K|/2$ points of $K$ are in $N_i^{small}$. Also, $\phi(K \cap N_i^{small}, C_i) \ge \phi(K, C_i)/800$. 
    
\end{enumerate}
\end{claim}
\begin{proof}
Let $c^* = \argmin_{c \in C_i} d(\mu(K), c)$. Recall that \cref{eq:Ri_reasonable} implies that
\begin{align}
    \label{eq:yes_reasonable}
    d(\mu(K), c) \le R_i \le 10d(\mu(K), c). 
\end{align}


\begin{enumerate}
    \item Consider any $x \in N_i^{big}$. We have 
    \begin{align*}
    \phi(x, C_i) 
    &\le d(x, c^*)^2 \le 2(d(x, \mu(K))^2 + d(\mu(K), c^*)^2) && \text{\cref{eq:triangle_inequality}} \\
    &\le 2\left(\left(\frac{R_i}{10}\right)^2 + R_i^2\right) && \text{\cref{eq:yes_reasonable}}\\
    &\le 3 R_i^2
    \end{align*}
On the other hand, for any $x \in N_i^{small}$, we have
\begin{align*}
    \phi(x, C_i)
    &\ge \left( d(\mu(K), c^*) - d(\mu(K), x) \right)^2\\
    &\ge \left( R_i/10 - R_i/100\right)^2\\
    &\ge R_i^2/200.
\end{align*}

\item This follows by applying bullet point (1) to every point in $N_i^{small}$ and $N_i^{big}$, respectively. 

\item 

Recall that by \cref{lem:steiner}, we have 
\begin{equation}
\label{eq:something}
\phi(K, C_i) \le \phi(K, c^*) 
= |K|\cdot d(\mu(K), c^*)^2 + \phi^*(K)
\le |K|R_i^2 + \phi^*(K). 
\end{equation}
On the other hand, 
\[
\phi(K, C_i) \ge 10^5\phi^*(K)
\]
by \cref{def:solved}. 
Hence, 
\begin{equation}
\label{eq:home}
\phi^*(K) \le \frac{|K|R_i^2}{10^5 - 1}. 
\end{equation}

If we plug \cref{eq:home} back to \cref{eq:something}, we get
\[
\phi(K, C_i) \le |K|R_i^2 + \frac{|K|R_i^2}{10^5 - 1} \le 2|K|R_i^2. 
\]
We postpone the proof of the other inequality to the end of the proof. 
\item 
Define $d$ by $\phi^*(K) = |K|d^2$, that is, $d$ is the average squared distance of a point of $K$ to $\mu(K)$. 
Note that \cref{eq:home} implies that $d^2 \le R_i^2/(10^5 - 1)$. 
By Markov's inequality, at most $|K|/2$ points can have a cost of at least $2d^2$, which is at most $2R_i^2/(10^5 - 1) \le (R_i/100)^2$. 
Hence, at least $|K|/2$ points of $K$ need to be in $N_i^{small}$. 

Moreover, using the above result and bullet point (2), we get
\[
\phi(K \cap N_i^{small} , C_i) 
\ge \frac{|K|}{2} \cdot  \frac{R_i^2}{200}
\]
and using the part of bullet point (3) that we have already proven, we have
\[
2|K|R_i^2  \ge \phi(K, C_i). 
\]
Combining the two bounds, we get
\begin{align*}
\phi(K \cap N_i^{small}, C_i) 
\ge \frac{|K|R_i^2}{400}
\ge \frac{\phi(K, C_i)}{800}.
\end{align*}
\item Finally, using the fact that at least $|K|/2$ points of $K$ are in $N_i^{small}$ by bullet point (4) and that each point $x \in N_i^{small}$ satisfies $\phi(x, C_i) \ge R_i^2/200$ by bullet point (1) we infer that $\phi(K, C_i) \ge |K|/2 \cdot R_i^2/200 = |K|R_i^2/400$ and the proof of bullet point (3) is finished. 

\end{enumerate}

\end{proof}

The reason why we deal with the steps before $i_0$ differently is that for all $i \ge i_0$ we know that $|N_i^{big}|$ is at most $O(k)$ times larger then $|K|$. 
\begin{claim}
\label{cl:small_neighborhood}
Let $i_0$ as defined in \cref{def:R}. For all $i \ge i_0$, we have
\[
|N_i^{big}| \le 4k|K|. 
\]
\end{claim}
\begin{proof}
It suffices to prove that $|N_{i_0}^{big}| \le 4k|K|$, since \cref{def:R} implies that $|N_{i+1}^{big}| \le |N_i^{big}|$ for any $i$. 

Recall that $R_{i_0} :=  d(\mu(K), C_{i_0})$.  
Consider any $x \in N_{i_0}^{big}$. 
We have 
\begin{equation}
\label{eq:phone}
    \phi(x, C_{i_0}) 
    \ge \left( d(C_{i_0}, \mu(K)) - d(\mu(K), x) \right)^2 
    \ge (R_{i_0} - R_{i_0}/10)^2 
    \ge R_{i_0}^2/2. 
\end{equation}
This implies that 
\begin{align}
\label{eq:tony}
    \phi(X, C_{i_0})
    \ge \phi(N_{i_0}^{big}, C_{i_0})
    \ge |N_{i_0}^{big}| \cdot R_{i_0}^2/2.
\end{align}

One the other hand, using  bullet point (3) of \cref{cl:basic_properties} we get that 
\begin{equation}
\label{eq:tonik}
2|K|R_{i_0}^2 \ge \phi(K, C_{i_0}). 
\end{equation}
By definition of $i_0$ we know that $\phi(K, C_{i_0}) \ge \phi(X, C_{i_0})/k$.  
Putting this together with \cref{eq:tony,eq:tonik}, we get
\begin{align*}
    2|K|R_{i_0}^2 \ge |N_{i_0}^{big}| \cdot R_{i_0}^2/(2k)
\end{align*}
or
\[
|N_{i_0}^{big}| \le 4k|K|,
\]
as needed. 
\end{proof}

\subsection{Inductive version of the main result}
\label{subsec:proof_lemma_greedy}

Finally, we are ready to state the technical, inductive version of \cref{lem:greedy_lemma}. 
The lemma is a potential argument: we cook up three potentials such that their sum in the step $i$ is an upper bound on how many candidate centers are expected to be sampled from $K$. 
The advantage of such an argument is that once the potentials are written down, checking the correctness of the argument reduces to algebra. The disadvantage is that it is hard to get an intuition about the high-level picture (compare with the original analysis of $k$-means++ and its rewording of Dasgupta \cite{Dasgupta}). Hence, we first invite the reader to read \cref{subsec:return} where we try to convey the high-level intuition about the argument. This section is primarily optimized for making it easy to check the correctness of the proofs. 

In \cref{lem:greedy_lemma_induction}, we track the number of hits to $K$ using three potentials. 
\begin{enumerate}
    \item The first potential $\pi_i$ is ``dropping'' a potential of $O(\ell/k)$ uniformly in every step. 
    This allows us to argue about some edges cases like when $\phi(K, C_i) < \phi(X, C_i)/k$. In these cases we have probability of at most $\ell/k$ of hitting $K$ and the drop in $\pi_i$ can pay for that. 
    \item Whenever the value of $R_i$ changes, we ``refill'' the value of potential $\sigma_i$ to its maximum value of $O\left( \ell\log k \cdot \frac{|K|}{|N_i^{small}|}\right)$. 
    The purpose of $\rho_i$ is to pay for this refill of $\sigma_i$. 
    Intuitively, we hope that whenever $R_i$ drops, the size of $|N_i^{big}|$ drops by at least $(1 - O(1/\ell\log k))$ multiplicative factor (see the intuition in \cref{subsec:hope}). In that case, the drop in $\rho_i$ is proportional to exactly $O(\frac{\ell^2\log^2 k }{\ell \log k} \cdot \frac{|K|}{|N_i^{small}|})$, so it can pay for the refill of $\sigma_i$. 
    \item The third potential $\sigma_i$ is designed to pay for the possibility of hitting $K$ until we redefine $R_i$. Intuitively, if we are in the hard case when hitting $K$ does not result in covering it, we are expecting the cost of our solution to drop. Accordingly, drop in the overall cost $\phi(X, C_i)$ results in the drop in $\sigma_i$ and this drop is paying for the possibility of hitting $K$. 
\end{enumerate}

\begin{lemma}
\label{lem:greedy_lemma_induction}
Fix an optimal cluster $K$. 
Assume we have already sampled the first $i$ points where $i \in \{1,2,\ldots,k\}$.

Then, we have 
\[
\E_{\geq i+1}[\hit_{\ge i+1}(K)] \le \pi_i + \rho_i + \sigma_i, 
\]

where
\[
\pi_i  = \left(1 - \frac{i}{k} \right) \cdot \ell + 10^{50}\ell\log k,
\]
\[
\rho_i = 10^{50}\ell^2\log^2(k) \cdot \left( 4 - \frac{|K|}{|N_i^{small}|} - \frac{|K|}{|N_i^{big}|} \right) 
\]
and
\[
\sigma_i = 10^{25}\ell \cdot  \log\min\left( 10^5 k, \frac{400\phi(X, C_i)}{|K|R_i^2}\right) \cdot \frac{|K|}{|N_i^{small}|}
\]
This is only when $K$ is neither solved nor covered with respect to $C_i$. Otherwise, we define $\pi_i = \rho_i = \sigma_i = 0$.

\end{lemma}

Before proving \cref{lem:greedy_lemma_induction}, we briefly verify that it implies \cref{lem:greedy_lemma}:

\begin{proof}[Proof of \cref{lem:greedy_lemma}]

First, note that in the first step we sample at most $\ell$ points from $K$, that is, $\hit_1 \le \ell$. 

Next, we consider running \cref{alg:kmpp_greedy} until the first time it happens that 
\begin{align}
\label{eq:passive}
\phi(K, C_i) > \phi(X, C_i)/k. 
\end{align}
Fix some $i$ such that \cref{eq:passive} does not hold and it did not hold for any $\iota \in \{1,2,\ldots,i-1\}$.  
Then, the expected number of points we sample from $K$ in the $(i+1)$-th step is at most
\begin{align}
\label{eq:passive_bound}
   \E[\hit_{i+1}] 
    &\le \frac{\ell \phi(K, C_i)}{\phi(X, C_i)}
    \le \frac{\ell}{k}.
\end{align}

In fact, the first inequality is an equality, whenever the cluster $K$ is not solved or covered. 
The second inequality uses our assumption that $\phi(K, C_i) \le \phi(X, C_i)/k$. 

Next, consider the first $i$ such that \cref{eq:passive} holds. Then, we apply \cref{lem:greedy_lemma_induction}. This lemma states that $\E[\hit_{\ge i+1}]$ can be upper bounded by a sum $\pi_i + \rho_i + \sigma_i$. From the definition of these quantities in \cref{lem:greedy_lemma_induction} we immediately see that 
\[
\pi_i = O(\ell \log k),
\]
\[
\rho_i = O(\ell^2 \log^2k)
\]\todo{though we did not define $N$ yet}
and
\[
\sigma_i = O(\ell \log k). 
\]

Hence, we get that 
\begin{align*}
\E[\hit_{\ge 1}] \le \ell + k \cdot \frac{\ell}{k} + O(\ell^2 \log^2 k) = O(\ell^2 \log^2 k). 
\end{align*}
\end{proof}

The rest of this section is dedicated to the proof of \cref{lem:greedy_lemma_induction}.

We prove the statement by (reverse) induction. First, consider the base case $i = k$. Note that for the base case it suffices to show that $\pi_i, \rho_i, \sigma_i \ge 0$. If $K$ is covered or solved with respect to $C_i$, then this directly follows from the definition.
Next, assume that $K$ is neither solved nor covered with respect to $C_i$.
Clearly $\pi_{k+1} \ge 0$. Using bullet point (4) of \cref{cl:basic_properties}, we conclude that $|N_k^{big}| \ge |N_k^{small}| \ge |K|/2$ which implies $\rho_{k+1} \ge 0$. Using bullet point (3) of \cref{cl:basic_properties}, we conclude that $\phi(X, C_k) \ge \phi(K, C_k) \ge \frac{|K|R_k^2}{400}$ which implies $\sigma_{k+1} \ge 0$. 

Next, we consider $i < k$.  Note that the statement trivially holds if $K$ is solved or covered with respect to $C_i$. Hence, from now on we consider the case that $K$ is neither solved nor covered with respect to $C_i$.

Note that we have

\todo{Subscript notation is not consistent}
\begin{align}
\label{eq:lala}
    \E_{\geq i+1}[\hit_{\ge i+1}(K)]
    &= \E_{i+1}[ \sum_{j = 1}^\ell \hit_{i+1}^j] + \E_{i+1}\left[ \E_{\geq i+2}[\hit_{\ge i+2}(K)] \right] 
\end{align}
By definition of $\hit$ and the fact that $K$ is both uncovered and unsolved, we have $\E_{i+1}[ \sum_{j = 1}^\ell \hit_{i+1}^j] = \frac{\ell \phi(K, C_i)}{\phi(X, C_i)}$ and we can use induction to bound

\begin{align}
    \E_{i+1}\left[ \E_{\geq i+2}[\hit_{\ge i+2}(K)] \right] \le \E_{i+1}[\pi_{i+1} + \rho_{i+1} + \sigma_{i+1}]. 
\end{align}
Plugging back to \cref{eq:lala}, we get
\begin{align}
    \E_{\geq i+1}[\hit_{\ge i+1}(K)] 
    &\le \frac{\ell\phi(K, C_i)}{\phi(X, C_i)} +   \E_{i+1}\left[ \pi_{i+1} + \rho_{i+1} + \sigma_{i+1} \right]. 
\end{align}

Note that the claim we want to prove is
\begin{align}
    \E_{\geq i+1}[ \hit_{\ge i+1}] \le \pi_i + \rho_i + \sigma_i, 
\end{align}
so it suffices if we prove that 
\[
\frac{\ell\phi(K, C_i)}{\phi(X, C_i)} +   \E_{i+1}\left[ \pi_{i+1} + \rho_{i+1} + \sigma_{i+1} \right] 
\le  \pi_i + \rho_i + \sigma_i. 
\]
After rearranging, we get
\begin{align}
\label{eq:almost_fundamental}
\left( \pi_i - \E_{i+1}[\pi_{i+1}]\right) + \left(\rho_i - \E_{i+1}[\rho_{i+1}]\right) + \left( \sigma_i - \E_{i+1}[\sigma_{i+1}] \right) \ge  \frac{\ell \phi(K, C_i)}{\phi(X, C_i)}.
\end{align}

The potential $\sigma$ is the only one of $\pi,\sigma,\rho$ that is not necessarily monotone in $i$. Hence, to better understand the term $\left( \sigma_i - \E_{i+1}[\sigma_{i+1}] \right)$, given the sampled point $c_{i+1}$, we define $\sigmabar$ as 
\begin{align}
    \sigmabar = 10^{25}\ell\log\min\left(10^5k, \frac{400\phi(X, C_{i+1})}{|K|R_i^2}\right) \frac{|K|}{|N_i^{small}|}. 
\end{align}
That is, in $\sigmabar$ we already change the cost $\phi(X, C_i)$ to $\phi(X, C_{i+1})$ but we do not replace $R_i$ by $R_{i+1}$ and $N_i^{small}$ by $N_{i+1}^{small}$ yet. 
We can rewrite \cref{eq:almost_fundamental} and get
\begin{align}
\label{eq:fundamental}
\left( \pi_i - \E_{i+1}[\pi_{i+1}]\right) + \left(\rho_i - \E_{i+1}[\rho_{i+1}]\right) + \left( \sigma_i - \E_{i+1}[\sigmabar] \right) + \E_{i+1}[\sigmabar - \sigma_{i+1}]  \ge  \frac{\ell \phi(K, C_i)}{\phi(X, C_i)}.
\end{align}

In the rest of the proof, we will simply need to show that \cref{eq:fundamental} is satisfied. 
We will need to consider several cases and in each one of them, we will have to lower bound the terms on the left-hand side of \cref{eq:fundamental}. 
Formally, the proof follows from \cref{cl:fundamental_small_clusters,cl:fundamental_big,cl:fundamental_small_easy,cl:fundamental_small_hard} that cover all possible cases that can occur. 

\subsection{Basic properties of potentials}
Here we collect some basic claims about the potentials $\pi, \rho$ and $\sigma$. 
We start with $\pi$. Note that we always have 
\begin{align}
    \label{eq:pi_monotone}
    \pi_i - \E_{i+1}[\pi_{i+1}] \ge \frac{\ell}{k}
\end{align}
by definition of $\pi_i$. 
Note that this is just an inequality. When $K$ becomes solved or covered, we have 
\begin{align}
    \label{eq:pi_drop}
    \pi_i - \pi_{i+1} \ge 10^{50}\ell\log k. 
\end{align}

We continue with $\rho$ whose changes we handle through the following claim. 
The main message of the claim is that whenever we have $R_i \not= R_{i+1}$ and, even more, $|N_{i+1}^{big}| \le (1 - O(1/\ell\log k))|N_i^{big}|$, we have $\rho_i - \rho_{i+1}$ as large as the maximum size of $\sigma_i$.

\begin{claim}
\label{cl:drop_rho}
We have
\begin{align}
\label{eq:willitend}
\rho_{i} - \rho_{i+1} \ge 10^{50}\ell^2 \log^2 k |K| \left( \frac{1}{|N_{i+1}^{small}|}  - \frac{1}{|N_{i}^{small}|} + \frac{1}{|N_{i+1}^{big}|}  - \frac{1}{|N_{i}^{big}|} 
\right)
\end{align}
In particular, we always have
\[
\rho_i - \rho_{i+1} \ge 0. 
\]
Moreover, if we assume that 
\[
|N_{i+1}^{big}| \le (1 - 1/(10^{20}\ell\log k))|N_i^{big}|,
\]
then we have:
\[
\rho_{i} - \rho_{i+1} \ge 10^{30}\ell \log k \frac{|K|}{|N_{i+1}^{small}|}
\]
\end{claim}
\begin{proof}
The first inequality follows from the definition of $\rho$. 

Next, assume that
\begin{equation}
\label{eq:tony_blabbing}
|N_{i+1}^{big}| \le (1 - 1/(10^{20}\ell\log k))|N_i^{big}|.
\end{equation}
We have 
\begin{align*}
&\left( 
\frac{1}{|N_{i+1}^{small}|}  - \frac{1}{|N_{i}^{small}|} + \frac{1}{|N_{i+1}^{big}|}  - \frac{1}{|N_{i}^{big}|} 
\right)\\
&\ge
\left( 
\frac{1}{|N_{i+1}^{small}|}  - \frac{1}{|N_{i}^{big}|} 
\right) && \text{$N_{i+1}^{big} \subseteq N_i^{small}$ by definition of $R_i$}\\
& = \frac{|N_{i}^{big}| - |N_{i+1}^{small}|}{|N_{i}^{big}||N_{i+1}^{small}|}\\
&\ge \frac{\frac{1}{10^{20}\ell\log k} \cdot |N_{i}^{big}|}{|N_{i}^{big}||N_{i+1}^{small}|} && \text{$N_{i+1}^{small} \subseteq N_{i+1}^{big}$, \cref{eq:tony_blabbing}}\\
&\ge \frac{1}{10^{20}\ell\log k} \cdot \frac{1}{|N_{i+1}^{small}|}
\end{align*}
and the claim follows. 

\end{proof}

The idea behind $\sigma$ is that it drops by an amount proportional to the drop in the cost $\phi(X, C_i) - \phi(X, C_{i+1})$. This is substantiated by the following claim. 

\begin{claim}
\label{cl:drop_sigma}
Assume that 
$
10^5k > \frac{400\phi(X, C_i)}{|K|R_i^2}. 
$
Then for any $C_{i+1}$ we have
\[
\sigma_i - \sigmabar \ge 10^{25}\ell\frac{ \phi(X, C_i) - \phi(X, C_{i+1})}{\phi(X, C_i)} \frac{|K|}{|N_i^{small}|}. 
\]
\end{claim}
\begin{proof}
Using the assumption from the statement we get that $\sigma_i = 10^{25}\ell \log\left( \frac{400\phi(X, C_i)}{|K|R_i^2}\right) \frac{|K|}{|N_i^{small}|}$. 
We have
\begin{align*}
    \frac{\sigma_i - \sigmabar}{10^{25}\ell|K|/|N_i^{small}|}
    &=   \left( \log\frac{400\phi(X, C_i)}{|K|R_i^2} - \log\frac{400\phi(X, C_{i+1})}{|K|R_i^2}\right)\\
    &=    \log \frac{\phi(X, C_{i})}{\phi(X, C_{i+1})}\\
    &\ge   \left( 1 - \frac{\phi(X, C_{i+1})}{\phi(X, C_{i})} \right) && \text{\cref{fact:log_cool}} \\
    &=    \frac{\phi(X, C_i) - \phi(X, C_{i+1})}{\phi(X, C_i)}
\end{align*}
as needed. 
\end{proof}

Note however, that $\E[\sigmabar - \sigma_{i+1}]$ is not necessarily positive since $\sigma_{i+1} > \sigmabar$ whenever $R_{i+1} \not= R_i$.  
In these cases, we can bound the difference $\sigma_{i+1} - \sigmabar \le \sigma_{i+1} = O(\ell \log k \cdot \frac{|K|}{|N_{i+1}^{small}|})$. If there is large enough drop in the size of the neighborhood, we have seen in \cref{cl:drop_rho} that the drop in $\rho$ can pay for the negative value of $-\sigma_{i+1}$ that we need to pay to make the left hand side of \cref{eq:fundamental} positive. 

This is the point of the first part of the next claim. The second part argues that if we are in the case $\phi(K, C_i) = O( \phi(X, C_i)/k)$, we can also account for the potentially negative term $\E[\sigmabar - \sigma_{i+1}]$. This time this is because in this special case, the specific value of $R_i$ anyway does not affect the size of $\sigma_i$ which is ``maxed out'' at value $O(\ell \log k \cdot \frac{|K|}{|N_i^{small}|})$. 

\begin{claim}
\label{cl:rho_pays_for_sigma}
\begin{enumerate}
    
\item Assume that  $|N_{i+1}^{big}| \le (1-1/(10^{20}\ell\log k) |N_i^{big}|$. Then, we have
\[
\rho_i - \rho_{i+1} + \sigmabar - \sigma_{i+1} \ge 0.
\]
\item Assume that $10^5k \le \frac{400\phi(X, C_i)}{|K|R_i^2 }$. 
Then, 
\[
\rho_i - \rho_{i+1} + \sigma_i - \sigma_{i+1} \ge 0
.\]
\end{enumerate}

\end{claim}
\begin{proof}
First, assume that $|N_{i+1}^{big}| \le (1-\frac{1}{10^{20}\ell\log k}) |N_i^{big}|$. 
In this case, we use \cref{cl:drop_rho} to get that 
\[
\rho_i - \rho_{i+1}  \ge 10^{30}\ell\log k \frac{|K|}{|N_{i+1}^{small}|}
\]
On the other hand, we certainly have
\[
\sigma_{i+1} \le 10^{25}\ell\log (10^5k) \cdot \frac{|K|}{|N_{i+1}^{small}|}
\]
Hence, we have
\begin{align}
    \rho_i - \rho_{i+1} + \sigmabar - \sigma_{i+1} \ge 0
\end{align}
since we can assume $k \ge 2$ (for $k = 1$ there is not much to prove) and we are done.

Next, assume $10^5k \le \frac{400\phi(X, C_i)}{|K|R_i^2 }$. 
Simplifying the bound \cref{eq:willitend} from \cref{cl:drop_rho} we get
\begin{align}
\label{eq:small1}
\rho_i - \rho_{i+1} 
\ge 10^{50} \ell^2\log^2 k |K| \left( \frac{1}{|N_{i+1}^{small}|} - \frac{1}{|N_{i}^{small}|}\right)
\end{align}

On the other hand, by our assumption we have $\sigma_i = 10^{25}\ell\log (10^5k) \cdot\frac{|K|}{|N_{i}^{small}|}$ and it certainly has to hold that $\sigma_{i+1} \le 10^{25}\ell\log (10^5k) \cdot\frac{|K|}{|N_{i+1}^{small}|}$, hence we get
\begin{align}
\label{eq:small2}
\sigma_{i+1} - \sigma_i
&\le 10^{25}\ell\log (10^5k) \cdot |K| \left( \frac{1}{|N_{i+1}^{small}|} - \frac{1}{|N_{i}^{small}|}\right)
\end{align}
Comparing \cref{eq:small1,eq:small2}, we infer that $\rho_i - \rho_{i+1} + \sigma_i - \sigma_{i+1} \ge 0$, as needed. 

\end{proof}

\subsection{Hard and easy cases}
In this section, we formalize the ``hard and easy case'' from \cref{subsec:hope} and prove the necessary preparatory results for each case. 

At first, we get rid of the special case when $\phi(K, C_i) = O(\phi(X, C_i)/k)$. 

\begin{claim}
\label{cl:fundamental_small_clusters}
Assume that $10^5k \le \frac{400\phi(X, C_i)}{|K|R_i^2 }$. Then \cref{eq:fundamental} is satisfied. 
\end{claim}
\begin{proof}
The condition from the statement in other words means that $\sigma_i$ has the ``maxed out'' value of $10^{25}\ell\log(10^5 \cdot k)\frac{|K|}{|N_{i+1}^{small}|}$.

Using \cref{cl:rho_pays_for_sigma} item (2) we infer that the left hand side of \cref{eq:fundamental} can be lower bounded by
\begin{align}
    \label{eq:blabla}
    (\pi_i - \E[\pi_{i+1}]) + 0 \ge \ell/k. 
\end{align} 

On the other hand, for the right-hand side of \cref{eq:fundamental} we have 

\begin{align}
\label{eq:blabla2}
\frac{\ell\phi(K, C_i)}{\phi(X, C_i)}
&\le \frac{\ell \cdot 2|K|R_i^2}{\phi(X, C_i)} && \text{\cref{cl:basic_properties}} \\
&\le \frac{2 \cdot 400\ell}{10^5 k} && \text{assumption}\nonumber
\end{align}

\cref{eq:blabla,eq:blabla2} imply that \cref{eq:fundamental} holds in this case, as needed.

\end{proof}

In the rest of the proof we only consider the case when
\begin{align}
\label{eq:big_clusters}
 10^5\cdot k > \frac{400\phi(X, C_i)}{|K|R_i^2}   
\end{align}

Consider the probability distribution over $X$ used to sample in the current, $i+1$th, step. That is, consider the probability space where a point $c \in X$ has probability $\phi(c, C_i) / \phi(X, C_i)$. 
Consider the random variable $\delta_i$ on this space that assigns the value $\phi(X, C_i) - \phi(X, C_i \cup \{c\})$ to the sampled point $c$. That is, $\delta_i$ is the random variable measuring the drop in the cost if we sampled just one point of $X$ proportional to its individual cost. 

We define a value $\xi_i$ as the $1/(2\ell)$th quantile of the distribution of $\delta_i$. 
Formally, $\xi_i$ is the largest number such that 
\begin{align}
\label{eq:def_ksi}
\P\left(\phi(X, C_i) - \phi(X, C_i \cup \{c\}) \ge \xi_i\right) \ge \frac{1}{2\ell}
\end{align}

\begin{definition}[Easy and hard clusters]
\label{def:easy_hard}
We say that $K$ is easy with respect to $C_i$ (or in the $i+1$th step) if and only if
\begin{align}
\xi_i < \frac{\phi(N_i^{small}, C_i)}{1500} .   
\end{align}
Otherwise, $K$ is hard. 
\end{definition}

The argumentation for the easy and hard cases differs. We will next prove \cref{cl:easy_case_is_nice} that we rely on in the easy case and \cref{cl:hard_case_is_nice} that we rely on in the hard case. 

\begin{claim}
\label{cl:dropping}
Any point $c \in N_i^{small}$ has the property that $\phi(X, C_i) - \phi(X, C_i\cup \{c\}) \ge \frac{\phi(N_i^{small}, C_i)}{1500}$. 
\end{claim}
\begin{proof}
Let us fix some $c \in N_i^{small}$. Consider any point $x \in N_i^{small}$. 
By \cref{cl:basic_properties}, we have $\phi(x, C_i) \ge R_i^2/200$. 
On the other hand, we have
\[
d(c, x) \le d(c, \mu(K)) + d(\mu(K), x)
\le 2\cdot R_i/100.
\]
Hence, we get
\begin{align}
    \label{eq:tony_cries}
    \phi(x, C_i) - \phi(x, C_i \cup \{c\})
    &\ge \frac{R_i^2}{200} - \frac{4R_i^2}{10^4}
    \ge \frac{R_i^2}{500}
\end{align}
and summing up \cref{eq:tony_cries} for all $x \in N_i^{small}$, we get
\begin{align}
    \phi(X, C_i) - \phi(X, C_i \cup \{c\})
    &\ge \phi(N_i^{small}, C_i) - \phi(N_i^{small}, C_i \cup \{c\})\\
    &\ge |N_i^{small}|\cdot \frac{R_i^2}{500} && \text{\cref{eq:tony_cries}}\\
    &\ge \frac{\phi(N_i^{small}, C_i)}{1500} && \text{\cref{cl:basic_properties} item 2}
\end{align}
\end{proof}

\begin{claim}[Claim for the easy case]
\label{cl:easy_case_is_nice}
Assume that $K$ is easy in the $i+1$th step. Then, for any $x \in N_i^{small}$ we have that $c_{i+1} = x$ with probability at least $\frac{\ell\phi(x, C_i)}{2\phi(X, C_i)}$. 
In particular, this implies: 
\begin{enumerate}
    \item $c_{i+1} \in N_{i}^{small}$ with probability at least $\frac{\ell\phi(N_i^{small}, C_i)}{2\phi(X, C_i)}$,
    \item $c_{i+1} \in K$ with probability at least $\frac{\ell\phi(K, C_i)}{2000\phi(X, C_i)}$.
\end{enumerate}
\end{claim}
\begin{proof}

Fix any $x \in N_i^{small}$. 
For any $1 \le j \le \ell$ consider the following event $E_j$. 

Event $E_j$: We have $c_{i+1}^j = x$. Moreover, for every $1 \le j' \le \ell$ with $j'\not= j$, we have $\phi(X, C_i) - \phi(X, C_i\cup \{c_{j'}\}) \le \xi_i$. 

By independence of all $\ell$ samples of candidate centers and definition of $\xi_i$, we have that
\begin{align}
    \P(E_j) 
    &\ge \frac{\phi(x, C_i)}{\phi(X, C_i)} \cdot \left( 1 - 1/(2\ell)\right)^{\ell - 1}\\
    &\ge \frac{\phi(x, C_i)}{\phi(X, C_i)} \cdot \left( 1 - \frac{\ell-1}{2\ell}\right) && \text{union bound}\\
    &\ge \frac{\phi(x, C_i)}{2\phi(X, C_i)} 
\end{align}

Note that since $K$ is easy, we have $\xi_i < \frac{\phi(N_i^{small}, C_i)}{1500}$. Hence, we  apply \cref{cl:dropping} to conclude that the event $E_j$ implies that $c_{i+1} = x$. 
The upper bound from $K$ being easy also implies that all events $E_j$ are disjoint for different $j$. Thus we get
\begin{align}
    \P(c_{i+1} = x) 
    \ge \sum_{j = 1}^\ell \P(E_j)
    \ge \frac{\ell \phi(x, C_i)}{2\phi(X, C_i)}
\end{align}
as needed. 

Next, we prove the second part of the claim. The first bullet point is proven by summing up over all points $x \in N_i^{small}$:
\[
\P(c_{i+1} \in N_i^{small}) 
\ge \sum_{x \in N_i^{small}} \frac{\ell\phi(x, C_i)}{2\phi(X, C_i)}
= \frac{\ell \phi(N_i^{small}, C_i)}{2\phi(X, C_i)} 
\]
Similarly, using \cref{cl:basic_properties} item 4, we conclude that
\[
\P(c_{i+1} \in K) \ge \P(c_{i+1} \in K \cap N_i^{small}) 
\ge \sum_{x \in K \cap N_i^{small}} \frac{\ell\phi(x, C_i)}{2\phi(X, C_i)}
\ge \frac{\ell \phi(K, C_i)}{2000\phi(X, C_i)}.  
\]

\end{proof}

\begin{claim}[Claim for the hard case]
\label{cl:hard_case_is_nice}
Assume $K$ is hard and 
\[
10^5k > \frac{400\phi(X, C_i)}{|K|R_i^2}. 
\]

Then,
\[
\phi(X, C_i) - \E_{i+1}[\phi(X, C_{i+1})] 
\ge \frac{\phi(N_i^{small}, C_i)}{3000}
\ge \frac{\phi(K, C_i)}{10^7}
\]
and
\[
\sigma_i - \E_{i+1}[\sigmabar] \ge 10^{15} \frac{\ell\phi(K, C_i)}{\phi(X, C_i)}. 
\]
\end{claim}
\begin{proof}
Note that by the definition of $\xi_i$ as the $1-1/(2\ell)$th quantile, the probability that $\phi(X, C_i) - \phi(X, C_{i+1}) < \xi_i$ is at most $(1 - 1/(2\ell))^\ell \le 1/3$. Hence, with probability at least $2/3$ we have  $\phi(X, C_i) - \phi(X, C_{i+1}) \ge \xi_i$ and this implies that 
\[
\phi(X, C_i) - \E_{i+1}[\phi(X, C_{i+1}) ] \ge 2\xi_i/3.
\]
Plugging in that $K$ is hard (\cref{def:easy_hard}), we get 
\begin{align*}
\phi(X, C_i) - \E_{i+1}\phi(X, C_{i+1}) 
&\ge \frac{\phi(N_i^{small}, C_i)}{3000}\\
&\ge \frac{\phi(K, C_i)}{10^7} && \text{\cref{cl:basic_properties} item (4)}.
\end{align*}

Using \cref{cl:drop_sigma}, this implies
\begin{align*}
    \sigma_i -  \E_{i+1}[\sigmabar]
    &\ge  10^{25}\frac{\ell(\phi(X, C_i) - \E_{i+1}[\phi(X, C_{i+1})])}{\phi(X, C_i)}\cdot \frac{|K|}{|N_i^{small}|} && \text{\cref{cl:drop_sigma}}\\
    &\ge10^{25} \frac{\ell\phi(N_i^{small},C_i)}{3000\phi(X, C_i)}\cdot \frac{|K|}{|N_i^{small}|}\\
    &\ge 10^{20} \frac{\ell\phi(N_i^{small},C_i)}{\phi(X, C_i)} \cdot \frac{\phi(K, C_i)/(2R_i^2)}{200\phi(N_i^{small}, C_i)/R_i^2} && \text{\cref{cl:basic_properties}}\\
    &\ge 10^{15} \frac{\ell\phi(K, C_i)}{\phi(X, C_i)}
\end{align*}
as needed. 
\end{proof}

\subsection{Finishing the analysis}
We are now ready to do a case distinction where for each case we combine results from the previous section to verify \cref{eq:fundamental}.

We will first assume that
\begin{equation}
\label{eq:big_annulus_case}
 |N_i^{big} \setminus N_i^{small}| \ge \frac{1}{10^{20}\ell\log k} |N_i^{big}|   
\end{equation}

Intuitively, in this case, we are happy since there are many points in $N_i^{big}$ that will not be present in $N_{i+1}^{big}$. This implies a large drop in the potential $\rho$ via \cref{cl:drop_rho} that can pay for everything. 

\begin{claim}
\label{cl:fundamental_big}
Assume that $10^5k > \frac{400\phi(X, C_i)}{|K|R_i^2}$ and $ |N_i^{big} \setminus N_i^{small}| \ge \frac{1}{10^{20}\ell\log k} |N_i^{big}|  $.
Then \cref{eq:fundamental} is satisfied. 
\end{claim}
\begin{proof}

First, assume that $K$ is easy. 
Note that if $R_i \not= R_{i+1}$ then we can use the fact that $N_{i+1}^{big} \subseteq N_i^{small}$ and our assumption to conclude that \begin{align}
|N_{i+1}^{big}|\le (1 - 1/(10^{20}\ell\log k))|N_i^{big}|
\label{eq:tatkaprisel}
\end{align}
Hence, we may apply the first item in \cref{cl:rho_pays_for_sigma} and get that
\begin{equation}
    \label{eq:luunch}
    \rho_i-\rho_{i+1} + \sigma_i - \sigma_{i+1} \ge 0
\end{equation}
When $R_i = R_{i+1}$, the same equation holds since by definition $\sigmabar \le \sigma_i$. 

That is, the sum of all potentials always drops. 
We write $*$ for the event that $c_{i+1} \in K$ and compute that

\begin{align*}
&\left( \pi_i - \E[\pi_{i+1}]\right) + \left(\rho_i - \E[\rho_{i+1}]\right) + \left( \sigma_i - \E[\sigma_{i+1}] \right) \\
&= \P(*) \left( \pi_i - \E[\pi_{i+1}|*]\right) + \P(\neg *) \left( \pi_i - \E[\pi_{i+1}|\neg *]\right) \\
&+ \left(\rho_i - \E[\rho_{i+1}]\right) + \left( \sigma_i - \E[\sigma_{i+1}] \right) \\
&\ge \P(*) \cdot 10^{50}\ell\log k + 0 &&  \text{\cref{eq:pi_drop,eq:luunch}}\\
&\ge \frac{\ell \phi(K, C_i)}{2000\phi(X, C_i)}\cdot 10^{50}\ell\log k && \text{\cref{cl:easy_case_is_nice}}\\
&\ge \frac{\ell\phi(K, C_i)}{\phi(X, C_i)}.
\end{align*}
That is, \cref{eq:fundamental} is satisfied. 

Next, assume $K$ is hard. Then, we use \cref{cl:hard_case_is_nice} to get 
\[
\sigma_i - \E_{i+1}[\sigmabar]
\ge 10^{15} \frac{\ell \phi(K, C_i)}{\phi(X, C_i)}.
\]

Next,  whenever $R_i \not= R_{i+1}$, we have necessarily $|N_{i+1}^{big}| \le (1-1/(10^{20}\ell\log k)) |N_i^{big}|$ by \cref{eq:tatkaprisel} and
using \cref{cl:rho_pays_for_sigma} item (2) we conclude that
\[
\rho_i - \rho_{i+1} + \sigmabar - \sigma_{i+1} \ge 0
\]
If $R_i = R_{i+1}$, above equation is also clearly satisfied since in that case $\sigmabar = \sigma_{i+1}$. 

In view of the above reasoning, we get
\begin{align*}
&\left(\pi_i - \E[\pi_{i+1}]\right) + \left( \sigma_i - \E[\sigmabar] \right)  + \left(\rho_i - \E[\rho_{i+1}]\right) + \left( \E[\sigmabar - \sigma_{i+1}] \right) \\
&\ge 0 + 10^{15}\frac{\ell \phi(K, C_i)}{\phi(X, C_i)} + 0\\
\end{align*}
and \cref{eq:fundamental} is proven. 
\end{proof}

It remains to argue about the case when
\begin{equation}
\label{eq:small_annulus_case}
 |N_i^{big} \setminus N_i^{small}| < \frac{1}{10^{20}\ell\log k} |N_i^{big}|   
\end{equation}

In this case, we have that the two sets $N_i^{big}$ and $N_i^{small}$ are basically the same. 
This allows us to carry out the planned argument as promised in \cref{subsec:hope}.

Let us define $M \subseteq N_i^{big}$ as the set of $|N_i^{big}|/(10^{20}\ell \log k)$ points of $M$ of maximum distance to $\mu(K)$. 

We observe that whenever $c_{i+1} \in N_i^{big} \setminus M$, then for each $m \in M$ we have $m \not \in N_{i+1}^{big}$. 
This means that $c_{i+1} \in N_i^{big} \setminus M$ implies that 
\begin{align}
\label{eq:drop_in_N}
|N_{i+1}^{big}| \le (1-1/(10^{20}\ell\log k)) \cdot |N_i^{big}|. 
\end{align}
Also, since each point $x \in M$ satisfies $\phi(x,C_i) \le 3R_i^2$ by \cref{cl:basic_properties}, item (1), we infer 
\[
\phi(M, C_i) \le \frac{|N_i^{big}|}{10^{20}\ell\log k} \cdot 3R_i^2.
\]
On the other hand, by \cref{cl:basic_properties}, item (2) and the fact that $|N_i^{small}|\ge|N_i^{big}|/2$ by \cref{eq:small_annulus_case}, we have 
\[
\phi(N_i^{small}, C_i)
\ge |N_i^{small}|R_i^2/200
\ge |N_i^{big}|R_i^2/400. 
\]
Putting these two facts together, we get

\begin{align}
\label{eq:Msmall}
\phi(M, C_i) \le  \frac{3R_i^2 \cdot 400\phi(N_i^{small}, C_i)}{10^{20}\ell \log k \cdot R_i^2} \le \frac{\phi(N_i^{small}, C_i)}{10^{16}\ell \log k}. 
\end{align}

Since by \cref{eq:small_annulus_case} we have $M \supseteq N_i^{big} \setminus N_i^{small}$, we can write
\begin{align}
\label{eq:NeqN}
\phi(N_i^{big}, C_i) \le \phi(N_i^{small}, C_i) + \phi(M, C_i) \le 3\phi(N_i^{small}, C_i)/2
\end{align}

Hence, we have two results \cref{eq:small_annulus_case} and \cref{eq:NeqN} that both formalize the intuition that we do not really need to distinguish between  $N_i^{big}$ and $N_i^{small}$.  
We now consider the easy and the hard case separately and finish the analysis in the following two claims. 

\begin{claim}
\label{cl:fundamental_small_easy}
Assume that $10^5k > \frac{400\phi(X, C_i)}{|K|R_i^2}$ and $ |N_i^{big} \setminus N_i^{small}| < \frac{1}{10^{20}\ell\log k} |N_i^{big}|  $. 
Moreover, assume $K$ is easy. 
Then \cref{eq:fundamental} is satisfied.
\end{claim}

\begin{proof}
Using \cref{cl:easy_case_is_nice} for the set $N_i^{big} \setminus M = N_i^{small} \setminus M$, we infer that $c_{i+1} \in N_i^{big} \setminus M$ with probability at least 
\[
\frac{\ell \phi(N_i^{big} \setminus M, C_i)}{2\phi(X, C_i)} 
\ge \frac{\ell \phi(N_i^{big}, C_i)}{4\phi(X, C_i)}
\]
where we used \cref{eq:Msmall}. 

This implies that with probability at least $\frac{\ell \phi(N_i^{big}, C_i)}{4\phi(X, C_i)}$ we have $\pi_{i+1} = 0$ and using \cref{eq:pi_drop}, we get
\[
\pi_i - \E[\pi_{i+1}] 
\ge \frac{\ell \phi(N_i^{small}, C_i)}{4\phi(X, C_i)}\cdot 10^{50}\ell\log k
\]
Using \cref{eq:NeqN}, we infer that
\[
\pi_i - \E[\pi_{i+1}] 
\ge \frac{\ell \phi(N_i^{big}, C_i)}{10\phi(X, C_i)}\cdot 10^{50}\ell\log k
\]

Next, note that $\sigma_{i+1} \ge \sigma_i$ only in the case when $R_i \not= R_{i+1}$ and in that case we have 
\[
\sigma_{i+1} 
\le 10^{25}\ell\log (10^5k)\frac{|K|}{|N_{i+1}^{small}|} 
\le 10^{25}\ell\log (10^5k). 
\]
We have $\P(R_i \not= R_{i+1}) \le \P(c_{i+1} \in N_i^{big}) \le \frac{\ell\phi(N_i^{big}, C_i)}{\phi(X, C_i)}$. Hence, 
\begin{align*}
\sigma_i - \E[\sigma_{i+1}] \ge \frac{\ell\phi(N_i^{big}, C_i)}{\phi(X, C_i)} \cdot (-10^{25}\ell\log (10^5k))
\end{align*}

This implies 
\begin{align*}
    &\pi_i - \E[\pi_{i+1}] + \rho_i- \E[\rho_{i+1}] + \sigma_i - \E[\sigma_{i+1}] \\
    &\ge \frac{\phi(N_i^{big}, C_i)}{\phi(X, C_i)}\cdot 10^{49}\ell^2\log k + 0 -\frac{\phi(N_i^{big}, C_i)}{\phi(X, C_i)} \cdot 10^{25} \ell^2 \log(10^5 k) \\
    &\ge 10^{40} \frac{\phi(N_i^{big}, C_i)}{\phi(X, C_i)}\\
    &\ge \frac{\phi(K, C_i)}{\phi(X, C_i)} && \text{\cref{cl:basic_properties} item 4}
\end{align*}
as needed. 

\end{proof}

\begin{claim}
\label{cl:fundamental_small_hard}
Assume that $10^5k > \frac{400\phi(X, C_i)}{|K|R_i^2}$ and $ |N_i^{big} \setminus N_i^{small}| < \frac{1}{10^{20}\ell\log k} |N_i^{big}|  $. 
Moreover, assume $K$ is hard. 
Then \cref{eq:fundamental} is satisfied.
\end{claim}
\begin{proof}
We will lower bound the terms $\sigma_i  - \E[\sigmabar]$, $\E[\sigmabar - \sigma_{i+1}]$, and $\rho_i - \E[\rho_{i+1}]$.  

First, recall that \cref{cl:hard_case_is_nice} imply
\begin{align}
\label{eq:part1}
\sigma_i - \E_{i+1}[\sigmabar] 
\ge 10^{15}\ell\phi(K, C_i)/\phi(X, C_i). 
\end{align}
Next, we have that $\sigma_{i+1} > \sigmabar$ only when $R_i \not= R_{i+1}$, which happens only when $c_{i+1}\in N_i^{big}$, otherwise we have $\sigma_{i+1} = \sigmabar$. Also, we have $\sigma_{i+1} \le 10^{25}\ell\log (10^5k) \frac{|K|}{|N_{i+1}^{small}|}$, hence we get
\begin{align}
\label{eq:mamka}
    \E[\sigmabar - \sigma_{i+1}] \ge \P(c_{i+1} \in N_i^{big}) \cdot \left( -10^{25}\ell\frac{|K|}{|N_{i+1}^{small}|} \log (10^5k)\right). 
\end{align}
We rewrite the right hand side as follows. First, note that
\begin{align*}
\P(c_{i+1} \in N_i^{big}) =  \P(c_{i+1} \in M) + \P(c_{i+1} \in N_i^{big} \setminus M) . 
\end{align*}
We bound the first term as follows:
\begin{align*}
 \P(c_{i+1} \in M)
&\le  \P(\exists j : c_{i+1}^j \in M)\\
&\le \frac{\ell \phi(M, C_i)}{\phi(X, C_i)} \\
&\le \frac{\ell \phi(N_i^{small}, C_i)}{10^{16}\ell\log k \cdot \phi(X, C_i)} && \text{\cref{eq:Msmall}}
\end{align*}
Thus we can continue bounding one part of the right hand side of  \cref{eq:mamka} as
\begin{align}
    &\P(c_{i+1} \in M) \cdot 10^{25}\ell \frac{|K|}{|N_{i+1}^{small}|}\log (10^5k)\\
     &\le \frac{\phi(N_i^{small}, C_i)}{10^{16}\log k \cdot \phi(X, C_i)} \cdot 10^{25}\ell \frac{|K|}{|N_{i}^{small}|}\log (10^5k)\\
     &\le \frac{\phi(N_i^{small}, C_i)}{10^{16}\log k \cdot \phi(X, C_i)} \cdot 10^{25}\ell \frac{400\phi(K, C_i)/R_i^2}{\phi(N_i^{small}, C_i)/(3R_i^2)}\log (10^5k) && \text{\cref{eq:NeqN,cl:basic_properties}}\\
     &\le 10^{14} \frac{\ell \phi(K, C_i)}{\phi(X, C_i)}
\end{align}

Putting all this together, we get
\begin{align}
\label{eq:part2}
\E[\sigmabar - \sigma_{i+1}] 
&\ge -10^{14} \frac{\ell \phi(K, C_i)}{\phi(X, C_i)} - \P(c_{i+1} \in N_i^{big}\setminus M) \cdot \left(10^{25}\ell\frac{|K|}{|N_{i+1}^{small}|} \log (10^5k)\right). 
\end{align}

Finally, we bound $\rho_i - \E[\rho_{i+1}]$. 
Using \cref{cl:drop_rho} and the fact that if we sample from $M$, we have $|N_{i+1}^{big}| \le |N_i^{big}|- |M| \le (1-1/(10^{20}\ell\log k))|N_i^{big}|$, we get 
\begin{align}
\rho_i - \E[\rho_{i+1}] 
&\ge \P(c_{i+1} \in N_i^{big} \setminus M) \cdot 10^{30}\ell\log k \cdot \frac{|K|}{|N_{i+1}^{small}|} + \P(c_{i+1} \in M)\cdot 0\\
&=\P(c_{i+1} \in N_i^{big} \setminus M) \cdot 10^{30}\ell\log k \cdot \frac{|K|}{|N_{i+1}^{small}|}\label{eq:part3}
\end{align}

Putting \cref{eq:part1,eq:part2,eq:part3} together, we get
\begin{align*}
    \pi_i &- \E[\pi_{i+1}] + \rho_i - \E[\rho_{i+1}] + \sigma_i - \E[\sigmabar] + \E[\sigmabar - \sigma_{i+1}]\\ 
    &\ge 0 +  \P(c_{i+1} \in N_i^{big} \setminus M) \cdot 10^{30} \ell\log k \cdot \frac{|K|}{|N_{i+1}^{small}|} +  10^{15}\frac{\ell \phi(K, C_i)}{\phi(X, C_i)} \\
    & -10^{14} \frac{\ell \phi(K, C_i)}{\phi(X, C_i)} - \P(c_{i+1} \in N_i^{big}\setminus M) \cdot \left(10^{25}\ell\frac{|K|}{|N_{i+1}^{small}|} \log (10^5k)\right)\\
    &\ge \frac{\ell\phi(K, C_i)}{\phi(X, C_i)}
\end{align*}
and \cref{eq:fundamental} is proven.

\end{proof}

The proof of \cref{lem:greedy_lemma_induction} now follows from \cref{cl:fundamental_small_clusters,cl:fundamental_big,cl:fundamental_small_easy,cl:fundamental_small_hard} that cover all possible cases.

\section{Analysis of greedy $k$-means++}
\label{sec:greedy_ub}

In this section, we prove \cref{thm:greedy_ub} that we restate here for convenience. 
The proof relies on \cref{lem:greedy_lemma} proven in \cref{sec:greedy_lemma}. 

\greedyUb*

We prove the theorem formally by a potential argument. We set up a potential in \cref{def:main_potential} and track it during the algorithm. We prove in \cref{prop:begin} that at the beginning the size of the potential is $O(\ell^3\log^3 k) \cdot OPT$. At the end of the algorithm, the potential is at least as large as the cost of the final solution as proved in \cref{prop:end}. Finally, in \cref{prop:monotone} we prove that we expect the potential only to decrease in between two steps of the algorithm. Together, these results prove \cref{thm:greedy_ub}. 

\subsection{The potential and the intuition behind it}
\label{subsec:potential}

In the rest of the section, we prove \cref{thm:greedy_ub}. 
As in the original proof of \cite{arthur2007k}, we introduce a potential function that assigns each optimal cluster some potential. 

Before we introduce it, recall \cref{def:hit} where $\hit_{i+1}^j(K)$ is defined as an indicator for whether in $i+1$th step $K$ is uncovered and unsolved and $c_{i+1}^j \in K$. We also have $\hit_{i+1} = \sum_{j = 1}^\ell \hit_{i+1}^j$ and $\hit_{\ge i+1} = \sum_{\iota = i+1}^k \hit_{\iota}$. Also, let $b_i$ be the number of bad steps so far where a step $i+1$ is bad whenever $c_{i+1}$ is a point of a cluster covered with respect to $C_i$. In the definition of $\Phi_i$ we condition on the randomness of the first $i$ steps of the algorithm which makes values like $b_i$ deterministic. 

We also use the following notation: $\fK$ is the set of all clusters of a fixed optimal solution; we have $\fK = \fK_i^\fU \sqcup \fK_i^\fC$, i.e., we split the clusters to uncovered and covered with respect to $C_i$. We have $X_i^\fU = \bigcup_{x \in K \in \fK_i^\fU} x$, i.e., $X_i^\fU$ is the set of points in uncovered clusters, we have $X_i^\fC = X \setminus X_i^\fU$. Finally, we split the uncovered clusters into unsolved and solved. Formally, $\fK_i^\fU = \fK_i^{\fU\fU} \sqcup \fK_i^{\fU\fS}$ and $X_i^\fU = X_i^{\fU\fU} \sqcup X_i^{\fU\fS}$. 
We do not use the notation $u_i$ for the number of uncovered clusters as in \cref{sec:intuitive} since this value is exactly equal to $k -i + b_i$. 

We choose our potential as follows: 
\begin{definition}
\label{def:main_potential}
Fix a step $i$ of \cref{alg:kmpp_adversary}. We define a potential $\Phi_i$ as follows. 
\begin{align}
\Phi_i 
&= \Phi_i^1 + \Phi_i^2 + \Phi_i^3\\
&= 10^{10}\ell \left( 1 + H_{k-i}\right) \cdot \phi(X_i^\fC, C_i) \\
&+ 10^{20}\ell \sum_{K \in \fK_i^\fU}  \left( 1 + \E_{\ge i+1}[\hit_{\ge i+1}(K)]\right) \cdot \left( 1 + H_{k-i}\right) \cdot \phi^*(K)\\
&+ b_i \cdot \frac{\phi(X, C_i)}{k - i + b_i}  \end{align}
\end{definition}

The intuition behind the potential is as follows.
The potential function is very similar to the potential of \cite{arthur2007k} although our analysis is more complicated. 
Let us walk through the three terms $\Phi_i^1, \Phi_i^2, \Phi_i^3$ of the potential and explain the intuition behind each of them. 

The first term of the potential, $\Phi_i^1$, can be thought of as follows: every covered cluster $K$ has potential proportional to $(1+H_{k-i})\phi(K, C_i)$. In the end, the cluster needs to have potential $\phi(K, C_i)$ to pay for itself, so it already has a surplus of $H_{k-i} \cdot \phi(K, C_i)$ of potential. 
This means that in the $i+1$-th step, each covered cluster can ``pay'' a cost of $\phi(K, C_i)/(k-i)$. We use this cost to pay for the fact that $i+1$th step can be bad; formally, in that case, $\Phi_i^3$ increases and we pay for that increase by the decrease in $\Phi_i^1$. 

The second term of the potential, $\Phi_i^2$, has the following intuition. At the beginning, every (uncovered) cluster gets potential proportional to $\E_{\ge 1}[\hit_{\ge 1}(K)] \cdot (1 + H_{k-i}) \cdot \phi^*(K)$. 
In the original analysis of \cite{arthur2007k} it would be only $(1 + H_{k-i}) \cdot 5\phi^*(K)$  and the aim of the potential would be that if we at some point sample from $K$, we use the $5$ approximation result of \cref{lem:5apx} to argue that, in expectation, we can now change $5\phi^*(K)$ for $\phi(K, C_{i+1})$, which would make the potential of the newly covered cluster proportional to $(1+H_{k-i})\cdot \phi(K, C_{i+1})$ which is exactly the potential that every covered cluster is supposed to have. 

In our analysis, the additional term $\E_{\ge 1}[\hit_{\ge 1}(K)]$ allows $K$ to ``pay'' the cost $(1+H_{k-i})\phi(K, C_i \cup \{c_{i+1}^j\})$ whenever some candidate center $c_{i+1}^j$ happens to be sampled from $K$. 
If $c_{i+1}^j = c_{i+1}$, we use the paid cost to give $K$ enough potential as it is required being now covered. If $c_{i+1}^j \not= c_{i+1}$, this part of the potential that $K$ ``paid'' is still subtracted from the potential of $K$ although it remains uncovered. 

One additional subtlety is that we replace $H_k$ by $H_{k-i}$ in the potential of every uncovered cluster. 
This allows us to argue that every solved uncovered cluster, i.e., every uncovered cluster with $\Phi(K, C_i) \le 10^5\phi^*(K)$ also pays the cost proportional to $\Phi(K, C_i)/(k-i)$ in every round in the same way as uncovered clusters do. We need to use this fact essentially because our random variable $\hit$ is counting hits of a cluster only until it becomes solved. Hence, we need a small separate argument for solved clusters inside the proof. 

Finally, we come to the last part of the potential, $\Phi_i^3$. This part of the potential is paying for the fact that there were some bad steps. In \cite{arthur2007k}, this part of the potential would be equal to $b_i \cdot \frac{\phi(X_i^\fU, C_i)}{k - i + b_i}$ and its meaning would be that it can pay for $b_i$ ``average'' uncovered clusters. In the end, when $i = k$, it simply pays for all uncovered clusters. 
The intuition about the new problems we face here is described in \cref{subsec:return}. In summary, there is a mismatch between the optimization of $\phi(X_{i+1}^\fU, C_{i+1})$ that we wish to optimize for and $\phi(X, C_{i+1})$ that the greedy optimizes for. While this makes the proof substantially more technical, the definition of the potential $\Phi_i^3$ is very similar to that used by \cite{arthur2007k}; the only difference is that we replace the term $\phi(X_i^\fU, C_i)$ by $\phi(X, C_i)$, essentially because the greedy rule optimizes for the latter, not the former expression. 

\subsection{The formal proof}
\label{subsec:formal}

In this section, we give a formal proof of \cref{thm:greedy_ub}. 
It follows from \cref{prop:begin,prop:end,prop:monotone}. 

\begin{proposition}
\label{prop:begin}
We have $\E[\Phi_1] = O(\ell^3 \log^3 k) \cdot OPT$. 
\end{proposition}
\begin{proof}
Let us go through the three parts of $\Phi_1$. There was only one node sampled, hence only one covered cluster. Using \cref{lem:2apx}, we conclude that $\E[\phi(X_1^\fC, C_1)] \le 2OPT$, hence $\E[\Phi_i^1] = O(\ell \log k) \cdot OPT$. 
Next, we use \cref{lem:greedy_lemma} to conclude that $\E[\hit_{\ge 1}(K)] = O(\ell^2 \log^2 k)$ for every cluster $K$, hence $\Phi_1^2 = O(\ell^3 \log^3 k) \cdot OPT$. Finally, $b_1 = 0$ since the first center was certainly picked from an uncovered cluster, hence $\Phi_1^3 = 0$. 
\end{proof}

\begin{proposition}
\label{prop:end}
We have $\Phi_k \ge \phi(X, C_k)$. 
\end{proposition}
\begin{proof}
For $i = k$ we have $b_i / (k - i + b_i) = 1$ and hence $\Phi_k \ge \Phi_k^3 = \phi(X, C_k)$. 
\end{proof}

The main part of our proof of \cref{thm:greedy_ub} is to show that the potential $\Phi_i$ only decreases in expectation. 
We prove it in \cref{prop:monotone} after analyzing all three parts of the potential $\Phi_i^1, \Phi_i^2, \Phi_i^3$.

\begin{proposition}
\label{prop:monotone1}
Fix a step $i \ge 1$.  We have
\begin{align*}
&\Delta\Phi_i^1 = \Phi_i^1  - \E_{i+1}[\Phi^1_{i+1}] \\
&\ge \frac{10^{10}\ell}{k-i} \cdot \phi(X_i^\fC, C_i)\\
    &- 10^{15}\ell \sum_{K \in \fK_i^{\fU\fU}} \frac{\ell \phi(K, C_i)}{\phi(X, C_i)} \cdot \left( 1 + H_{k-i-1}\right) \cdot \phi^*(K)\\
    &- 10^{15}\ell \sum_{K \in \fK_i^{\fU\fS}} \P(c_{i+1} \in K)\cdot \left( 1 + H_{k-i-1}\right) \cdot \phi^*(K)
. 
\end{align*}
\end{proposition}
The intuition behind $\Delta\Phi_i^1$ is as follows. The first part is proportional to $\phi(X^\fC, C_i)/(k-i)$; this is what we are paying for the fact that the $i$-th step can be bad, i.e., the first term will dominate a similar, negative, term in $\Delta\Phi_i^3$. 
The second and the third part of the difference corresponds to the fact that some uncovered clusters can become covered and we need to ensure they have potential proportional to $(1+H_{k-i})\phi(K, C_{i+1})$ on them in this case. 
We argue differently about the solved and unsolved clusters, hence two expressions. They are accounted for by the corresponding drop in the potential $\Delta \Phi_i^2$.

\begin{proof}
We write
\[
\Delta\Phi_i^1 / (10^{10}\ell) = 
 \left(
 \left( 1 + H_{k-i}\right) \cdot \phi(X_i^\fC, C_i) 
\right)
- \E_{i+1} \left[ \left( 1 + H_{k-i-1}\right) \cdot \phi(X_{i+1}^\fC, C_{i+1})  \right]
\]
For every $K \in \fK_i^\fC$ we upper bound the the term $\phi(K, C_{i+1})$ by $\phi(K, C_{i})$ in the above expression. 
However, notice that $X_{i+1}^\fC$ potentially contains one additional newly covered cluster. We can hence write: 
\begin{align*}
    \Delta\Phi_i^1 / (10^{10}\ell)
    &\ge \frac{1}{k-i} \cdot \phi(X_i^\fC, C_i)
    - \sum_{K \in \fK_i^\fU}\P(c_{i+1} \in K) \cdot \left( 1 + H_{k-i-1}\right) \cdot \E_{i+1}[\phi(K, C_{i+1}) | c_{i+1} \in K]
\end{align*}

We split the sum on the right hand side to the summation over $K \in \fK_i^{\fU\fS}$ and $K \in \fK_i^{\fU\fU}$. To bound the first part, consider any $K \in \fK_i^{\fU\fU}$ and write

\begin{align*}
&\P(c_{i+1} \in K) \cdot \E_{i+1}[\phi(K, C_{i+1}) | c_{i+1} \in K]\\
&=\P(c_{i+1} \in K)  \cdot \sum_{j = 1}^{\ell} \P(c_{i+1} = c_{i+1}^j | c_{i+1} \in K) \cdot \E\left[\phi(K, C_i \cup \{c_{i+1}^j\}) | c_{i+1} \in K \wedge c_{i+1} = c_{i+1}^j\right]\\
&= \sum_{j = 1}^{\ell} \P(c_{i+1}^j \in K \wedge c_{i+1} = c_{i+1}^j ) \cdot \E\left[\phi(K, C_i \cup \{c_{i+1}^j\}) | c_{i+1}^j \in K \wedge c_{i+1} = c_{i+1}^j\right]\\
&\le \sum_{j = 1}^{\ell} \P(c_{i+1}^j \in K \wedge c_{i+1} = c_{i+1}^j ) \cdot \E\left[\phi(K, C_i \cup \{c_{i+1}^j\}) | c_{i+1}^j \in K \wedge c_{i+1} = c_{i+1}^j\right]\\
&\ \ \ \ +\P(c_{i+1}^j \in K \wedge c_{i+1} \not= c_{i+1}^j ) \cdot \E\left[\phi(K, C_i \cup \{c_{i+1}^j\}) | c_{i+1}^j \in K \wedge c_{i+1} \not= c_{i+1}^j\right]\\
&= \sum_{j = 1}^{\ell} \P(c_{i+1}^j \in K ) \cdot \E\left[\phi(K, C_i \cup \{c_{i+1}^j\}) | c_{i+1}^j \in K \right]\\
&\le \frac{\ell\phi(K, C_i)}{\phi(X, C_i)} \cdot 5\phi^*(K)\\
\end{align*}
where we used \cref{lem:5apx} in the last inequality.

On the other hand, for each $K \in \fK_i^{\fU \fS}$ we can use the definition of solved clusters to get that
\begin{align*}
&\P(c_{i+1} \in K) \cdot \E_{i+1}\left[\phi(K, C_{i+1}) | c_{i+1} \in K\right]\\
&\le \P(c_{i+1} \in K) \cdot \phi(K, C_i)\\
&\le \P(c_{i+1} \in K) \cdot 10^5\phi^*(K).
\end{align*}

\end{proof}

We continue with $\Phi_i^2$. 
\begin{proposition}
\label{prop:monotone2}
Fix a step $i \ge 1$. We have
\begin{align*}
\Phi_i^2 - \E_{i+1}[\Phi^2_{i+1}] 
&\ge 
10^{20}\ell \sum_{K \in \fK_i^{\fU\fU}}\frac{\ell \phi(K, C_i)}{\phi(X, C_i)} \cdot \left( 1 + H_{k-i-1}\right) \cdot \phi^*(K)\\
&+10^{20}\ell\sum_{K \in \fK_i^{\fU\fS}}\P(c_{i+1} \in K) \cdot  \left( 1 + H_{k-i-1}\right) \cdot \phi^*(K)\\
&+ \frac{10^{10}\ell \phi(X_i^{\fU\fS}, C_i)}{k-i}
\end{align*}
\end{proposition}
The intuition behind $\Delta\Phi_i^2$ is as follows. The first part of the potential is saying that whenever a candidate center $c_{i+1}^j$ hits an unsolved cluster $K$, we can pay the potential for $K$ to become covered. 
The second part is saying that whenever a solved cluster $K$ becomes covered, we can also pay the due potential; this is simple since the cost of $K$ is already small.
These two terms dominate the respective decreases in $\Delta\Phi_i^1$. 
Finally, each solved cluster pays a cost proportional to $\phi^*(K) / (k-i)$ which is proportional to $\phi(K, C_i)/(k-i)$ in every step; this is analogous to the first term of $\Delta\Phi_i^1$.

\begin{proof}

We have
\begin{align}
\label{eq:rybnik}
\Delta\Phi_i^2 / (10^{20}\ell)
=& \sum_{K \in \fK_i^\fU} 
\left( 1 + \E_{\ge i+1}[\hit_{\ge i+1}(K)]\right) \cdot \left(1 + H_{k-i}\right) \cdot \phi^*(K)\\
&- \E_{i+1} \left[ \sum_{K \in \fK_{i+1}^\fU}  \left( 1 + \E_{\ge i+2}[\hit_{\ge i+2}(K)] \right) \cdot  \left(1 + H_{k-i-1}\right) \cdot \phi^*(K) \right]\nonumber
\end{align}

We bound this expression as follows:
\begin{align}
    \Delta\Phi_i^2 / (10^{20}\ell)
&\ge \sum_{K \in \fK_i^{\fU\fU}} \left( \E_{\ge i+1}[\hit_{\ge i+1}(K)] - \E_{i+1}[ \E_{\ge i+2}[\hit_{\ge i+2}(K)] ]\right) \left(1 + H_{k-i-1}\right) \cdot \phi^*(K) \label{eq:fst}\\ 
&+\sum_{K \in \fK_i^{\fU\fS}} \P\left(c_{i+1} \in K \right) \cdot \left(1 + H_{k-i-1}\right) \cdot \phi^*(K) \label{eq:snd}\\ 
&+\frac{1}{k - i} \cdot \phi^*(X_i^\fU)\label{eq:thd}
\end{align}
We did the following. For each cluster that is unsolved in the $i$th step we simply used the fact that $\fK_{i+1}^\fU \subseteq \fK_i^\fU$ and subtracted the two expressions of \cref{eq:rybnik}. For the solved clusters we on the other hand used that with probability $\P(c_{i+1}\in K)$ we have $K \not\in \fK_{i+1}^\fU$. 
Finally, the last term comes from the replacement of $H_{k-i}$ in $\Phi_i^2$ by $H_{k-i-1}$ in $\Phi_{i+1}^2$. 

Comparing with the desired bound from the statement, we see that the second term \cref{eq:snd} in our bound is already what it should be. 
For the third term \cref{eq:thd}, we first use $\phi^*(X_i^\fU) \ge \phi^*(X_i^{\fU\fS})$ and then, by definition of solved clusters, $\phi^*(X_i^{\fU\fS}) \ge \frac{1}{10^5} \phi(X_i^{\fU\fS}, C_i)$. 

It remains to deal with the first term \cref{eq:fst}. 
Consider any cluster $K  \in \fK_i^\fU$ that is also not solved. Then we have $\E_{i+1}[\hit_{i+1}^j(K)] = \frac{\phi(K, C_i)}{\phi(X, C_i)}$ for any $1 \le j\le \ell$ and we can hence compute that
\begin{align*}
    &\E_{\ge i+1}[\hit_{\ge i+1}(K)] - \E_{i+1}\left[ \E_{\ge i+2}[\hit_{\ge i+2}(K)]\right]\\
    &= \E_{\ge i+1}[\hit_{\ge i+1}(K)] - \hit_{\ge i+2}(K)]]\\
    &= \E_{\ge i+1}[\sum_{j = 1}^\ell \hit_{i+1}^j]\\
    &= \frac{\ell\phi(K, C_i)}{\phi(X, C_i)}. 
\end{align*}
which concludes the proof. 
\end{proof}

We finish with the third part of the potential, $\Phi_i^3$. 

\begin{proposition}
\label{prop:monotone3}
Fix a step $i \ge 1$. We have
\begin{align}
\Phi_i^3 - \E_{i+1}[\Phi^3_{i+1}] 
&\ge - \frac{2\ell \phi(X_i^\fC \cup X_i^{\fU\fS}, C_i)}{k - i}\label{eq:bla1}\\
&- 5\sum_{K \in \fK_i^{\fU\fU}} \frac{\phi(K, C_i)}{\phi(X, C_i)} \phi^*(K)\label{eq:bla3}
\end{align}
\end{proposition}
The intuition behind $\Delta\Phi_i^3$ is as follows. 
In the original $k$-means++ analysis, we would here want to prove that $\Delta\Phi_i^3 \ge - \phi(X_i^\fC, C_i)/(k-i)$ where the right-hand side corresponds to the probability of a bad step multiplied by the cost of average uncovered cluster. In our setting, we first lose an additional  $\ell$-factor since the probability of having a bad step is $\ell$ times larger. We also lose a few more error terms as discussed in \cref{subsec:potential}, the important part is that they can be paid for by $\Delta\Phi_i^1 + \Delta\Phi_i^2$.

\begin{proof}
We will bound $\E_{i+1}[\Phi_{i+1}^3]  - \Phi_i^3$ instead of $\Phi_i^3 - \E_{i+1}[\Phi_{i+1}^3]$ to make the relevant terms positive.

At first we note that we surely know that 
\begin{align}
\label{eq:trivial}
\Phi_{i+1}^3 \le (b_i + 1) \frac{\phi(X, C_i)}{k - (i+1) + (b_i+1)}
= (b_i + 1) \frac{\phi(X, C_i)}{k - i + b_i}
\end{align}
This follows by bounding $\phi(X, C_{i+1})\le \phi(X, C_i)$ and noting that the value of $\Phi_{i+1}$ is larger whenever $b_{i+1} = b_i+1$ as opposed to $b_{i+1} = b_i$. 
To see it formally, we note that for any $b_i > 0$ and any $k-(i+1) \ge 0$ the inequality $\frac{b_i}{k-(i+1)+b_i} \le \frac{b_i + 1}{k - (i+1) + (b_i + 1)}$ is equivalent to $\frac{k-(i+1)}{k-(i+1)+b_i} \ge \frac{k-(i+1)}{k-(i+1)+(b_i+1)}$. 

We start by analysing the special case when $\phi(X_i^\fU, C_i) \le \phi(X_i^\fC, C_i)$. 
In this case, we simply use this assumption and \cref{eq:trivial} to bound 
\begin{align}
    &\E_{i+1}[ \Phi_{i+1}^3] - \Phi_{i}^3\\
    &\le (b_i+1) \frac{\phi(X, C_i)}{k - i + b_i}
    - b_i\frac{\phi(X, C_i)}{k - i + b_i}\\
    &= \frac{\phi(X, C_i)}{k - i + b_i}
    \le 2\frac{\phi(X_i^\fC, C_i)}{k - i + b_i}
    \le 2\frac{\phi(X_i^\fC, C_i)}{k - i}
\end{align}
and we are done as this term is dominated by the right hand side of \cref{eq:bla1}. 

Next, we assume 
\begin{align}
    \label{eq:ass}
    \phi(X_i^\fU, C_i) \ge \phi(X_i^\fC, C_i)
\end{align}

We start by writing
\begin{align}
    \label{eq:omg}
    &\E_{i+1}[ \Phi_{i+1}^3 ] - \Phi_i^3\\
    &= \E_{i+1}\left[b_{i+1} \frac{\phi(X, C_{i+1})}{k - (i+1) + b_{i+1}} \right] - b_i \cdot \frac{\phi(X, C_i)}{k - i + b_i}\\
    &\le  
    \P(c_{i+1} \in X_i^\fC \cup X_i^{\fU\fS}) \left(
    (b_i + 1) \frac{\phi(X, C_i)}{k - i + b_i}
    - b_i \frac{\phi(X, C_i)}{k - i + b_i}
    \right)\label{eq:omg11}\\
    &+  
    \P(c_{i+1} \in X_i^{\fU\fU}) \cdot
    \left(
    b_i \frac{\E_{i+1}[ \phi(X, C_{i+1}) | c_{i+1} \in  X_i^{\fU\fU}]}{k - (i+1) + b_i }
    - b_i \frac{\phi(X, C_i)}{k - i + b_i}
    \right)\label{eq:omg31}
\end{align}
That is, we distinguish two cases based on where the center $c_{i+1}$ is picked from. In the first case we pessimistically bound $\Phi_{i+1}^3$ using \cref{eq:trivial}, while in the second case we use the fact that $c_{i+1} \in X_i^{\fU\fU}$ implies that $b_{i+1} = b_i$. 

To bound the first term, i.e. \cref{eq:omg11}, we first use that $$\P(c_{i+1} \in X_i^\fC \cup X_i^{\fU\fS}) \le \P( \exists j : c_{i+1}^j \in X_i^\fC \cup X_i^{\fU\fS}) \le \frac{\ell \phi(X_i^\fC \cup X_i^{\fU\fS}, C_i)}{\phi(X, C_i)}.$$ 
We also have
\[
(b_i + 1) \frac{\phi(X, C_i)}{k - i + b_i}
    - b_i \frac{\phi(X, C_i)}{k - i + b_i}
    \le \frac{\phi(X, C_i)}{k - i}
\]
Hence, the value of \cref{eq:omg11} is at most 
\[
\frac{\ell \phi(X_i^\fC \cup X_i^{\fU\fS}, C_i)}{\phi(X, C_i)} \cdot \frac{\phi(X, C_i)}{k - i}
\le \frac{\ell \phi(X_i^\fC\cup X_i^{\fU\fS}, C_i)}{k - i}
\]
Hence, the first term, corresponding to the case when the step is bad, is conveniently dominated by \cref{eq:bla1}. 

In the rest of the proof, we analyze the second term \cref{eq:omg31} that corresponds to the drift of the size of the average uncovered cluster. 

We start by proving that 
\begin{align}
    \label{eq:jo}
\E_{i+1}[ \phi(X, C_{i+1}) | c_{i+1} \in  X_i^{\fU\fU}]
\le \E_{i+1}[ \phi(X, C_{i} \cup \{c_{i+1}^1\}) | c_{i+1}^1 \in  X_i^{\fU\fU}]
\end{align}
That is, we claim that if we reveal that the center $c_{i+1}$ taken by the greedy rule is from $X_i^{\fU\fU}$, we know that the expected new cost $\phi(X, C_{i+1})$ is smaller than if we simply sampled some candidate center $c_{i+1}^j$ and revealed it is sampled from $X_i^{\fU\fU}$. 
\cref{eq:jo} then allows us to analyze further only the expression on its right-hand side that does not rely anymore on the greedy rule. 

\cref{eq:jo} follows from the fact our rule is greedy; to formally verify it holds, let us first write
\begin{align}
\label{eq:rain}
\E_{i+1}[ \phi(X, C_{i+1}) | c_{i+1} \in  X_i^{\fU\fU}]
= \sum_{I \in \fI} \P(I | c_{i+1} \in  X_i^{\fU\fU}) \cdot \E_{i+1}[\phi(X, C_{i+1}) | I ]
\end{align}
where $\fI$ is the set of all of at most $2^\ell \cdot \ell$ possible following revelations: For each $1 \le j \le \ell$, we reveal whether $c_{i+1}^j \in X_i^{\fU\fU}$, and we also reveal for which index $j_0$ we have $c_{i+1} = c_{i+1}^{j_0}$. 
Notice that on the right hand side of \cref{eq:rain} we used $\E_{i+1}[\phi(X, C_{i+1}) | I ] = \E_{i+1}[\phi(X, C_{i+1}) | I \wedge c_{i+1} \in X_i^{\fU\fU}]$ since $c_{i+1} \in X_i^{\fU\fU}$ is always either implied by $I$ or it is incompatible with it and then $\P(I | c_{i+1}\in X_i^{\fU\fU}) = 0$. 

Fixing any revelation $I$ with $c_{i+1} = c_{i+1}^{j_0}$, we observe that 
\begin{align}
\label{eq:ach}
\E_{i+1}[\phi(X, C_{i+1}) | I ] \le \E_{i+1}[\phi(X, C_{i} \cup \{c_{i+1}^1\}) | c_{i+1}^1 \in X_i^{\fU\fU} ]
\end{align}
To see this, first rewrite the equation equivalently as 
\begin{align} \label{eq:router}
\E_{i+1}[\phi(X, C_{i} \cup \{c_{i+1}^{j_0}\}) | I ]
\le \E_{i+1}\left[\phi(X, C_{i} \cup \{c_{i+1}^{j_0}\}) | c_{i+1}^{j_0} \in X_i^{\fU\fU} \right]. 
\end{align}
Observe that the information $I$ can be viewed as describing distributions from which all candidate centers $c_{i+1}^j$ are sampled from, independently, together with the information that after candidate centers were sampled, it happened that $\phi(X, C_{i} \cup \{c_{i+1}^{j_0}\}) \le \phi(X, C_{i} \cup \{c_{i+1}^{j}\}) $ for any $j\not= j_0$. 
This means that the correctness of \cref{eq:router} follows from \cref{lem:conditioning}. 
Plugging \cref{eq:ach} to \cref{eq:rain} proves \cref{eq:jo}. 
\vspace{1cm}

We now continue with analysing the term $\E_{i+1}[ \phi(X, C_{i} \cup \{c_{i+1}^1\}) | c_{i+1}^1 \in  X_i^{\fU\fU}]$ from \cref{eq:jo} even further. We write:
\begin{align}
&\E_{i+1}[\phi(X, C_{i} \cup \{c_{i+1}^1\}) | c_{i+1}^1 \in X_i^{\fU\fU} ]\\
&\le \phi(X, C_i) -  \sum_{K \in \fK_i^{\fU\fU}} \frac{\phi(K, C_i)}{\phi(X_i^{\fU\fU}, C_i)} \cdot \left( \phi(K, C_i) - \E_{i+1}\left[ \phi\left(K, C_i \cup \{c_{i+1}^1\} | c_{i+1}^1 \in K\right)\right] \right)\\
&\le \phi(X, C_i) -  \sum_{K \in \fK_i^{\fU\fU}} \frac{\phi(K, C_i)}{\phi(X_i^{\fU\fU}, C_i)} \cdot \left( \phi(K, C_i) - 5\phi^*(K)  \right) && \text{\cref{lem:5apx}}\\
&= \phi(X, C_i) 
- \sum_{K \in \fK_i^{\fU\fU}} \frac{\phi^2(K, C_i)}{\phi(X_i^{\fU\fU}, C_i)}
+  \sum_{K \in \fK_i^{\fU\fU}} \frac{\phi(K, C_i)}{\phi(X_i^{\fU\fU}, C_i)} \cdot 5\phi^*(K)\\
&\le \phi(X, C_i) - \frac{\phi(X_i^{\fU\fU})}{|\fK_i^{\fU\fU}|} +  \sum_{K \in \fK_i^{\fU\fU}} \frac{\phi(K, C_i)}{\phi(X_i^{\fU\fU}, C_i)} \cdot 5\phi^*(K)\label{eq:willitend2}
\end{align}
where the last bound follows from the Cauchy-Schwartz inequality (or AK inequality) $\sum_{i = 1}^n x_i^2 \ge \left( \sum_{i = 1}^n x_i\right)^2 / n$. 

It is time to reap the fruits of our work. We plug in the bounds from \cref{eq:jo,eq:willitend2} to the term \cref{eq:omg31} and bound $\P(c_{i+1} \in X_i^{\fU\fU}) \le 1$ there to conclude that
\begin{align}
\label{eq:almost}
    \text{\cref{eq:omg31}}
    &\le \left(
    b_i \frac{\E_{i+1}[ \phi(X, C_{i+1}) | c_{i+1} \in  X_i^{\fU\fU}]}{k - (i+1) + b_i }
    - b_i \frac{\phi(X, C_i)}{k - i + b_i}
    \right)\\
    &\le b_i \cdot \left(
    \frac{\phi(X, C_i) - \frac{\phi(X_i^{\fU\fU})}{|\fK_i^{\fU\fU}|} +  \sum_{K \in \fK_i^{\fU\fU}} \frac{\phi(K, C_i)}{\phi(X_i^{\fU\fU}, C_i)} \cdot 5\phi^*(K)}{k - (i+1) + b_i }
    - \frac{\phi(X, C_i)}{k - i + b_i}
    \right)
\end{align}

This can be further simplified to 
\begin{align}
\label{eq:almost_there}
\text{\cref{eq:omg31}}
    &\le b_i \cdot \left(
    \frac{\phi(X, C_i) - \frac{\phi(X_i^{\fU\fU})}{|\fK_i^{\fU\fU}|} }{k - (i+1) + b_i }
    - \frac{\phi(X, C_i)}{k - i + b_i}
    \right)+  \sum_{K \in \fK_i^{\fU\fU}} \frac{\phi(K, C_i)}{\phi(X_i^{\fU\fU}, C_i)} \cdot 5\phi^*(K)
\end{align}
Note that the last term of the right-hand side is already equal to \cref{eq:bla3} so to finish we need to analyze the first term of the right-hand side. 
We do it as follows. 
\begin{align}
     &b_i \cdot \left(
    \frac{\phi(X, C_i) - \frac{\phi(X_i^{\fU\fU})}{|\fK_i^{\fU\fU}|} }{k - (i+1) + b_i }
    - \frac{\phi(X, C_i)}{k - i + b_i}
    \right)\\
    & \le b_i \cdot \left(
    \frac{\phi(X, C_i) - \frac{\phi(X_i^{\fU\fU})}{k - i + b_i} }{k - (i+1) + b_i }
    - \frac{\phi(X, C_i)}{k - i + b_i}
    \right) \ \ \ \ \ \ \ \ \  \ \ \ \ \ \ \ \ \  \ \ \ \ \ \ \ \ \  \ \ \ \ \ \ \ \ \ |\fK_i^{\fU\fU}| \le |\fK_i^\fU|\\
    & = b_i \cdot \left(
    \frac{\phi(X_i^{\fU\fU}, C_i) - \frac{\phi(X_i^{\fU\fU})}{k - i + b_i} }{k - (i+1) + b_i }
    - \frac{\phi(X_i^{\fU\fU}, C_i)}{k - i + b_i}
    \right) + b_i\phi(X_i^{\fU\fS} \cup X_i^\fC)\left(\frac{1}{k-(i+1)+b_i} - \frac{1}{k-i+b_i}\right)\\
        & =0 +   \frac{b_i\phi(X_i^{\fU\fS} \cup X_i^\fC)}{(k-i+b_i)(k-i+b_i-1)}\\
        &\le \frac{\phi(X_i^{\fU\fS} \cup X_i^\fC)}{k-i}
\ \ \ \ \ \ \ \ \  \ \ \ \ \ \ \ \ \  \ \ \ \ \ \ \ \ \  \ \ \ \ \ \ \ \ \ \ \ \ \ \ \ \ \ \  \ \ \ \ \ \ \ \ \  \ \ \ \ \ \ \ \ \  \ \ \ \ \ \ \ \ \ i \le k-1
\end{align}

Plugging back to \cref{eq:almost_there} and all the way back to \cref{eq:omg} finishes the proof. 

\end{proof}

\begin{proposition}
\label{prop:monotone}
Fix a step $i > 1$. We have
\[
\Phi_i  \ge \E_{i+1}[\Phi_{i+1}]. 
\]
\end{proposition}
\begin{proof}
Putting all bounds of \cref{prop:monotone1,prop:monotone2,prop:monotone3} together, we get
\begin{align}
&\Phi_i - \E_{i+1}[\Phi_{i+1}] = 
  \Delta\Phi_i^1+ \Delta\Phi_i^2 +\Delta\Phi_i^3\\
  &\ge \frac{10^{10}\ell}{k-i} \cdot \phi(X_i^\fC, C_i)\\
    &- 10^{15}\ell \sum_{K \in \fK_i^{\fU\fU}} \frac{\ell \phi(K, C_i)}{\phi(X, C_i)} \cdot \left( 1 + H_{k-i-1}\right) \cdot \phi^*(K)\\
    &- 10^{15}\ell \sum_{K \in \fK_i^{\fU\fS}} \P(c_{i+1} \in K)\cdot \left( 1 + H_{k-i-1}\right) \cdot \phi^*(K)\\
&+
10^{20}\ell \sum_{K \in \fK_i^{\fU\fU}}\frac{\ell \phi(K, C_i)}{\phi(X, C_i)} \cdot \left( 1 + H_{k-i-1}\right) \cdot \phi^*(K)\\
&+10^{20}\ell\sum_{K \in \fK_i^{\fU\fS}}\P(c_{i+1} \in K) \cdot  \left( 1 + H_{k-i-1}\right) \cdot \phi^*(K)\\
&+ \frac{10^{10}\ell \phi(X_i^{\fU\fS}, C_i)}{k-i}\\
&- \frac{2\ell \phi(X_i^\fC \cup X_i^{\fU\fS}, C_i)}{k - i}\\
&- 5\sum_{K \in \fK_i^{\fU\fU}} \frac{\phi(K, C_i)}{\phi(X, C_i)} \phi^*(K)\\
&\ge 0
\end{align}

\end{proof}

\section{A hard instance for greedy $k$-means++}
\label{sec:log3_lb}

In this section we provide a construction of a (weighted) point set where \cref{alg:kmpp_greedy} returns a solution with approximation $\Omega(\frac{\ell^3 \log^3 k}{\log^2 (\ell \log k)})$ with constant probability. Formally, we prove \cref{thm:greedy_lb} that we restate here for convenience. 

\greedyLB*

Recall that we already gave an informal description in \cref{subsec:lb}.
We first describe the point set in  \cref{subsec:lb_description}. 
We then give the formal analysis of greedy $k$-means++ on the point set in \cref{subsec:lb_analysis_formal}.

\subsection{The point set}
\label{subsec:lb_description}

We start by describing the weighted point set $X$. In fact, we define the full input instance  $(X, C_0, \tilde{k})$ where $C_0$ is the starting set of centers (see \cref{lem:prescribed}).

We set $$t = \frac{\ell \log k }{1000\log (\ell \log k)}.$$
In the follow-up discussion, we always assume that $k$ is large enough and the expressions like the one above that defines $t$ are integers. This is for the purpose of readability; it is simple to make the proof work by adding $\lfloor \, \rfloor$ to all definitions that require integer values. 

Recall that in the statement of \cref{thm:greedy_lb} we assume that $\ell < k^{0.1}$; we did not try to optimize this bound but note that there has to be some since for $\ell \gg k$ we have $\ell^2\log^2 k \gg \ell\cdot k$ where the right-hand side is the trivial bound on the number of hits. 
Note that the lower bound $\Omega(\ell \log k)$ of \cite{bhattacharya2019noisy} holds also for large $\ell$, hence for $\ell > k^{0.1}$ the greedy $k$-means++ algorithm already necessarily has bad, polynomial, approximation guarantee. 

We will now describe the point set $X$, the weights of the points, and the distances between some pairs of points. Then, we discuss how exactly we embed the points in the Euclidean space.
We next list points of $X$ (the picture to have in mind is \cref{fig:adversary_lb_intuitive}). 

\begin{enumerate}
    \item There is a point $b$ for which we have $w(b) = \frac{1}{t}$.
    \item A point $c$ is at distance 1 from $b$. We have $w(c) = 1/10$. 
    \item We have a set of points $N = \{n_1, n_2, \dots, n_{t + 1}\}$ and  $M = \{m_1, \dots, m_t\}$ defined as follows. 
    We have $d(b, n_i) = k^i$ and $w(n_i) = w(m_i) = \frac{1000}{t}$. 
    Each $m_i$ lies at distance $10tk^{i}$ from $n_i$. 
    We put the point $n_{t + 1}$ to $C_0$, that is, we assume that point is already sampled at the beginning. 
    \item We have $A = \{a_1, \dots, a_{k^{1.2}}\}$ at distance $k$ from $c$. 
    The weight of each $a_i$ is $\frac{\ell\log k }{k^2}$ so their total weight is $\frac{\ell \log k }{k^{0.8}}$. 
    \item We have $E = \{e_0, \bigcup_{e \in E_1} e, \dots, \bigcup_{e \in E_t} e\}$ where $E_i = \{e_{i,1}, \dots,  e_{i, \sqrt{k}}\}$. 
    Each point in $E \setminus \{e_0\}$ has distance $1$ from a point $e_0$ which is at distance $k^{2t}$ from $b$. We include $e_0$ to $C_0$. 
    Since we also included $n_{t+1} \in C_0$ and we chose $d(b,e_0)$ large enough, it will never happen that a closest center to a point in $E$ is in $X \setminus E$ or vice versa. 
    Each $e_{i,j} \in E_i$ has the same weight $w_i$. This parameter needs to be set up quite precisely depending on the rest of the instance so we define  it only later. 

\end{enumerate}

The number of points in this point set $X$ is equal to $|X| = 1 + 1 + (t+1) + t  + k^{1.2} + (1 + t\cdot \sqrt{k}) = O(k^{1.2})$. Two points of $X$, $n_{t+1}$ and $d$, are already in $C_0$. We choose the number $\tilde{k} = |X| - |C_0| - 1$. 
We will work with the input instance $(X, C_0, \tilde{k})$. 
That is, in this instance, the optimal solution (\cref{lem:wlog_opt_subset}) as well as the solution of greedy $k$-means++ selects as centers all points of $X$, except for one. 

\paragraph{Arrangement}
We next specify fully how to embed the point set $X$ to the Euclidean space. 
So far, we only specified distances of pairs $(n_i, m_i), (b, n_i), (b, c), (c, a_i), (e_0, e_{i,j}), (e_0,b)$ for all $i,j$. These distances define a tree metric that we will simulate. 
Unfortunately, we cannot simulate exactly this tree metric in a Euclidean space, but we can come sufficiently close to it, using \cref{fact:simplex}.

We now describe the configuration. 
In view of \cref{fact:simplex}, vectors $(b,c), (b, n_1), \dots, (b, n_{t+1})$ are chosen as vertices of a $(t+1)$-dimensional simplex. 
Each $m_i$ lies on the ray $(b, n_i)$.  

To give an example how this embedding simulates  the idealized tree metric up to $1/t$ loss, let us verify that $d(n_{i'}, n_i) = d(n_i, b) + \Theta(d(b, n_{i'})/t)$:
\begin{align}
    \label{eq:ijbound2}
    d(n_{i'}, n_i)^2 
    = d(n_{i'}, b)^2 + d(b, n_i)^2 + 2d(n_{i'},b)d(n_i,b)/(t+1) 
    = k^{2{i'}} + k^{2i} + \frac{2}{t+1} k^{i+i'}
\end{align}
where we used the cosine law and \cref{fact:simplex}. 
That is, we have 
\begin{align}
    \label{eq:ijbound}
    d(n_{i'}, n_i) 
    \ge \sqrt{ k^{2i} + 2k^{i+i'}/(t+1)}
    = k^i \sqrt{ 1 + 2k^{i' -i}/(t+1)}
    \ge k^i (1 + k^{i'-i}/(2t))
    = k^i + k^{i'}/(2t)
\end{align}

Similarly, we can get the following bounds
\begin{equation}
    \label{eq:mnbound}
    d(m_{i'}, n_i) \ge k^i + k^{i'}/(2t)
\end{equation}
\begin{align}
    \label{eq:cnbound}
    d(c, n_i) \ge k^i + 1/(2t)
\end{align}
\begin{align}
    \label{eq:cmbound}
    d(c, m_i) \ge (10t+1)k^i + 1/(2t)
\end{align}

Next, we specify the directions of the points $A$. Each $a_i$ goes partly in the direction of the ray $b,c$ and partly in a direction orthogonal to everything else. 
Namely, the ray $(c, a_i)$ has direction $\frac12 (b,c) + \frac12 \nu_i$ where $\nu_i$ is orthogonal to the span of $X \setminus \{a_i\}$. 
Hence, the projection of $a_i$ to $(b,c)$ is at distance $k/\sqrt{2}$ from $c$. Also, for any $i < j$ we can project to the plane defined by $\nu_i, \nu_j$ to see that
\begin{align}
\label{eq:as_far}
    d(a_i, a_j) = k
\end{align}
Finally, using cosine law, we get
\begin{align}
\label{eq:a_drop}
    d(b, a_i)^2
    &= d(c, b)^2 + d(c, a_i)^2 - d(b,c)d(c,a_i)\cos 135^\circ\\
    &= 1 + d(c, a_i)^2 + d(c, a_i)/\sqrt{2}\label{eq:a_drop2}\\
    &\ge d(c, a_i)^2 + d(c, a_i)/2\\
    &= k^2 + k/2
\end{align}

Finally, each vector $(e_0, e_{i,j})$ is orthogonal to the span of $X \setminus \{e_{i,j}\}$. In particular, we have 
\begin{align}
\label{eq:e_sane}
d(e_{i,j}, e_{i', j'}) \ge d(e_0, e_{i,j})
\end{align}
for every $i,j,i',j'$. 

\paragraph{Precise definition of $w_i$}
It remains to define $w_i$. 
We define $S_i = \{b, c, a_j, n_{\le i}, m_{\le i}\}$. Then, we define
    \begin{align}
        \Delta(b) = \phi(S_i, \{n_{i+1} \cup b\}) - \phi(S_i, \{n_{i+1}\})
    \end{align}
    and
    \begin{align}
        \Delta(c) = \phi(S_i, \{n_{i+1} \cup c\}) - \phi(S_i, \{n_{i+1}\})
    \end{align}
    We define 
    \begin{align}
        w_i = \frac{\Delta(b) + \Delta(c)}{2}.
    \end{align}
    That is, $w_i$ is set up so that, under some assumptions about what centers are already taken (e.g. $n_{i+1}$ is but points of $S_i$ are not), the drop resulted by taking $e_{i,j}$ as a center is smaller than the drop when we take $b$ but bigger than the drop when we take $c$ (we are yet to prove that $\Delta(b) > \Delta(c)$). 
    Note that $w_i$ satisfies
    \begin{equation}
        \label{eq:wi}
        k^{2i+2} \le w_i \le  3000 k^{2i+2}
    \end{equation}  
    The lower bound follows from $\phi(c, \{n_{i+1}\}) \ge k^{2i+2}$ using \cref{eq:cnbound}. The upper bound follows from the fact that $w(S_i)$ is dominated by the cost of $N_i$ and $M_i$ and for any $x \in S_i$ we have $d(n_{i+1}, x) < 1.1k^{i+1}$ which follows from looking at \cref{fig:adversary_lb_intuitive}. 

This concludes the description of the point set $X$. 

\begin{remark}
Although our point set is weighted, we can make it unweighted by scaling all the weights up by a sufficiently large number and rounding them to the nearest integer. 

In fact, we believe, but do not prove, that all weights and positions of points in $X$ from \cref{thm:greedy_lb} can be made integers of order $k^{O(\log k)}$. 
We do not attempt a formal proof since that would require tedious arguments about rounding errors. 
    
    We note that in view of the $O(\ell^{O(1)}\log \frac{OPT(1)}{OPT(k)})$ upper bound sketched in \cref{sec:alternative}, the size of point weights and positions cannot be both improved to $k^{O(1)}$. Namely, for constant $\ell$, any instance where \cref{alg:kmpp_greedy} is $\Omega(\log^3 k / \poly\log\log k)$ approximate needs to satisfy 
    \[
    \log \frac{OPT(1)}{OPT(k)} \ge \log^2 k / \poly\log\log k. 
    \]
    Hence, whenever $OPT(k)$ is a positive integer, we get that necessarily
    \[
    OPT(1) = \phi(X, \mu(X)) \ge k^{\log k / \poly\log\log k}. 
    \]
\end{remark}

\subsection{Analysis of greedy $k$-means++ on the hard point set}
\label{subsec:lb_analysis_formal}

In this subsection, we give the formal proof of \cref{thm:greedy_lb}. 

\paragraph{First epoch}

We define the first epoch formally as the first $\tilde{k} - k^{1.2}$ steps of \cref{alg:kmpp_greedy} (cf. \cref{subsec:lb} for the intuition behind the first epoch). This means our aim is to prove the following claim. 

\begin{claim}
\label{cl:first_epoch}
After running \cref{alg:kmpp_greedy} on the instance $(X, \tilde{k}, C_0)$ for $\tilde{k} - k^{1.2}$ steps, with positive probability we have 
\[
C_{\tilde{k} - k^{1.2}} =X \setminus ( A \cup \{c\} )
\]
\end{claim}

We split the epoch into $t$ phases that we, for notational reasons, index in a decreasing order as $i = t, t-1, \dots, 1$. 
Our main task is to prove that in each $i$-th phase the point $b$ is selected as a center with probability $\Omega(1/t)$. 
The $i$-th phase is formally defined as follows. With the exception of the very first phase, it starts when the last point of $E_{i+1}$ is taken as a center. 
Alternatively, we say that a phase finishes whenever $b$ is taken. 

As a first claim, we prove that whenever $b$ is selected as a center, with constant probability, we finish the first epoch as intended. 

\begin{claim}
\label{cl:finish}
Assume that in some step $\iota_0$ during the first phase we have $b \in C_{\iota_0}$ and $(\{c\} \cup A) \cap C_{\iota_0} = \emptyset$. 
Then, with probability at least $1/2$, the first phase finishes with
\[
C_{\tilde{k} - k^{1.2}} =X \setminus ( A \cup \{c\} ).
\]
\end{claim}
\begin{proof}
First, we upper bound the total cost of $\{c\}\cup A$ in every step $\iota \ge \iota_0$ assuming $(\{c\} \cup A) \cap C_{\iota} = \emptyset$. 
We have $\phi(c, b) = 1 \cdot 1^2 = 1$ and $\phi(A, b) = k^{1.2} \cdot \frac{\ell \log k}{k^2} \cdot k^2 = \ell k^{1.2}\log k$. That is, $\phi(\{c\}\cup A, C_\iota) \le \ell k^{1.3}$. 

On the other hand, we claim that any point $x \in N \cup M \cup E$ has always cost $\phi(x, C_\iota) \ge k^2/t$, unless $x \in C_\iota$. 
For the points of $N \cup M$, this is because their weight is $1000/t$ and the distance to the closest other point of $X$ is always at least $k$. For the points of $E$, this is because the distance to the closest already taken point is always $1$ (this is the point $e_0 \in C_0 \subseteq C_\iota$) and the smallest weight of a point in $E$ is at least $k^4$ by \cref{eq:wi}. 

In view of the above computations, we have that the probability we sample a candidate center from $\{c\}\cup A$ in step $\iota$ is at most 
\[
\ell \cdot \frac{\ell k^{1.3}}{(\tilde{k} - k^{1.2} - \iota)\cdot k^2/t} 
\le \frac{1}{(\tilde{k} - k^{1.2} - \iota) \cdot k^{0.4}}. 
\]
Union bounding over all $\tilde{k} < k^2$ step leads to a harmonic series summing up to $O(\log k)/k^{0.4} < 1/2$, as needed. 
\end{proof}

Next, let us analyze one phase of the first epoch.
Recall that the $i$th phase starts after step $\iota$ for which $E_{i+1} \subseteq C_\iota$ and finishes when $E_i \subseteq C_\iota$ or $b \in C_\iota$. 
In the next claim, we compare the cost drops of various points with the ``baseline cost drop'' of taking a point in $E_i$. 

\begin{claim}
\label{cl:cost_drops}
Assume that $N_{\ge i+1}\cup M_{\ge i+2}\cup E_{\ge i+1} \subseteq C_\iota$ while $(N_{\le i} \cup M_{\le i} \cup E_{\le i} \cup A \cup \{b,c\} ) \cap C_\iota = \emptyset$.  
Let $\Delta(x) = \phi(X, C_\iota) - \phi(X, C_\iota \cup \{x\})$ be a cost drop of a point $x \in X$. 
Then, we have for any $i' < i$ and any $j$ we have
\[
\Delta(n_{i'}),\Delta(m_{i'}),\Delta(c) < \Delta(e_{i,j}) < \Delta(b) < \Delta(n_i), \Delta(m_{i+1})
\]
where $e_{i,j}$ is arbitrary point not in $C_\iota$.  
\end{claim}
\begin{proof}

First, we prove that $\Delta(b) > \Delta(c)$. 
For $x \in \{c\}\cup A$ we have $\phi(x, c) \le \phi(x, b)$, whereas for any $x \in \{b\} \cup N_{\le i} \cup M_{\le i+1}$ we have $\phi(x, b) \le \phi(x, c)$. 

So, we first upper bound the drop difference for $\{c\}\cup A$:
\begin{align*}
\phi(c \cup A, b) - \phi(c \cup A, c)
&= 1 + \phi(A, b)  - \phi(A, c)\\
&= 1 + k^{1.2} \cdot \frac{\log k}{k^2} \cdot (1 + d(c, a_1)/\sqrt{2})\\
&= O(k^{0.2}\log k)
\end{align*}
where we used \cref{eq:a_drop2} and $d(c,a_1) = k$. 

On the other hand, consider just the point $n_1$. Using the cosine law, we have
\begin{align*}
    \phi(n_1, b) - \phi(n_1, c)
    &\ge 2d(n_1, b)d(b,c)/t \ge k/t
\end{align*}

That is, the difference in the cost drop at point $n_1$ dominates all other points where $c$ has a larger cost drop than $b$. 
This means that $\Delta(b) > \Delta(c)$ and by definition of $w_i$, we already get for any $j$ that
\[
\Delta(c) < \Delta(e_{i,j}) < \Delta(b). 
\]

Next, consider the point $n_i$. We will prove that $\Delta(n_i) > \Delta(b)$. The intuitive reason for this is that $m_i$ is sufficiently far away to make the drop difference between $n_i$ and $b$ there dominate the other terms. Concretely, we have 
\begin{align}
\label{eq:boring}
    \phi(m_i, b) - \phi(m_i, n_i) 
    = 1000/t \cdot \left( ((10t+1)k^i)^2 - (10tk^i)^2 \right)
    \ge 1000/t \cdot 20tk^{2i}
    \ge 20000k^{2i}
\end{align}

On the other hand, consider the set $T = \{b,c\} \cup A \cup N_{<i} \cup M_{<i}$ of points $x$ for which $\phi(x, n_i) \ge \phi(x, b)$ (note that $\phi(m_{i+1}, C_\iota \cup \{b\}) = \phi(m_{i+1}, C_\iota \cup \{c\}) = \phi(m_{i+1}, n_{i+1})$, that is, $m_{i+1}$ is not enjoying any cost drop). Using the fact that the largest distance $d(x, n_i)$ for $x \in T$ is $10tk^{i-1}$ for $x = m_{i-1}$, and the fact that $w(T) \le 3000$, we get
\begin{align}
\label{eq:horrible}
    \phi(T, n_i) - \phi(T, b)
    \le \phi(T, n_i)
    \le w(T) \cdot (10tk^{i-1} + k^i)^2 
    \le 3000 \cdot (2k^i)^2 
    = 12000k^{2i}
\end{align}
Comparing with \cref{eq:boring}, we conclude that $\Delta(n_i) > \Delta(b)$ as needed. 

Next, consider the point $m_{i+1}$. We have $\Delta(m_{i+1}) = \phi(m_{i+1}, C_\iota) = 1/t \cdot (10tk^{i+1})^2 = \Omega(k^{2i+2})$.
This term dominates $\Delta(b) = O((tk^i)^2) = O(k^{2i+0.2})$ where we used that all points affected by $b$ have total weight of $O(1)$ and distance from $b$ is at most $10tk^{i}$. So, we get $\Delta(m_i) > \Delta(b)$, as needed. 

Next, consider any point $n_{i'}$ for $i' < i$; we will prove that $\Delta(n_{i'}) < \Delta(c)$. 
We have $\phi(x, n_{i'}) < \phi(x, c)$ only for $x = n_{i'}$ and $x = m_{i'}$. 
Using triangle inequality to bound $d(m_{i'}, c) \le d(m_{i'}, n_{i'}) + d(n_{i'}, b) + d(b, c) = (10t+1)k^j + 1$, we get
\begin{align*}
    \phi(m_{i'}, c) - \phi(m_{i'}, n_{i'})
    \le 1000/t \cdot \left( ((10t+1)k^j + 1)^2 - (10tk^j)^2 \right)
    \le 30000k^{2{i'}}
\end{align*}
(cf. \cref{eq:boring}) and 
\begin{align*}
    \phi(n_{i'}, c) - \phi(n_{i'}, n_{i'})
    = \phi(n_{i'}, c) \le 2000/t \cdot k^{2{i'}}
\end{align*}

We will show that these terms are dominated by $\phi(n_i, n_{i'}) - \phi(n_i, c)$. First, using the cosine law, we have
\begin{align*}
    \phi(n_i, n_{i'}) - \phi(n_i, b)
    = 1000/t \cdot \left(d(n_i, n_{i'})^2 - d(n_i, b)^2 \right)
    \ge 1000/t \cdot 1/t \cdot d(b, n_{i'})d(b, n_i)
    = 1000k^{i + {i'}}/t^2
\end{align*}
and
\begin{align*}
    \phi(n_i, c) - \phi(n_i, b)
    = 1000/t \cdot (2 d(b, c)d(b, n_i)/t + d(b,c)^2)
    \le 1000k^i
\end{align*}
Combining the two bounds, we get 
\begin{align*}
    \phi(n_i, n_{i'}) - \phi(n_i, c)
    \ge 1000k^{i+{i'}}/t^2 - 1000k^i
\end{align*}

Since $i > {i'} \ge 1$, it is certainly true that $1000k^{i + {i'} }/t^2 - 1000k^i > 30000k^{2i'} + 2000k^{2i'}/t$ and we get that $\Delta(n_{i'}) < \Delta(c)$, as needed. 

A very similar argument works for every $m_{i'}$ with $i' < i$ and we omit the proof. 

The only missing point is now $m_i$ for which we prove that $\Delta(m_i) < \Delta(c)$. 
To see that, note that the only point $x$ for which $\phi(x, m_i) < \phi(x, c)$ is $x = m_i$ itself. 
We have
\begin{align*}
    \phi(m_i, b) - \phi(m_i, m_i) 
    = \phi(m_i, b)
    = 1000/t \cdot ((1+10t)k^i)^2
\end{align*}
On the other hand, we have
\begin{align*}
    \phi(b, m_i) - \phi(b, b)
    = \phi(b, m_i)
    = 1 \cdot  ((1+10t)k^i)^2
\end{align*}
That is, the cost drop of $b$ if we take $b$ dominates the cost drop of $m_i$ if we take $m_i$. Whence $\Delta(m_i) < \Delta(b)$, as needed. 
\end{proof}

Consider the first $k^{0.5} - k^{0.4}$ steps of the $i$-th phase. 
We will show that what happens with high probability is that the algorithm select only points of $E_i \cup \{m_{i+1}\}$ as new centers, until at some point it selects $n_i$ or $b$.

\todo{this claim is not proven very formally}
\begin{claim}
\label{cl:one_phase}
Fix $\iota_0$ to be the first step of phase $i$. 
Let $\iota_1$ be the first point in time when either there were at least $k^{0.5} - k^{0.4}$ sampling steps of the $i$-th phase, or until $\{n_i, b\} \cap C_{\iota_1} \not= \emptyset$. 

Then, with probability at least $1 - 1/(10000t)$, we have $C_{\iota_1} \setminus C_{\iota_0} \subseteq E_i \cup \{n_i, b, m_{i+1}\}$ and either $\{n_i, m_{i+1}\} \subseteq C_{\iota_1}$ or $\{b\} \subseteq C_{\iota_1}$.  
Moreover, $b \in C_{\iota_1}$ with probability at least $1/(10^7 t)$. 
\end{claim}
\begin{proof}
We will need to upper bound the probability of various bad events. 
First, we upper bound the probability of the events that at least one point from $E_1 \cup \dots \cup E_{i-1} \cup A$ is sampled as a candidate center. 
To compute the relevant probabilities, first note that at any point in time $\iota \le \iota_1$, we have $|E_i| \ge k^{0.4}$, hence 
\begin{align}
\label{eq:robbing_bernhard}
\phi(E_i, C_\iota) \ge k^{0.4} w_i \ge k^{0.4} \cdot k^{2i+2}    
\end{align}
where we used \cref{eq:wi}. 

On the other hand, we have $\phi(E_1 \cup \dots \cup E_{i-1}, C_\iota) \le t\sqrt{k} w_{i-1} \le 3000t\sqrt{k} k^{2i}$ using \cref{eq:wi} and $\phi(A, C_\iota) \le k^{1.2} \cdot \frac{\log  k}{k^2} \cdot (2k^{i+1})^2$. 
We get 
\[
\frac{\phi(E_1 \cup \dots \cup E_{i-1} \cup A, C_\iota)}{\phi(X, C_\iota)} = O(\frac{k^{2i + 1.3}}{k^{2i + 2.4}}) = O(1/k). 
\]
Hence, during at most $k^{0.5} - k^{0.4} = O(\sqrt{k})$ steps of the phase, the probability of this event is at most $O(\sqrt{k} \cdot \ell/k) = O(1/k^{0.4})$. 

Another bad event is that in some step of the algorithm none of the $\ell > 1$ points sampled is from $E_i$. 
To compute the probability of this event, we upper bound the following cost: 
\begin{align*}
&\phi(\{b,c\} \cup N_{\le i} \cup M_{\le i+1}\cup E_1 \cup \dots \cup E_{i-1} \cup A, C_\iota)\\
&= O(\phi(m_{i+1}, n_{i+1})) && \text{cost dominated by the cost of $m_{i+1}$}\\
&= O((10tk^{i+1})^2 \cdot 1/t )\\
&= O(tk^{2i+2})
\end{align*}
On the other hand, in view of \cref{eq:robbing_bernhard}, we have 
\begin{align*}
    \P(c_{\iota + 1}^1, c_{\iota + 1}^2 \not\in E_i) = \left( \frac{O(tk^{2i+2})}{k^{2i + 2.4}}\right)^2 = O(1/k^{0.7})
\end{align*}
Hence, the probability that this bad event happens in the first $k^{0.5} - k^{0.4}$ steps is at most $O(k^{0.5}/k^{0.7}) = 1/k^{0.2}$. 

A final bad event that we need to deal with is that in one sampling step we sample at least two points from the set $\{b, n_i, m_{i+1}\}$. We argue just about the probability that $n_i, m_{i+1}$ are sampled in one sampling step as this event has the largest probability out of the three pairs $\{b,n_i\}, \{b,m_{i+1}\},\{n_i, m_{i+1}\}$. The probability of this event in one step is at most 
\[
\frac{\ell\phi(n_i, C_\iota)}{\phi(X, C_\iota)}
\cdot \frac{\ell\phi(m_{i+1}, C_\iota)}{\phi(X, C_\iota)}
\le \frac{ \ell \cdot \frac{1000}{t} \cdot (2k^{i+1})^2 \cdot \ell \cdot \frac{1000}{t} \cdot (10tk^{i+1})^2}{\left(\gamma \cdot k^{2i+2}\right)^2}
=O((\ell/k^{0.4})^2) = O(1/k^{0.6})
\]
where we used \cref{eq:robbing_bernhard}. 

Hence, the probability that this bad event happens in the first $k^{0.5} - k^{0.4}$ steps is at most $O(k^{0.5}/k^{0.6}) = 1/k^{0.1}$.  

In view of the computations above, we may assume that in the first $k^{0.5} - k^{0.4}$ steps we always sample one point from $E_i$ and $\ell - 1$ points from $E_i \cup \{b,c\} \cup N_{\le i} \cup M_{\le i+1}$. 
Moreover, we do not sample more than one point from $\{b,n_i,m_{i+1}\}$ in one sampling step. 
We now invoke \cref{cl:cost_drops} and get that only when the non-$E_i$ point we sample is $n_i$ or $m_{i+1}$ or $b$, we add that point to the set of centers. 
Otherwise, we always select the new center as one of the sampled  points from $E_i$. 

We will now prove the claims from the statement. First, we prove that with probability $1 - O(1/t)$ we have either both $n_i$ and $m_{i+1}$ in $C_{\iota_1}$, or we have $b \in C_{\iota_1}$.

In each step $\iota$, there are at least $\gamma = k^{0.5} - (\iota - \iota_0) - 3$ not-yet-taken points of $E_i$, since we already assume that the algorithm chooses a point from  $E_i \cup \{b,n_i,m_{i+1}\}$ as the next center in each step. 
The probability that the next point sampled is $n_i$ is at least 
\begin{align*}
    \P(c_{\iota + 1} = n_i) 
    &\ge \frac{\ell \phi(n_i, C_{\iota})}{\phi(X, C_\iota)}\\
    &\ge \frac{\ell \cdot 1000/t \cdot k^{2i + 2}}{10000k^{2i + 2} + 3000\gamma k^{2i+2}}
    \ge \frac{\ell }{10t \gamma}
\end{align*}
where we bounded $\phi(\{b,c\}\cup A \cup N_{\le i} \cup M_{\le i+1}) \le 10000k^{2i+2}$, used \cref{eq:wi} and used that since $\gamma \ge k^{0.4}$, it  dominates the constant term $10000$. 
Hence, the probability that we do not take $n_i$ nor $b$ in the $k^{0.5} - k^{0.4}$ steps, is at most
\[
\prod_{\gamma = k^{0.5}}^{k^{0.4}+3} \left( 1 - \ell/(10t\gamma) \right)
\le \e^{ - \sum_{\gamma = k^{0.5}}^{k^{0.4}+3}  \ell/(10t\gamma )  }
\le \e^{-(\ell \log k) / 100t}
= \e^{-\log(\ell\log k)} 
< 1/(100000t)
\]
where we used $t = \frac{\ell \log k}{1000\log(\ell\log k)}$ and summed up a harmonic series. Similarly, we can bound that the probability that we take neither $m_{i+1}$ nor $b$ is at most $1/(1000\ell\log k)$. The first part of the claim follows after we also subtract the probabilities of various bad events bounded above from $1 - 1/(100000t)$. 

Note that after $n_i$ is selected as a center during the phase in some step $\iota$, the cost drop resulted by adding a point from $\{b,c\} \cup A \cup N_{<i} \cup M_{\le i}$ to $C_\iota$ is smaller than the cost drop resulted by adding a point from $E_i$ to $C_\iota$. Hence, we have $C_{\iota_1} \setminus C_{\iota} \subseteq E_i \cup \{m_{i+1}\}$ which consequently implies $C_{\iota_1} \setminus C_{\iota_0} \subseteq E_i \cup \{n_i, b, m_{i+1}\}$ as required. 

Finally, we need to prove that we sample $b$ with probability at least $1/(100t)$. To see this, let us condition on $b$ or $n_i$ being one of the sampled and taken points in the first $n^{0.5} - n^{0.4}$ steps of the phase. In every step the ratio of the probability we sample $b$ versus that we sample $n_i$ is equal to 

\begin{align*}
\phi(b, C_\iota) / \phi(n_i, C_\iota) 
&\ge \left( \frac{\log^2(\ell\log k)}{1000\ell^2\log^2 k} \cdot k^{2i+2}  \right) / \left( 1000/t \cdot 2k^{2i+2}  \right)\\
&=  \frac{\frac{\ell\log k}{\log(\ell\log k)} \cdot \log^2(\ell\log k)}{2\cdot 10^6 \cdot \ell^2\log^2 k}  \\
&= \frac{\log(\ell\log k)}{2\cdot 10^6 \cdot\ell \log k}
= 1/(2\cdot 10^6 \cdot t)
\end{align*}
Hence, after we subtract the probabilities of various bad events from $1/(2\cdot 10^6 \cdot t)$, we get that the probability that we sample $b$ during the process is at least $1/(10^7 t)$, as needed. 

\end{proof}

We are now ready to deduce \cref{cl:first_epoch}. 

\begin{proof}
We iterate \cref{cl:one_phase}. As we have only $1/(10000t)$ probability of various failure modes inside one phase, the total probability of failing in some phase is at most $1/10000$. If we condition on no bad events happening in a phase, we get probability of at least $1/(10^7 t)$ of sampling $b$ per phase. Hence, the total probability of sampling $b$ in at least one phase is at least 
\[
1 - (1 - 1/10^7t)^t
\ge 1 - \e^{-1/10^7} > 0
\]
Finally, after $b$ is sampled, we use \cref{cl:finish} to conclude that with positive probability, at the end of the first epoch we have $C_{\tilde{k} - k^{1.2}} = X \setminus (A \cup \{c\})$ as needed. 
\end{proof}

It remains to argue about the second epoch. We need to prove that there is constant probability of sampling $c$ as a candidate center during the first $k^{1.2} - k^{1.1}$ steps. 
Then we need to argue that in that case $c$ is taken as a center by the greedy rule. 
We start by lower bounding the probability of sampling $c$. 

\begin{claim}
\label{cl:one_batch}
    Assume that in some step $\iota_0$ we have $X \setminus C_{\iota_0} \subseteq \{c\} \cup A$ but $c \not\in C_{\iota_0}$. 
    Moreover, assume that $k^{1.1} \le |A|$.  
    Then, with probability at least $1 - \e^{- 1/(6\log k)}$, \cref{alg:kmpp_greedy} samples the point $c$ as a candidate center in the following $|A|/2$ steps. 
\end{claim}
\begin{proof}
    In each of the next $|A|/2$ steps $\iota_0 \le \iota \le \iota_0 + |A|/2$, unless $c \in C_\iota$, we sample $c$ as a fixed candidate center $c_\iota^j$ with probability at least $\frac{\phi(c, b)}{\phi(c, b) + |A|\cdot \phi(a_1, b)}  \ge \frac{1\cdot 1^2}{1\cdot 1^2 + |A|\cdot \frac{\ell \log k}{k^2} \cdot (k+1)^2} \ge \frac{1}{3|A| \ell \log k}$. 
    Hence, the probability that $c$ is not sampled in any of the $|A|/2$ steps as no candidate center $c_\iota^j$ is at most
    \begin{align*}
        \left( 1 - \frac{1}{3|A| \ell\log k} \right)^{\ell \cdot |A|/2}
        \le \e^{- \frac{1}{3|A| \ell\log k} \cdot \ell|A|/2 }
        = \e^{- 1/(6\log k)}
    \end{align*}
    and we are done. 
\end{proof}

We can now finish the proof. 

\begin{claim}
\label{cl:second_phase}
Assume that $C_{\tilde{k} - k^{1.2}} = X \setminus (A \cup \{c\})$. 
Then, with positive probability we have $c \in C_{\tilde{k}}$. 
\end{claim}

\begin{proof}
Consider the first $k^{1.2} - k^{1.1}$ steps after the end of the first epoch. We split these steps to $\log_2 \frac{k^{1.2}}{k^{1.1}} = 0.1 \log_2 k$ batches where the batch $i$ contains $k^{1.2} / 2^i$ steps. 
Fix one such $i$-th batch. We first prove that whenever  the algorithm samples a point $c$ as a candidate center, it also takes it as a center. 
To see that, we need to compute the drop in cost $\Delta(c)$ of $c$ and $\Delta(a_j)$ of any $a_j$. First, we bound the drop of $c$. We use the fact that each point of at least $k^{1.1}$ points $A$ that are not yet taken will have its cost dropped by a small yet non-negligible amount. Namely for the drop in the cost $\Delta(c)$ after taking $c$ as a new center we have
\begin{align}
\Delta(c)\ge
    \sum_{j = 1}^{k^{1.1}} \phi(a_j, b) - \phi(a_j, c) 
    &\ge k^{1.1} \cdot \frac{\ell \log k}{k^{2}} \cdot \left( d(b, a_1)^2 - d(c, a_1)^2 \right)\\
    &\ge \frac{\ell \log k}{k^{0.9}} \cdot k/2 && \text{\cref{eq:a_drop}} \\
    &\ge \ell k^{0.1}
\end{align}
On the other hand, whenever we take some point $a_j$ as a center, the only point whose cost is dropped is $a_j$ itself; this follows from \cref{eq:as_far}. Thus for the drop in the cost $\Delta(a_j)$ after taking $a_j$ as a new center we have
\begin{align}
    \Delta(a_j)
    = \phi(a_j, b)
    &= \frac{\ell \log k}{k^2} \cdot k^2 =\ell \log k
\end{align}

That is, for $k$ large enough we get $\Delta(c) > \Delta(a_j)$ and, hence, whenever $c$ is sampled, it is also taken as a center by the greedy rule of \cref{alg:kmpp_greedy}. 

Finally, we use \cref{cl:one_batch} to conclude that $c$ is not sampled in any of the $0.1\log_2 k$ batches only with probability at most $\left( \e^{- 1/(6\log k)}\right)^{0.1 \log_2 k} \le \e^{-1/100} < 1$.
\end{proof}

\paragraph{Acknowledgment} 
We thank Mohsen Ghaffari, Bill Kuszmaul, Silvio Lattanzi, Aleksander Łukasiewicz, and Melanie Schmidt for useful discussions. 
CG and VR were supported by the European Research Council (ERC) under the European Unions Horizon 2020 research and innovation programme (grant agreement No. 853109). JT is part of BARC, Basic Algorithms Research Copenhagen, supported by the VILLUM Foundation grant 16582.

\bibliographystyle{alpha}
\bibliography{ref}

\appendix

\section{A hard instance for general $k$-means++}
\label{sec:l_adversary_lower_bound}

In this section, we prove the precise version of \cref{thm:adversary_lb_informal}. 
To this end, we first formally describe the general version of $k$-means++ with for an arbitrary seeding rule in \cref{alg:kmpp_adversary}.

\begin{algorithm}[ht]
	\caption{General $k$-means++ seeding}
	\label{alg:kmpp_adversary}
	\label{alg:kmpp_adversarial}
	Input: $X$, $k$, $\ell$, and a rule $\fR$ that picks one point from $\ell$ points $x_1, \dots, x_\ell$, given access to $X, C_i, k, \ell$. 
	\begin{algorithmic}[1]
		\State Uniformly independently sample $x_1, \dots, x_\ell \in X$;
		\State Let $x \in \{x_1, \dots, x_\ell\}$ be selected by $\fR$ and set $C_1 = \{ x \}$.
		\For{$i \leftarrow 1, 2, 3, \dots, k-1$}
		\State Sample $c_{i+1}^1, \dots, c_{i+1}^\ell \in X$ independently w.p. $\frac{\phi(x, C_i)}{\phi(X, C_i)}$;
		\State Let $c_{i+1} \in \{c_{i+1}^1, \dots, c_{i+1}^\ell\}$ be selected by $\fR$ and set $C_{i+1} = C_{i} \cup \{c_{i+1}\}$.
		\EndFor
		\State \Return {$C := C_k$}
	\end{algorithmic}
\end{algorithm}

The rest of this section is then devoted to the proof of the following theorem. 

\begin{restatable}{theorem}{adversaryLb}
\label{thm:adversary_lb}
There exists a point set $X  \subseteq \mathbb{R}^d$ and a rule $\fR$ such that \cref{alg:kmpp_adversary} with $\fR$ is $\Omega(k^{1 - 1/\ell})$-approximate with constant probability. 
\end{restatable}

We already sketched a proof of this theorem in \cref{subsec:strike} for $\ell = \Omega(\log k)$. The generalization for any $\ell$ and to a Euclidean space below is routine. 

We remark that we believe one can improve the lower bound from \cref{thm:adversary_lb} to $\Omega(k^{1- 1/\ell} \ell\log k)$ by adding the set $A$ as in the proof of \cref{thm:greedy_lb}. 

We begin by describing the input instance $(X, k, C_0)$ (recall that in view of \cref{lem:prescribed} we can assume we start with a non-empty set of centers $C_0$).  Throughout the proof, we will assume that the weights are integers. This is for readability, we could round the numbers to the closest integer and that would not hurt any asymptotic guarantees. Our instance $X$ is a subset of $k+1$-dimensional Euclidean space $R^{k+1}$. 

We next describe the input weighted point set $X$. 
\begin{enumerate}
    \item There is a point $d \in C_0$ in the origin. 
    \item There are $k-1$ points $x_1, x_2,\dots , x_{k-1}$ having weight $w(x_i) = 1$. Moreover each $x_i$ has the coordinate $(0,\dots,0, k, 0,\dots, 0)$, which has value $0$ at each dimension except the $i$-th. Hence, $d(d,x_i) = k$ for every $i\in \{1,2,\dots, k-1\}$.
    \item There is a point $c$ with weight $w(c) = \frac{k^{1-1/\ell}}{2}$ at $(0, 0, \dots, k, 0)$. Hence $d(d, c) = k$. 
    \item The final point $b$ has weight $w(b) = 1$ and lies in the plane generated by vectors $(0, \dots, 1, 0)$ and $(0, \dots, 0, 1)$ in such a position that it holds that $d(c,d) = d(b,d) = k$ and $d(b,c) = 1$ (i.e., $d, c, b$ form an isosceles triangle).

\end{enumerate}

In view of \cref{lem:wlog_opt_subset}, we require the optimal solution $C^* \subseteq X$. Then, we have $C^* = \{x_1,x_2,\dots,x_{k-1},c\}$ and the cost of it is $OPT = \phi(b, C^* \cup C_0) = d(c,b)^2\cdot w(b) = 1$.

We are going to pick the rule $\fR$ as follows: whenever we sample $b$ as a candidate, we take it as a center. Furthermore, we only take c as a center if all of the $\ell$ candidate points are $c$, hence when we have no other choice.

We will show that with constant probability we will take $b$ as a center after $k/2$ steps. That means that at least one of the points in $\{x_1,x_2,\dots,x_{k-1},c\}$ will not be selected as a center at the end.
If one of the $x_i$ is not selected as a center, then the cost of the solution will be at least $w(x_i) \cdot d(x_i,d)^2 = k^2$, since $d$ is the closest point to any $x_i$. 
If $c$ is not selected as a center, then the cost of the solution will be at least $ w(c) \cdot d(c,b)^2 =  \frac{k^{1-1/\ell}}{2} $, since the closest chosen center to $c$ is $b$. Hence if we pick $b$ as a center, the approximation factor will be at least $\frac{k^{1-1/\ell}}{2}$.

Let $B_i$ and $C_i$ be the events that $b$ and $c$ are chosen as a center at the $i$-th step, respectively and $B_{\leq i}$ and $C_{\leq i}$ be the events that $b$ and $c$ are chosen as a center in one of the first $i$ steps.

We will calculate the probability of picking $b$ as a center in the first $k/2$ steps as follows:
\begin{align*}
    \P(B_{\leq k/2}) &= \P(B_1) + \P(B_2\; |\;\neg B_{\leq 1})\cdot \P(\neg B_{\leq 1}) + \dots +  \P(B_{k/2} \; |\;\neg B_{\leq k/2-1}) \cdot \P(\neg B_{\leq k/2-1})\\
\end{align*}

To calculate a lower bound for $\P(B_{\leq k/2})$, we will need a lower bound for $\P(B_i\; |\;\neg B_{\leq i-1})$ and $\P(\neg B_{\leq i})$.
We start by showing that with constant probability we do not pick $b$ in the first $k/2$ steps.

\begin{lemma}
\label{lem:appndx_a_low_bound_not_picking_b}
For any $i\leq k/2$ we have $\P(\neg B_{\leq i}) \geq \frac{1}{10}$.
\end{lemma}
\begin{proof}
First we show that for any $i\leq \frac{k}{2}$ we have $\P(\neg B_i | \neg B_{\leq i-1}) \geq 1 - \frac{2}{k+2}$. We have
\begin{align*}
    \P(\neg B_{i} \; |\; \neg B_{\leq i-1} \land \neg C_{\leq i-1}) &= \frac{k-i + k^{1-1/\ell}/2}{k-i + k^{1-1/\ell}/2 + 1} \\
    &\geq  \frac{k/2 + k^{1-1/\ell}/2}{k/2 + k^{1-1/\ell}/2 + 1}\\ 
    &\geq \frac{k/2}{k/2 + 1} = 1 - \frac{k/2}{k/2 + 1} = 1 - \frac{2}{k + 2}  
\end{align*}

\begin{align*}
    \P(\neg B_{i} \; |\; \neg B_{\leq i-1} \land C_{\leq i-1}) &= \frac{(k-i+1)\cdot k^2}{1 + (k-i+1)\cdot k^2}\\
    &\geq  \frac{k^3/2}{1 + k^3/2} = 1 - \frac{2}{2 + k^3} \geq 1 - \frac{2}{k + 2}\\
\end{align*}

Hence $\P(\neg B_{i} \; |\; \neg B_{\leq i-1}) \geq 1 - \frac{2}{k + 2}$. Now we can calculate the probability of not picking $b$ in the first $i$ steps.
\begin{align*}
    \P(\neg B_{\leq i}) &= \P(\neg B_1 \cap \neg B_2 \cap \neg B_3 \cap \dots \cap \neg B_i)\\
    &= \P(\neg B_1\; |\; \neg B_{\leq 0}) \cdot \P(\neg B_2 \; |\; \neg B_{\leq1}) \cdot \P(\neg B_3 \; |\; \neg B_{\leq 2})\cdots \P(\neg B_i \; |\; \neg B_{\leq i-1})\\
    &\geq \left(1 - \frac{2}{k + 2}\right)^{i}\\
    &\geq \left(1 - \frac{2}{k + 2}\right)^{k/2}\\
    &\geq \left(e^{-\frac{4}{k+2}}\right)^{k/2} &&\text{\cref{fact:one_minus_x_lowerbound}}\\
    &\geq \left(e^{-\frac{2 k}{k+2}}\right) \geq \frac{1}{10}
\end{align*}

\end{proof}

Now, to calculate a lower bound for $\P(B_i\; |\;\neg B_{\leq i-1})$ we split it as follows:

\begin{align*}
    \P(B_i\; |\;\neg B_{\leq i-1}) &= \P(B_i\; |\;\neg B_{\leq i-1} \land \neg C_{\leq i-1}) \cdot \P(\neg C_{\leq i-1}) + \P(B_i\; |\;\neg B_{\leq i-1} \land C_{\leq i-1}) \cdot \P( C_{\leq i-1})\\
    &\geq \P(B_i\; |\;\neg B_{\leq i-1} \land \neg C_{\leq i-1}) \cdot \P(\neg C_{\leq i-1})
\end{align*}

First, we will show that the probability of picking $c$ in a step is at most $1/k$. 

\begin{lemma}
\label{lem:appndx_A_upper_bound_picking_c}
For every $i\leq k/2$ we have $\P(C_i\; |\;\neg C_{\leq i-1}) \leq 1/k $ . 
\end{lemma}

\begin{proof}
We will start by showing that the probability of $c$ being selected as a center is higher when $b$ is not selected as a center. This intuitively makes sense because $b$ lies very close to $c$.

\begin{align*}
    \P(C_i\; |\;\neg C_{\leq i-1} \land B_{\leq i-1}) &= \frac{k^{1-1/\ell}}{k^{1-1/\ell} + (k-i+1)\cdot k^2}\\
    &\leq \frac{k^{1-1/\ell} + (k^{3-1/\ell} - k^{1-1/\ell} )}{k^{1-1/\ell} + (k-i+1)\cdot k^2+ (k^{3-1/\ell} - k^{1-1/\ell} )}\\
    &\leq \frac{k^{3-1/\ell}}{k^{3-1/\ell} + (k-i+1)\cdot k^2} =  \P(C_i\; |\;\neg C_{\leq i-1} \land \neg B_{\leq i-1})
\end{align*}

Now using $ \P(C_i\; |\;\neg C_{\leq i-1} \land B_{\leq i-1}) \leq \P(C_i\; |\;\neg C_{\leq i-1} \land \neg B_{\leq i-1}) $, we can bound $\P(C_i\; |\;\neg C_{\leq i-1})$.

\begin{align*}
        \P(C_i\; |\;\neg C_{\leq i-1}) &= \P(B_{\leq i-1}) \cdot \P(C_i\; |\;\neg C_{\leq i-1} \land B_{\leq i-1}) + \P(\neg B_{\leq i-1}) \cdot \P(C_i\; |\;\neg C_{\leq i-1} \land \neg B_{\leq i-1})\\
        &\leq \P(B_{\leq i-1}) \cdot \P(C_i\; |\;\neg C_{\leq i-1} \land \neg B_{\leq i-1}) + \P(\neg B_{\leq i-1}) \cdot \P(C_i\; |\;\neg C_{\leq i-1} \land \neg B_{\leq i-1})\\
        &\leq \P(C_i\; |\;\neg C_{\leq i-1} \land \neg B_{\leq i-1})
\end{align*}

According to our rule $\fR$, for $c$ to be selected as a center, all of the $\ell$ samples in a step should be $c$.
\begin{align*}
    \P(C_{i} \; |\; \neg C_{\leq i-1} \land \neg B_{\leq i-1}) &= \left(\frac{k^{3-1/\ell}/2}{k^{3-1/\ell}/2 + (k-i+1)\cdot k^2}\right)^\ell\\
    &\leq  \left(\frac{k^{3-1/\ell}/2}{k^{3-1/\ell}/2 + k/2\cdot k^2}\right)^\ell\\
    &\leq  \left(\frac{k^{3 - 1/\ell}/2}{k^{3}/2}\right)^\ell\\
    &\leq \left(\frac{1}{k^{1/\ell}}\right)^\ell = \frac{1}{k}
\end{align*}
Hence $\P(C_i\; |\;\neg C_{\leq i-1}) \leq \P(C_i\; |\;\neg C_{\leq i-1} \land \neg B_{\leq i-1}) \leq \frac{1}{k}$
\end{proof}

Now we can show that the probability of not picking $c$ as a center in the first $k/2$ steps is at least $1/2$.
\begin{lemma}
\label{lem:appndx_A_not_picking_C}
    For $i\leq k/2$, $\P(\neg C_{\leq i}) \geq \frac{1}{2}$
\end{lemma}
\begin{proof}

\begin{align*}
     \P(C_{\leq i}) &= \P(C_1) + \P(C_2\; |\;\neg C_{\leq 1}) \cdot \P(\neg C_{\leq 1}) + \dots +  \P(C_{i} \; |\;\neg C_{\leq i-1}) \cdot \P(\neg C_{\leq i-1})\\
        &\leq \P(C_1\; |\;\neg C_{\leq 0}) + \P(C_2\; |\;\neg C_{\leq 1}) + \P(C_3\; |\;\neg C_{\leq 2}) + \dots +  \P(C_{i} \; |\; \neg C_{i-1} )\\
        &\leq \frac{1}{k} + \frac{1}{k} + \frac{1}{k} + \dots + \frac{1}{k} && \text{\cref{lem:appndx_A_upper_bound_picking_c}} \\
        &\leq \frac{i}{k} \leq \frac{1}{2}
\end{align*}
 Hence,  $ \P(\neg C_{\leq i}) = 1 - \P(C_{\leq i}) \geq \frac{1}{2}$
\end{proof}

Using \cref{lem:appndx_A_not_picking_C} we can finally calculate a lower bound for $\P(B_i\; |\;\neg B_{\leq i-1})$.
\begin{lemma}
\label{lem:appndx_A_low_bound_picking_B}
     For $i\leq k/2$, $\P(B_i\; |\;\neg B_{\leq i-1}) \geq \frac{1}{4k}$ 
\end{lemma}
\begin{proof}
\begin{align*}
    \P(B_i\; |\;\neg B_{\leq i-1}) &= \P(B_i\; |\;\neg B_{\leq i-1} \land \neg C_{\leq i-1}) \cdot \P(\neg C_{\leq i-1}) + \P(B_i\; |\;\neg B_{\leq i-1} \land C_{\leq i-1}) \cdot \P( C_{\leq i-1})\\
    &\geq \P(B_i\; |\;\neg B_{\leq i-1} \land \neg C_{\leq i-1}) \cdot \P(\neg C_{\leq i-1}) \\
    &\geq \P(B_i\; |\;\neg B_{\leq i-1} \land \neg C_{\leq i-1}) \cdot \frac{1}{2} && \text{\cref{lem:appndx_A_not_picking_C}}\\
    &\geq \frac{1}{(k-i) + k^{1-1/\ell}/2 + 1} \cdot \frac{1}{2} \\
    &\geq \frac{1}{2k} \cdot \frac{1}{2} = \frac{1}{4k}
\end{align*}

\end{proof}

Now we are ready to prove the theorem.
\begin{align*}
    \P(B_{\leq k/2}) &= \P(B_1) + \P(B_2\; |\;\neg B_{\leq 1})\cdot \P(\neg B_{\leq 1}) \dots +  \P(B_{k/2} \; |\;\neg B_{\leq k/2-1}) \cdot \P(\neg B_{\leq k/2-1})\\
    &\geq \frac{1}{10}\cdot \left( \P(B_1\; |\;\neg B_{\leq 0}) + \P(B_2\; |\;\neg B_{\leq 1}) + \P(B_3\; |\;\neg B_{\leq 2}) + \dots +  \P(B_{k/2} \; |\;\neg B_{\leq k/2-1} )\right) && \text{\cref{lem:appndx_a_low_bound_not_picking_b}}\\
    &\geq \frac{1}{10}\cdot( \frac{1}{4k} + \frac{1}{4k} + \frac{1}{4k} + \dots +  \frac{1}{4k}) && \text{\cref{lem:appndx_A_low_bound_picking_B}}\\
    &\geq \frac{1}{10}\cdot( \frac{k/2}{4k})\\
    &\geq \frac{1}{80}
\end{align*}

Hence, with constant probability, $b$ will be taken as a center in the first $k/2$ steps, and as a result, the algorithm will return a solution with an approximation ratio of at least $\Omega(k^{1-1/\ell})$.

\section{Analysis of general $k$-means++}
\label{sec:l_adversary_upper_bound}

In this section, we prove the precise version of \cref{thm:adversary_ub_informal} that we state next. 

\begin{theorem}
\label{thm:adversary_ub}
For any rule $\fR$, \cref{alg:kmpp_adversarial} is $O(  k^{2-1/\ell} \cdot \ell\log k)$-approximate. 
\end{theorem}

\subsection{Hitting optimal clusters}
We start by proving an analogue to the \cref{lem:greedy_lemma} that shows that any optimal cluster is expected to be hit $O(\ell k^{1-1/\ell})$ times. Note that it is trivially hit $O(\ell k)$ times. The reason we bother proving this only slightly better (and tight) result is that we wrote the proof before we realized that our lower and upper bounds for \cref{alg:kmpp_adversary} are not matching since we could not analyze the sampling process defined below. 

The improvement over the trivial $O(\ell k)$ bound is based on the fact that if a cluster $K$ dominates the cost of the whole point set, we have a nontrivial probability of sampling all $\ell$ candidate centers from it. 

\begin{lemma}
\label{lem:adversary_lemma}
For any rule $\fR$ in \cref{alg:kmpp_adversary} and any optimal cluster $K$ we have that $\E[\hit(K)] = O(\ell \cdot k^{1 - 1/\ell})$. 
\end{lemma}

\begin{proof}
We prove this statement by induction. Recall that $\hit(K)$ be the number of points of $K$ that we sample from $K$ until $K$ becomes covered (\cref{def:covered}) or solved (\cref{def:solved}) but for the purposes of this proof we even drop the ``solved'' requirement. 

We prove that 
\begin{align*}
    \E_{\ge k-i}[\hit_{\ge k-i}(K)] \le 10\ell i^{1 - 1/\ell}
\end{align*}

For $i = 0$ it clearly holds.
Next, assume the equation holds for $k-i+1$ and we prove it for $k-i$. 

Let us define $p = \frac{\phi(K, C_i)}{\phi(X, C_i)}$. 
We will now use the fact that whenever we sample all points $c_i^1, \dots, c_i^\ell$ from $K$, i.e., whenever $\hit_{i}(K) = \ell$, we have $c_i \in K$ and $K$ hence becomes covered. Namely, using induction hypothesis we compute: 
\begin{align}
\label{eq:charlie}
    \E_{\ge k-i}[ \hit_{\ge k-i}(K)] 
    &= \E_{k-i}[\hit_i(K) ] + \E_{\ge k- i}[ \hit_{\ge k-i+1}(K))\\
    &\le  \ell p + \P(\hit_i(K) = \ell)\cdot 0 + \P(\hit_i(K) \not= \ell) 10\ell (i-1)^{1 - 1/\ell}\\
    &= \ell p + (1 - p^{\ell}) 10\ell (i-1)^{1-1/\ell}
\end{align}

Next, we compute
\begin{align*}
    (i - 1)^{1 - 1/\ell} 
    &= i^{1 - 1/\ell} \cdot (1 - 1/i)^{1 - 1/\ell}\\
    &\le i^{1 - 1/\ell} \cdot (1 - 1/i)^{1/2} && \ell \ge 2\\
    &\le i^{1 - 1/\ell} \cdot (1 - 1/(4i)) \\
\end{align*}

Let us use $f(i) = 10 \ell i^{1 - 1/\ell} \cdot \left(1 - 1/(4i)\right)$. Then, the above computation in \cref{eq:charlie} says that 
\begin{equation}
\label{eq:distracted}
    \E_{\ge k-i}[ \hit_{\ge k-i}(K)] 
    \le \ell p + (1 - p^\ell)f(i)
\end{equation}

Let us analyze the right-hand side of that expression. 
We have

\[
\frac{\delta \left( \ell p + (1 - p^\ell)f(i) \right)}{\delta p}
= \ell - \ell p^{\ell - 1} f(i)
\]
Solving for the right hand side equal to zero, we get $1 - p^{\ell - 1}f(i) = 0$, hence $p = \left( 1 / f(i)\right)^{1 / (\ell - 1)}$. 
That is, the right hand side of \cref{eq:distracted} is maximized for that $p$ and then it is equal to 
\begin{align}
\label{eq:so_slow}
&\ell\left( 1 / f(i)\right)^{1 / (\ell - 1)}
+ \left(1 - \left(  \left( 1 / f(i)\right)^{1 / (\ell - 1)} \right)^{\ell}\right)f(i)
\end{align}

Let us plug in the definition of $f(i)$ to that expression. We start with the first term. We have
\begin{align*}
\ell\left( 1 / f(i)\right)^{1 / (\ell - 1)}
&= \ell \left( \frac{1}{10 \ell i^{1 - 1/\ell} \cdot \left(1 - 1/(4i)\right)} \right)^{1 / (\ell - 1)}
\le  \frac{\ell}{i^{1/\ell}} 
\end{align*}
where we used that $10\ell(1 - 1/(4i)) \ge 1$. 

Next, we handle the second term as follows:
\begin{align*}
\left(1 - \left(  \left( 1 / f(i)\right)^{1 / (\ell - 1)} \right)^{\ell}\right)f(i)
\le f(i)
\le 10 \ell i^{1 - 1/\ell} \cdot \left(1 -  1/(4i)\right)
\end{align*}
Hence, we can upper bound the expression in \cref{eq:so_slow} by
\begin{align*}
    &10 \ell i^{1 - 1/\ell} \cdot \left(1 - 1/(4i)\right) + \frac{\ell}{i^{1/\ell}}\\
    &=10 \ell i^{1 - 1/\ell} \cdot \left(1 -  1/(4i) + \frac{1}{10 i} \right) \\
    &\le 10\ell i^{1 - 1/\ell}
\end{align*}
and we are done. 

\end{proof}

As a corollary we get \cref{thm:adversary_ub}.
\begin{proof}[Proof Sketch]
The proof is very similar to the proof of \cref{thm:greedy_ub}. In that proof, we are using the greedy rule in two places:
\begin{enumerate}
    \item Through \cref{lem:greedy_lemma}; instead of that lemma we now use \cref{lem:adversary_lemma}. 
    \item Inside \cref{prop:monotone3} we use it to bound how much can the size of the average uncovered cluster increases during the algorithm. We can very crudely bound this multiplicative increase by $k$, hence our approximation guarantee picks up additional $k$-factor. 
\end{enumerate}

\end{proof}

\subsection{An interesting sampling process}
\label{subsec:lprocess}

This section is devoted to the explanation of an interesting open problem that, if solved, probably brings together the upper and lower bounds for \cref{alg:kmpp_adversary} that are now off by a factor of $k$. 
We note that losing a factor of $k$ because of the drift of the size of the average uncovered cluster looks very wasteful. In fact, we can replace this factor of $k$ in the upper bound by a function $g(k, \ell)$ that we discuss in the rest of this section. We believe that understanding $g(k, \ell)$ is an exciting open problem. 
The problem in the analysis of \cref{thm:adversary_ub} can be distilled into the following riddle, which one can understand without understanding the details of the analysis of Arthur and Vassilvitskii. 

\begin{definition}[$\ell$-point adversarial sampling process]
\label{def:process_l}
Let $\ell \in \mathbb{N}$. 
We define the \emph{$\ell$-point adversarial sampling process} as follows. 
At the beginning, there is a set $E_0$ of $k$ elements where each element $e \in E_0$ has some nonnegative weight $w_0(e)$. 
The process has $k$ rounds: in each round, we form the new set $E_{i+1}$ from $E_i$ as follows:
\begin{enumerate}
    \item We define the distribution $D_i$ over $E_i$ where the probability of $e$ is defined as $w_i(e) / \sum_{e \in E_i} w_i(e)$. 
    An adversary chooses an arbitrary number $\ell_i$ that satisfies $0 \le \ell_i \le \ell$. 
    We sample $\ell_i$ points $e_i^1, \dots, e_i^{\ell_i}$ independently from $D_i$. 
    Next, an adversary chooses a point $e_i \in \{e_i^1, \dots, e_i^{\ell_i}\}$. 
    We set $E_{i+1} = E_i \setminus \{e_i\}$. 
    \item An adversary chooses the new weight function $w_{i+1}(e)$ for every element $e \in E_{i+1}$ as an arbitrary function that satisfies
    \[
    0 \le w_{i+1}(e) \le w_i(e). 
    \]
\end{enumerate}
\end{definition}

The relationship between this process and the analysis of \cref{alg:kmpp_adversary} is as follows. The set $E_i$ of elements corresponds to a set of uncovered clusters. The steps where we sample $\ell_i \le \ell$ elements $e_i^1, \dots, e_i^{\ell_i}$ corresponds to the algorithm sampling at most $\ell$ centers from uncovered clusters. 
We assume that the rule $\fR$ behaves adversarially and can decide  to cover any of the sampled clusters, i.e., we allow removing any sampled element $e_i \in \{e_i^1, \dots, e_i^{\ell_i}\}$ from $E_i$ to form $E_{i+1}$. 
The adversarial decreasing of the element weights in between two sampling steps corresponds to the newly taken center decreasing the cost of the optimal clusters in an uncontrolled manner. 

We note that the analysis of $k$-means++ from \cite{arthur2007k} implicitly analyzes this game for $\ell = 1$. This case is qualitatively simpler than the general case: We can in fact even prove that for $\ell = 1$, the average element size can only decrease between two steps. To see this, we observe that it would stay the same if we picked each element uniformly at random. Picking heavier elements with higher probability can only decrease the average size then. 

However, this simple approach does not work anymore for $\ell > 1$. In fact, consider as an example the set $E_0$ consisting of $k-1$ elements of size one and one element of size $k$. Choosing $\ell =\Omega(k)$, the adversary can prevent us to remove the costly element until the very end, with constant probability. 
This increases the average element size from roughly $2$ to $k$, that is, by $\Omega(k) = \Omega(\ell)$ factor. 

We do not know how much the average can increase but it is clearly at most by a $O(k)$ factor and by the above reasoning it is at least by $\Omega(\ell)$ factor. 

\begin{fact}
\label{fact:process_l}
We define the function $g(k, \ell)$ as the smallest growing function satisfying the following condition. 
Let $\avg_i$ be defined as the average weight of an element in the $i$-th round of the $\ell$-point adversarial sampling process from \cref{def:process_l}, i.e., for any $0 \le i < k$ we define
\[
\avg_i = \frac{\sum_{e \in E_i} w_i(e)}{k-i}. 
\]
Then, for any adversary and any $0 \le i < k$, we have
\[
\avg_i \le g(k, \ell)\avg_0. 
\]
\end{fact}

\begin{question}
What is the value of $g(k, \ell)$? 
\end{question}

A similar problem to our sampling process was recently considered in \cite{bhattacharya2019noisy} where the authors consider the following related problem. 
Suppose that we run $k$-means++, but before each sampling step, an adversary distorts each probability of taking an element by a multiplicative $1\pm \eps$ factor. Does such an algorithm retain $O(\log k)$ approximation guarantees for fixed $\eps < 1$? 
This question leads to the analysis of a process very similar to \cref{def:process_l}; instead of choosing one of $\ell$ sampled elements, the power of the adversary is now to distort the sampling distribution pointwise by $1\pm\eps$ multiplicative factor. 
In \cite{bhattacharya2019noisy}, the authors show that the average in this game can increase only by a multiplicative $O(\log k)$ factor. This follows from the fact that all elements larger than $\Omega(\log k)$ will be already taken in the first $k/2$ steps of the process. This implies $O(\log^2 k)$ approximation guarantee for the final algorithm. 
In \cite{Grunau_Ozudogru_Rozhon2022noisy} this analysis of the game is improved to $O(1)$ which implies the tight $O(\log k)$ upper bound for $k$-means++ with noise.

\section{An incomparable bound on the number of hits}
\label{sec:alternative}

In this section, we sketch the proof of the following result which is incomparable with \cref{lem:greedy_lemma} but substantially easier to prove. In fact, we used this proof sketch as a way to build intuition towards the proof of \cref{lem:greedy_lemma} in an earlier draft of this writeup, before we realized \cref{sec:intro,sec:intuitive} are way too long. 

\begin{lemma}
\label{lem:greedy_lem_other}
For any optimal cluster $K$ we have $\E[\hit(K)] =  O(\ell \log \frac{OPT(1)}{OPT(K)})$. 
\end{lemma}

Here, $OPT(\tilde{k})$ is the size of the optimal solution with $\tilde{k}$ centers. 

\begin{proof}[Proof sketch]
Fix an optimal cluster $K$, consider a step $i+1$ and assume that for the cost of $K$ we have $\phi(K, C_i) \ge 10^5\phi^*(K)$. 
Being far from the optimum means that all centers of $C_i$ are very far from most of the points of $K$. Hence, whenever it happens that a potential center $c_{i+1}^j$ for some $1 \le j \le \ell$ is sampled from $K$, we have constant probability that $d(\mu(K), c_{i+1}^j) \le d(\mu(K), C_i)/2$, i.e., $c_{i+1}^j$ is substantially closer to $\mu(K)$ than all other centers in $C_i$. 
In that case, we have $\phi(X, C_i) - \phi(X, C_i \cup \{c_{i+1}^j\}) \ge \phi(K, C_i)/2$. 
That is, adding $c_{i+1}^j$ as the new center will result in the cost drop of at least $\phi(K, C_i)/2$. 

We will now need to distinguish two cases. Let us consider the distribution over the cost drop $\phi(X, C_i) - \phi(X, C_{i} \cup \{c\})$ where $c$ is sampled proportional to its current cost $\phi(c, C_i)$. That is, we consider the distribution of how the cost drops if we add the candidate center $c_{i+1}^1$ (or any other fixed candidate center) to the current solution. 
In the first, \emph{easy}, case we assume that with probability $1 - 1/\ell$, the cost drop $\phi(X, C_i) - \phi(X, C_{i} \cup \{c\})$ is less than $\phi(K, C_i)/2$; otherwise we are in the \emph{hard} case. 

What is easy in the easy case? The discussion above implies that in that case, whenever we sample some $c_{i+1}^j$ from $K$, we have constant probability that all other candidate centers create a cost drop smaller than $\phi(K, C_i)/2$, hence the greedy heuristic chooses $c_{i+1} = c_{i+1}^j$. 
Hence, sampling a point from $K$ in the easy case can happen only constantly many times, in expectation, before $K$ becomes covered after which we stop counting the hits to $K$. 

But what do we do in the hard case? There, we at least know that with constant probability the cost drops by $\phi(K, C_i)/2$, concretely we know that $\phi(K, C_i) - \E[\phi(K, C_{i+1})] \ge \phi(K, C_i)/5$.
Recall that we are counting hits to $K$ and each candidate center hits it with probability $\phi(K, C_i)/\phi(X, C_i)$. 

Our situation is very similar to the following deterministic process where we start with a number $X_0$ (corresponding to $\phi(X, C_1)$) and an empty counter $H_0 = 0$ (corresponding to counting hits). In each step we then choose some number $0 < K_i \le X_i$ (corresponding to the cost of the cluster $\phi(K, C_i)$) and define $X_{i+1} \leftarrow X_i - K_i$, while increasing the counter $C_{i+1} \leftarrow C_i + \ell \cdot K_i/X_i$. 
In this idealized process, it holds that  $C_i = O(\ell \cdot \log X_0/X_i)$. Intuitively, this is because the case when $K_i = \Theta(X_i)$ in every step is the hardest one. 

We can apply similar reasoning to our randomized process to get the expected bound $O\left(\ell  \log \frac{OPT(1)}{OPT(k)}\right)$ on the number of hits. 
Here, we additionally use that 1) the starting cost of our solution $\phi(X, C_1)$ is expected to be of order $OPT(1)$ by \cref{lem:5apx_intuitive} and 2) 
the final cost $\phi(X, C_k)$ has to be at least $OPT(k)$. 
\end{proof}

\end{document}